\pdfoutput=1

\documentclass[final,3p
,times
]{elsarticle}

\usepackage[T1]{fontenc}
\usepackage[utf8]{inputenc}
\usepackage{xcolor}
\definecolor[named]{ACMBlue}{cmyk}{1,0.1,0,0.1}
\definecolor[named]{ACMYellow}{cmyk}{0,0.16,1,0}
\definecolor[named]{ACMOrange}{cmyk}{0,0.42,1,0.01}
\definecolor[named]{ACMRed}{cmyk}{0,0.90,0.86,0}
\definecolor[named]{ACMLightBlue}{cmyk}{0.49,0.01,0,0}
\definecolor[named]{ACMGreen}{cmyk}{0.20,0,1,0.19}
\definecolor[named]{ACMPurple}{cmyk}{0.55,1,0,0.15}
\definecolor[named]{ACMDarkBlue}{cmyk}{1,0.58,0,0.21}
\usepackage[
           ]{hyperref}
\newif\ifelsevierstyle
\elsevierstyletrue
\newif\ifgenericfooter
\genericfootertrue
\newif\ifwithcomments
\withcommentstrue
\withcommentsfalse
\newif\ifsubmission
\submissiontrue
\submissionfalse
\newif\iflongversion
\longversiontrue
\usepackage{setup}

\usepackage{booktabs}   
\usepackage{subcaption} 

\usepackage{url}

\usepackage{amssymb}
\usepackage{amsthm}
\newtheorem {definition}  {Definition} [section]
\newtheorem {theorem}     [definition] {Theorem}

\newtheorem {lemma}       [definition] {Lemma}
\newtheorem {corollary}   [definition] {Corollary}

\journal{}

\hypersetup{
      colorlinks,
      linkcolor=ACMRed,
      citecolor=ACMPurple,
      urlcolor=ACMDarkBlue,
      filecolor=ACMDarkBlue
}
\begin{document}

\begin{frontmatter}



\title{Revisiting Occurrence Typing}


\author[1]{Giuseppe Castagna}
\address[1]{Institut de Recherche en Informatique Fondamentale (IRIF),
               CNRS - Université de Paris, France
}

\author[1]{Victor Lanvin}

\author[1]{Mickaël Laurent}

\author[2]{Kim Nguyen}
\address[2]{
  {Laboratoire de Méthodes Formelles (LMF),}
  {CNRS - Université Paris-Saclay,}
  {France}
}

\begin{abstract}
We revisit occurrence typing, a technique to refine the type of
variables occurring in type-cases and, thus, capture some programming
patterns used in untyped languages. Although occurrence typing was
tied from its inception to set-theoretic types---union types, in
particular---it never fully exploited the capabilities of these
types. Here we show how, by using set-theoretic types, it is possible
to develop a general typing framework that encompasses and generalizes
several aspects of current occurrence typing proposals and that can be
applied to tackle other problems such as the reconstruction of intersection
types for unannotated or partially annotated functions and the optimization of
the compilation of gradually typed languages.

\end{abstract}

\begin{keyword}
occurrence typing \sep
type inference  \sep
union types  \sep
intersection types  \sep
TypeScript  \sep
Flow language  \sep
dynamic languages  \sep
type case \sep
gradual typing.
\end{keyword}
\end{frontmatter}

\section{Introduction}
\label{sec:intro}
TypeScript and Flow are extensions of JavaScript that allow the programmer to specify in the code type annotations used to statically type-check the program. For instance, the following function definition is valid in both languages  
\begin{alltt}\color{darkblue}
  function foo(x\textcolor{darkred}{ : number | string}) \{
      return (typeof(x) === "number")? x+1 : x.trim();  \refstepcounter{equation}                               \mbox{\color{black}\rm(\theequation)}\label{foo}
  \}
\end{alltt}
Apart from the type annotation (in red) of the function parameter, the above is
standard JavaScript code defining a function that checks whether
its argument is an integer; if it is so, then it returns the argument's successor
(\code{x+1}), otherwise it calls the method \code{trim()} of the
argument. The annotation specifies that the parameter is either a
number or a string (the vertical bar denotes a union type).  If this annotation is respected and the function
is applied to either an integer or a string, then the application
cannot fail because of a type error (\code{trim()} is a string method of the ECMAScript 5 standard that
trims white-spaces from the beginning and end of the string) and  both the type-checker of TypeScript and the one of Flow
rightly accept this function. This is possible because both type-checkers
implement a specific type discipline called \emph{occurrence typing} or \emph{flow
typing}:\footnote{%
  TypeScript calls it ``type guard recognition'' while Flow uses the terminology ``type
  refinements''.}
as a matter of fact, standard type disciplines would reject this function.
The reason for that is that standard type disciplines would try to
type every part of the body of the function under the assumption that \code{x} has
type \code{number\,|\,string} and they would fail, since the successor is
not defined for strings and the method \code{trim()} is not defined
for numbers.  This is so because standard disciplines do not take into account the type
test performed on \code{x}. Occurrence typing is the typing technique
that uses the information provided by the test to specialize---precisely, to \emph{refine}---the type
of the occurrences of \code{x} in the branches of the conditional: since the program tested that
\code{x} is of type \code{number}, then we can safely assume that
\code{x} is of type \code{number} in the ``then'' branch, and that it
is \emph{not} of type \code{number} (and thus deduce from the type annotation that it must be of type
\code{string}) in the ``else'' branch.

Occurrence typing was first defined and formally studied by
\citet{THF08} to statically type-check untyped Scheme
programs,\footnote{%
According to Sam Tobin-Hochstadt, the terminology \emph{occurrence
typing} was first used in a simplistic form by~\citet{Komon05},
although he and Felleisen were not aware of it the at the moment of
the writing of~\cite{THF08}.}
and later extended by \citet{THF10}
yielding the development of Typed Racket. From its inception,
occurrence typing was intimately tied to type systems with
set-theoretic types: unions, intersections, and negation of
types. Union was the first type connective to appear, since it was already used by
\citet{THF08} where its presence was needed to characterize the
different control flows of a type test, as our \code{foo} example
shows: one flow for integer arguments and another for
strings. Intersection types appear (in limited forms) combined with
occurrence typing both in TypeScript and in Flow and serve to give, among other,
more precise types to functions such as \code{foo}. For instance,
since \code{x\,+\,1} evaluates to an integer and \code{x.trim()} to a string, then our
function \code{foo} has type
\code{(number|string)$\to$(number|string)}.
But it is clear that a more precise type would be
one that states that \code{foo} returns a number when it is applied to
a number and returns a string when it is applied to a string, so that the type deduced for, say, \code{foo(42)} would be \code{number} rather than \code{number|string}. This is
exactly what the \emph{intersection type}
\begin{equation}\label{eq:inter}
  \code{(number$\to$number) \&  (string$\to$string)}
\end{equation}
states (intuitively, an expression has an intersection of types, noted \code{\&}, if and only if it has all the types of the intersection) and corresponds in Flow to declaring \code{foo} as follows:
\begin{alltt}\color{darkblue}
  var foo : \textcolor{darkred}{(number => number) & (string => string)} =  x => \{
      return (typeof(x) === "number")? x+1 : x.trim();  \refstepcounter{equation}                               \mbox{\color{black}\rm(\theequation)}\label{foo2}
  \} 
\end{alltt}
For what concerns negation types, they are pervasive in the occurrence
typing approach, even though they are used only at meta-theoretic
level,\footnote{At the moment of writing there is a pending pull
  request to add negation types to the syntax of TypeScript, but that is all.} in
particular to determine the type environment when the type case
fails. We already saw negation types at work when we informally typed the ``else''
branch in \code{foo}, for which we assumed that $x$ did \emph{not} have
type \code{number}---i.e., it had the (negation) type \code{$\neg$number}---and
deduced from it that $x$ then had type \code{string}---i.e.,
\code{(number|string)\&$\neg$number} which is equivalent to the set-theoretic difference
\code{(number|string)\textbackslash\,number} and, thus, to \code{string}.

The approaches cited above essentially focus on refining the type
of variables that occur in an expression whose type is being tested. They
do it when the variable occurs at top-level in the test (i.e., the variable is the
expression being tested) or under some specific positions such as in
nested pairs or at the end of a path of selectors. \beppe{Not precise,
  please check which are the cases that are handled in the cited
  papers}  In this work we aim at removing this limitation on the
contexts and develop a general theory to refine the type of variables
that occur in tested expressions under generic contexts, such as variables occurring in the
left or the right expressions of an application. In other words, we aim at establishing a formal framework to
extract as much static information as possible from a type test. We leverage our
analysis on the presence of full-fledged set-theoretic types
connectives provided by the theory of semantic subtyping. Our analysis
will also yield two important byproducts. First, to refine the type of the
variables we have to refine the type of the expressions they occur in and
we can use this information to improve our analysis.  Therefore our
occurrence typing approach will refine not only the types of variables
but also the types of generic expressions--i.e., any expression
whatever form it has---bypassing usual type
inference. Second, and most importantly, the result of our analysis can be used to infer
intersection types for functions, even in the absence of precise type
annotations such as the one in the definition of \code{foo} in~\eqref{foo2}: to put it simply, we are able to
infer the type~\eqref{eq:inter} for the unannotated pure JavaScript
code of \code{foo} (i.e., no type annotation at all), while in TypeScript and Flow (and any other
formalism we are aware of) 
this requires an explicit and full type annotation as the one given in~\eqref{foo2}.

Finally, the natural target for occurrence typing are languages with
dynamic type tests, in particular, dynamic languages. To type such
languages occurrence typing is often combined not only, as discussed
above, with set-theoretic types, but also with extensible record types
(to type objects) and gradual type system (to combine static and
dynamic typing) two features that we study in
Section~\ref{sec:extensions} as two extensions of our core formalism. Of
particular interest is the latter. \citet{Gre19} singles out
occurrence typing and gradual typing as \emph{the} two ``lineages'' that
partition the research on combining static and dynamic typing: he
identifies the former as the ``pragmatic,
implementation-oriented dynamic-first'' lineage and the latter as the
``formal, type-theoretic, static-first'' lineage.  Here we demonstrate that
these two ``lineages'' are not orthogonal or mutually independent, and 
we combine occurrence and gradual typing  showing, in particular, how
the former can be used to optimize the compilation of the latter.

\subsection{Motivating examples}
We focus our study on conditionals that  test types and consider the following syntax:
\(
\ifty{e}{t}{e}{e}
\) (e.g., in this syntax the body of \code{foo} in \eqref{foo} and \eqref{foo2} is rendered as
\(
  \ifty{x}{\Int}{x+1}{(\textsf{trim } x)}
\)).
In particular, in this introduction we concentrate on applications, since they constitute the most difficult case and many other cases can be reduced to them. A typical example is the expression
\begin{equation}\label{typical}
\ifty{x_1x_2}{t}{e_1}{e_2}
\end{equation}
where $x_i$'s denote variables, $t$ is some type, and $e_i$'s are generic expressions. 
Depending on the actual $t$ and on the static types of $x_1$ and $x_2$, we
can make type assumptions for $x_1$, for $x_2$, \emph{and} for the application $x_1x_2$
when typing $e_1$ that are different from those we can make when typing
$e_2$. For instance, suppose $x_1$ is bound to the function \code{foo} defined in \eqref{foo2}. Thus $x_1$ has type $(\Int\to\Int)\wedge(\String\to\String)$ (we used the syntax of the types of Section~\ref{sec:language} where unions and intersections are denoted by $\vee$ and $\wedge$ and have priority over $\to$ and $\times$, but not over $\neg$).
Then, it is not hard to see that if $x_2:\Int{\vee}\String$, then the expression\footnote{This and most of the following expressions are just given for the sake of example. Determining  the type \emph{in each branch} of expressions other than variables is interesting for constructors but less so for destructors such as applications, projections, and selections: any reasonable programmer would not repeat the same application twice, (s)he would store its result in a variable. This becomes meaningful with constructor such as pairs, as we do for instance in the expression in~\eqref{pair}.}
\begin{equation}\label{mezzo}
\texttt{let }x_1 \texttt{\,=\,}\code{foo}\texttt{ in } \ifty{x_1x_2}{\Int}{((x_1x_2)+x_2)}{\texttt{42}}
\end{equation}
is well typed with type $\Int$: when typing the branch ``then'' we
know that the test $x_1x_2\in \Int$
succeeded and that, therefore, not
only $x_1x_2$ is of type \Int, but also that $x_2$ is of type $\Int$: the other possibility,
$x_2:\String$, would have made the test fail.
For~\eqref{mezzo} we reasoned only on the type of the variables in the ``then'' branch but we can do the same
on the ``else'' branch as shown by the following expression, where \code{@} denotes string concatenation
\begin{equation}\label{two}
  \ifty{x_1x_2}{\Int}{((x_1x_2)+x_2)}{((x_1x_2)\code{\,@\,}x_2)}
\end{equation}
If the static type of
$x_1$ is $(\Int\to\Int)\wedge(\String\to\String)$ then $x_1x_2$ is well
typed only if the static type of $x_2$ is (a subtype of)
$\Int\vee\String$ and from that it is not hard to deduce that~\eqref{two}
has type $\Int\vee\String$. Let us see this in detail. The expression in~\eqref{two} is
typed in the following type environment:
$x_1:(\Int\to\Int)\wedge(\String\to\String), x_2:\Int\vee\String$. All we
can deduce, then, is that the application $x_1x_2$ has type
$\Int\vee\String$, which is not enough to type either the ``then'' branch
or the ``else'' branch. In order to type the ``then'' branch
$(x_1x_2)+x_2$ we must be able to deduce that both $x_1x_2$ and $x_2$
are of type \Int. Since we are in the ``then'' branch, then we know that
the type test succeeded and that, therefore, $x_1x_2$ has type \Int. Thus
we can assume in typing this branch that $x_1x_2$ has both its static
type and type \Int{} and, thus, their intersection:
$(\Int\vee\String)\wedge\Int$, that is \Int. For what concerns $x_2$ we
use the static type of $x_1$, that is
$(\Int\to\Int)\wedge(\String\to\String)$, and notice that this function
returns an \Int\ only if its argument is of type \Int. Reasoning as
above we thus deduce that in the ``then'' branch the type of $x_2$ is
the intersection of its static type with \Int:
$(\Int\vee\String)\wedge\Int$ that is \Int. To type the ``else'' branch
we reason exactly in the same way, with the only difference that, since
the type test has failed, then we know that the type of the tested expression is
\emph{not} \Int. That is, the expression $x_1x_2$ can produce any possible value
barring an \Int. If we denote by \Any\ the type of all values (i.e., the
type \code{any} of TypeScript and Flow) and by
$\setminus$ the set difference, then this means that in the else branch we
know that $x_1x_2$ has type $\Any{\setminus}\Int$---written
$\neg\Int$---, that is, it can return values of any type barred \Int. Reasoning as for the ``then'' branch we then assume that
$x_1x_2$ has type $(\Int\vee\String)\wedge\neg\Int$ (i.e., $(\Int\vee\String)\setminus\Int$, that is, \String), that
$x_2$ must be of type \String\ for the application to have type
$\neg\Int$ and therefore we assume that $x_2$ has type
$(\Int\vee\String)\wedge\String$ (i.e., again \String).

We have seen that we can specialize in both branches the type of the
whole expression $x_1x_2$, the type of the argument $x_2$,
but what about the type of the function $x_1$? Well, this depends on the type of
$x_1$ itself. In particular, if instead of an intersection type $x_1$
is typed by a union type (e.g., when the function bound to $x_1$ is
the result of a branching expression), then the test may give us information about
the type of the function in the various branches. So for instance if in the expression
in~\eqref{typical} $x_1$ is of type, say, $(s_1\to t)\vee(s_2\to\neg t)$, then we can
assume for the expression~\eqref{typical} that $x_1$ has type $(s_1\to t)$ in the branch ``then'' and
$(s_2\to \neg t)$ in the branch ``else''. 
As a more concrete example, if
$x_1:(\Int{\vee}\String\to\Int)\vee(\Bool{\vee}\String\to\Bool)$ and
$x_1x_2$ is well-typed, then we can deduce for
\begin{equation}\label{exptre}
\ifty{x_1x_2}{\Int}{(x_1(x_1x_2)+42)}{\texttt{not}(x_1(x_1x_2))}
\end{equation}
the type $\Int\vee\Bool$: in the ``then'' branch $x_1$ has type
$\Int{\vee}\String\to\Int$ and $x_1x_2$ is of type $\Int$; in the
``else'' branch $x_1$ has type $\Bool{\vee}\String\to\Bool$ and
$x_1x_2$ is of type $\Bool$.

Let us recap. If $e$ is an expression of type $t_0$ and we are trying to type
\(\ifty{e}{t}{e_1}{e_2}\),
then we can assume that $e$ has type $t_0\wedge t$ when typing $e_1$
and type $t_0\setminus t$ when typing $e_2$. If furthermore $e$ is of
the form $e'e''$, then we may also be able to specialize the types for $e'$ (in
particular if its static type is a union of arrows) and for $e''$ (in
particular if the static type of $e'$ is an intersection of
arrows). Additionally, we can repeat the reasoning for all subterms of $e'$
and $e''$ as long as they are applications, and deduce distinct types for all subexpressions of $e$ that
form applications. How to do it precisely---not only for applications, but also for other terms such as pairs, projections, records etc---is explained in the rest of
the paper but the key ideas are pretty simple and are presented next.

\subsection{Key ideas}\label{sec:ideas}

First of all, in a strict language we can consider a type as denoting the set of values of that type and subtyping as set-containment of the denoted values.
Imagine we are testing whether the result of an application $e_1e_2$
is of type $t$ or not, and suppose we know that the static types of
$e_1$ and $e_2$ are $t_1$ and $t_2$ respectively. If the application $e_1e_2$ is
well typed, then there is a lot of useful information that we can deduce from it:
first, that $t_1$ is a functional type (i.e., it denotes a set of
well-typed $\lambda$-abstractions, the values of functional type) whose domain, denoted by $\dom{t_1}$, is a type denoting the set of all values that are accepted by any function in
$t_1$; second that $t_2$ must be a subtype of the domain of $t_1$;
third, we also know the type of the application, that is the type
that denotes all the values that may result from the application of a
function in $t_1$ to an argument in $t_2$, type that we denote by
$t_1\circ t_2$. For instance, if $t_1=\Int\to\Bool$ and $t_2=\Int$,
then $\dom{t_1}=\Int$ and $t_1\circ t_2 = \Bool$. Notice that, introducing operations such as $\dom{}$ and $\circ$ is redundant when working with simple types, but becomes necessary in the presence of set-theoretic types. If for instance $t_1$ is the type of \eqref{foo2}, that is, \(t_1=(\Int{\to}\Int)\) \(\wedge\) \((\String{\to}\String)\),
then
\( \dom{t} = \Int \vee \String \), 
that is the union of all the possible
input types, while the precise return type of such a 
function depends on the type of the argument the function is applied to:
either an integer, or a string, or both (i.e., the union type
\(\Int\vee\String\)). So we have \( t_1 \circ \Int = \Int \),
\( t_1 \circ \String = \String \), and 
\( t_1 \circ (\Int\vee\String) = \Int \vee \String \) (see Section~\ref{sec:typeops} for the formal definition of $\circ$). 

What we want to do
is to refine the types of $e_1$ and $e_2$ (i.e., $t_1$ and $t_2$) for the cases where the test
that $e_1e_2$ has type $t$ succeeds or fails. Let us start with refining the type $t_2$ of $e_2$ for the case in
which the test succeeds. Intuitively, we want to remove from $t_2$ all
the values for which the application will surely return a result not
in $t$, thus making the test fail. Consider $t_1$ and let $s$ be the
largest subtype of $\dom{t_1}$ such that%
\svvspace{-1.29mm}
\begin{equation}\label{eq1}
  t_1\circ s\leq \neg t
\end{equation}
In other terms, $s$ contains all the legal arguments that make any function
in $t_1$ return a result not in $t$. Then we can safely remove from
$t_2$ all the values in $s$ or, equivalently, keep in $t_2$ all the
values of $\dom{t_1}$ that are not in $s$. Let us implement the second
viewpoint: the set of all elements of $\dom{t_1}$ for which an
application \emph{does not} surely give a result in $\neg t$ is
denoted $\worra{t_1}t$ (read, ``$t_1$ worra $t$'') and defined as $\min\{u\leq \dom{t_1}\alt
t_1\circ(\dom{t_1}\setminus u)\leq\neg t\}$: it is easy to see that
according to this definition $\dom{t_1}\setminus(\worra{t_1} t)$ is
the largest subset of $\dom{t_1}$ satisfying \eqref{eq1}. Then we can
refine the type of $e_2$ for when the test is successful by using the
type $t_2\wedge(\worra{t_1} t)$: we intersect all the possible results
of $e_2$, that is $t_2$, with the elements of the domain that
\emph{may} yield a result in $t$, that is $\worra{t_1} t$. When the test fails,
the type of $e_2$ can be refined in a similar way just by replacing $t$ by $\neg t$:
we get the refined type $t_2\land(\worra{t_1}{\neg t})$. To sum up, to refine the type of an
argument in the test of an application, all we need is to define
$\worra{t_1} t$, the set of arguments that when applied to a function
of type $t_1$ \emph{may} return a result in $t$; then we can refine the 
type of $e_2$ as $t_2^+ \eqdeftiny t_2\wedge(\worra{t_1} t)$ in the ``then'' branch (we call it the \emph{positive} branch)  
and as  $t_2^- \eqdeftiny t_2\setminus(\worra{t_1} t)$ in the ``else'' branch (we call it the \emph{negative} branch).
As a side remark note\marginpar{\tiny\beppe{Remove if space is needed}}
that the set $\worra{t_1} t$ is different from the set of elements that return a
result in $t$ (though it is a supertype of it). To see that, consider
for $t$ the type \String{} and for $t_1$ the type $(\Bool \to\Bool)\wedge(\Int\to(\String\vee\Int))$,
that is, the type of functions that when applied to a Boolean return a
Boolean and when applied to an integer return either an integer or a
string; then we have that $\dom{t_1}=\Int\vee\Bool$ and $\worra{t_1}\String=\Int$,
but there is no (non-empty) type that ensures that an application of a
function in $t_1$ will surely yield a $\String$ result.

Once we have determined $t_2^+$, it is then not very difficult to refine the
type $t_1$ for the positive branch, too. If the test succeeded, then we know two facts: first,
that the function was applied to a value in $t_2^+$ and, second, that
the application did not diverge and returned a result in
$t$. Therefore, we can exclude from $t_1$ all the functions that, when applied to an argument
in $t_2^+$, yield a result not in $t$.
It can be obtained simply by removing from $t_1$ the functions in
$t_2^+\to \neg t$, that is, we refine the type of $e_1$ in the ``then'' branch as
$t_1^+=t_1\setminus (t_2^+\to\neg t)$.
Note that this also removes functions diverging on $t_2^+$ arguments.
In particular, the interpretation of a type $t\to s$ is the set of all functions that when applied to an argument of type $t$
either diverge or return a value in $s$. As such the interpretation of $t\to s$ contains
all the functions that diverge (at least) on $t$. Therefore removing $t\to s$ from a type $u$ removes from $u$
not only all the functions that when applied to a $t$ argument return a result in $s$, but also all the functions that diverge on $t$.
Ergo $t_1\setminus (t_2^+\to
\neg t)$ removes, among others, all functions in $t_1$ that diverge on $t_2^+$.
Let us see all this on our example \eqref{exptre}, in particular, by showing how this technique deduces that the type of $x_1$ in the positive branch is (a subtype of) $\Int{\vee}\String\to\Int$.
Take the static type of $x_1$, that is $(\Int{\vee}\String\to\Int)\vee(\Bool{\vee}\String\to\Bool)$ and intersect it with
$\lnot(t_2^+\to \neg t)$, that is, $\neg(\String\to\neg\Int)$. Since intersection distributes over unions we
obtain
\svvspace{-1mm}
\[((\Int{\vee}\String{\to}\Int)\wedge\neg(\String{\to}\neg\Int))\vee((\Bool{\vee}\String{\to}\Bool)\wedge\neg(\String{\to}\neg\Int))\svvspace{-1mm}\]
and since
$(\Bool{\vee}\String{\to}\Bool)\wedge\neg(\String{\to}\neg\Int)$ is empty
(because $\String\to\neg\Int$ contains $\Bool{\vee}\String\to\Bool$),
then what we obtain is the left summand, a strict subtype of $(\Int{\vee}\String)\to\Int$, namely the functions of type $\Int{\vee}\String{\to}\Int$ minus those
that diverge on all \String{} arguments.

This is essentially what we formalize in Section~\ref{sec:language}, in the type system by the rule \Rule{PAppL} and in the typing algorithm with the case \eqref{due} of the definition of the function \constrf.

\subsection{Technical challenges}\label{sec:challenges}

In the previous section we outlined the main ideas of our approach to
occurrence typing. However, the devil is in the details. So the formalization we give in Section~\ref{sec:language} is not so smooth as we just outlined: we must introduce several auxiliary definitions to handle some corner cases. This section presents by tiny examples the main technical difficulties we had to overcome and the definitions we introduced to handle them. As such it provides a kind of road-map for the technicalities of  Section~\ref{sec:language}.

\paragraph{Typing occurrences} As it should be clear by now, not only variables
but also generic expressions are given different types in the ``then'' and
``else'' branches of type tests. For instance, in \eqref{two} the expression
$x_1x_2$ has type \Int{} in the positive branch and type \Bool{} in the negative
one. In this specific case it is possible to deduce these typings from the
refined types of the variables (in particular, thanks to the fact that $x_2$ has
type \Int{} the positive branch and \Bool{} in the negative one), but this is
not possible in general. For instance, consider $x_1:\Int\to(\Int\vee\Bool)$,
$x_2:\Int$, and the expression
\svvspace{-1mm}
\begin{equation}\label{twobis}
  \ifty{x_1x_2}{\Int}{...x_1x_2...}{...x_1x_2...}\svvspace{-1mm}
\end{equation}
It is not possible to specialize the type of the variables in the
branches. Nevertheless, we want to be able to deduce that $x_1x_2$ has
type \Int{} in the positive branch and type \Bool{} in the negative
one. In order to do so in Section~\ref{sec:language} we will use
special type environments that map not only variables but also generic
expressions to types. So to type, say, the positive branch of
\eqref{twobis} we extend the current type environment with the
hypothesis that the expression $x_1x_2$ has type \Int{}.

When we test the type of an expression we try to deduce the type of
some subexpressions occurring in it. Therefore we must cope with
subexpressions occurring multiple times. A simple example is given by
using product types and pairs as in
$\ifty{(x,x)}{\pair{t_1}{t_2}}{e_1}{e_2}$. It is easy to see that the
positive branch $e_1$ is selected only if $x$ has type $t_1$
\emph{and} type $t_2$ and deduce from that that $x$ must be typed in
$e_1$ by their intersection, $t_1\wedge t_2$. To deal with multiple
occurrences of a same subexpression the type inference system of
Section~\ref{sec:language} will use the classic rule for introducing intersections \Rule{Inter}, while the algorithmic counterpart will use the operator $\Refine{}{}{}$ that
intersects the static type of an expression with all the types deduced
for the multiple occurrences of it.

\paragraph{Type preservation} We want our type system to be sound in the sense of~\citet{Wright1994}, that is, that it
satisfies progress and type preservation. The latter property is
challenging because, as explained just above, our type
assumptions are not only about variables but also about 
expressions. Two corner cases are particularly difficult. The first is
shown by the following example\svvspace{-.9mm}
\begin{equation}\label{bistwo}
  \ifty{e(42)}{\Bool}{e}{...} \svvspace{-.9mm}
\end{equation}
If $e$ is an expression of type $\Int\to t$, then, as discussed before,
the positive branch will have type $(\Int\to t)\setminus(\Int\to\neg
\Bool)$. If furthermore the negative branch is of the same type (or of a subtype), then this
will also be the type of the whole expression in \eqref{bistwo}. Now imagine
that the application $e(42)$ reduces to a Boolean value, then the whole
  expression in \eqref{bistwo} reduces to $e$; but this has type
  $\Int\to t$ which, in general, is \emph{not} a subtype of $(\Int\to
  t)\setminus(\Int\to\neg\Bool)$, and therefore type is not preserved by the reduction. To cope with this problem, the proof
  of type preservation (see \Appendix\ref{app:subject-reduction}) resorts to \emph{type schemes}, a
  technique introduced by~\citet{Frisch2008} to type expressions by
  sets of types, so that the expression in \eqref{bistwo} will have both the types at issue. 

The second corner case is a modification of the example above
where the positive branch is $e(42)$, e.g.,
$\ifty{e(42)}{\Bool}{e(42)}{\textsf{true}}$. In this case the type deduced for the
whole expression is \Bool, while after reduction we would obtain
the expression $e(42)$ which is not of type \Bool{} but of type $t$ (even though it will
eventually reduce to a \Bool). This problem will be handled in the
proof of type preservation by considering parallel reductions (e.g, if
$e(42)$ reduces in a step to, say, $\textsf{false}$, then
$\ifty{e(42)}{\Bool}{e(42)}{\textsf{true}}$ reduces in one step to
$\ifty{\textsf{false}}{\Bool}{\textsf{false}}{\textsf{true}}$): see \Appendix\ref{app:parallel}.

\paragraph{Interdependence of checks} The last class of technical problems arise
from the mutual dependence of different type checks. In particular, there are two cases
that pose a problem. The first can be shown by two functions $f$ and $g$ both of type $(\Int\to\Int)\wedge(\Any\to\Bool)$, $x$ of type $\Any$ and the test:
\begin{equation}\label{nest1}
\ifty{(f\,x,g\,x)}{\pair\Int\Bool}{\,...\,}{\,...}
\end{equation}
If we independently check $f\,x$ against $\Int$ and $g\,x$ against $\Bool$
we deduce $\Int$ for the first occurrence of $x$ and $\Any$ for the
second. Thus we would type the positive branch of \eqref{nest1} under the hypothesis
that $x$ is of type $\Int$. But if we use the hypothesis generated by
the test of $f\,x$, that is, that $x$ is of type \Int, to check $g\,x$ against \Bool,
then the type deduced for $x$ is $\Empty$---i.e., the branch is never selected. In other words, we want to produce type
environments for occurrence typing by taking into account  all
the available hypotheses, even when these hypotheses are formulated later
in the flow of control. This will be done in the type systems of
Section~\ref{sec:language} by the rule \Rule{Path} and will require at
algorithmic level to look for a fix-point solution of a function, or
an approximation thereof.

Finally, a nested check may help refining
the type assumptions on some outer expressions. For instance, when typing
the positive branch $e$ of\svvspace{-.9mm}
\begin{equation}\label{pair}
\ifty{(x,y)}{(\pair{(\Int\vee\Bool)}\Int)}{e}{...}\svvspace{-.9mm}
\end{equation}
we can assume that the expression $(x,y)$ is of type
$\pair{(\Int\vee\Bool)}\Int$ and put it in the type environment. But
if in $e$ there is a test like
$\ifty{x}{\Int}{{\color{darkred}(x,y)}}{(...)}$ then we do not want use
the assumption in the type environment to type the expression $(x,y)$
occurring in the inner test (in red). Instead we want to give to that occurrence of the expression
$(x,y)$ the type $\pair{\Int}\Int$. This will be done by temporarily
removing the type assumption about $(x,y)$ from the type environment and
by retyping the expression without that assumption (see rule
\Rule{Env\Aa} in Section~\ref{sec:algorules}).

\iflongversion
\subsubsection*{Outline}
\else
\subsubsection*{Outline and Contributions}
\fi
In Section~\ref{sec:language} we formalize the
ideas we just presented: we define the types and
expressions of our system, their dynamic semantics and a type system that
implements occurrence typing together with the algorithms that decide
whether an expression is well typed or not. Section~\ref{sec:extensions} extends our formalism to record
types and presents two applications of our analysis: the inference of
arrow types for functions and a static analysis to reduce the number
of casts inserted by a compiler of a gradually-typed
language. Practical aspects are discussed in
Section~\ref{sec:practical} where we give several paradigmatic examples of code typed by our prototype implementation, that can be interactively tested at
\ifsubmission
the (anonymized) site
\fi
\url{https://occtyping.github.io/}. Section~\ref{sec:related} presents
related work.
\iflongversion
A discussion of future work concludes this
presentation.
\fi
To ease the presentation all the proofs are
omitted from the main text and can be found in the appendix.

\iflongversion
\subsubsection*{Contributions}
\fi
\noindent The main contributions of our work can be summarized as follows:
\begin{itemize}[left=0pt .. \parindent,nosep]
 \item We provide a theoretical framework to refine the type of
   expressions occurring in type tests, thus removing the limitations
   of current occurrence typing approaches which require both the tests and the refinement
   to take place on variables.

 \item We define a type-theoretic approach alternative to the
   current flow-based approaches. As such it provides different
   results and it can be thus profitably combined with flow-based techniques.
   
 \item We use our analysis for defining a formal framework that
   reconstructs intersection types for unannotated or
   partially-annotated functions, something that, in our ken, no
   other current system can do.
   
 \item We prove the soundness of our system. We define algorithms to
   infer the types that we prove to be sound and show different completeness results
   which in practice yield the completeness of any reasonable implementation.

 \item We show how to extend our approach to records with field
   addition, update, and deletion operations.  
   
 \item We show how  occurrence typing can be extended to and combined with
   gradual typing and apply our results to optimize the compilation of the
   latter.

\end{itemize}
  We end this introduction by stressing the practical
  implications of our work: a perfunctory inspection may give the
  wrong impression that the only interest of the heavy formalization
  that follows is to have generic expressions, rather than just variables, in type
  cases: this would be a bad trade-off. The important point is,
  instead, that our formalization is what makes analyses such as those
  presented in Section~\ref{sec:extensions} possible (e.g., the
  reconstruction of the type~\eqref{eq:inter} for the unannotated pure
  JavaScript code of \code{foo}), which is where the actual added
  practical value and potential of our work resides.

\section{Language}
\label{sec:language}
\newlength{\sk}
\setlength{\sk}{-1.9pt}
\iflongversion
In this section we formalize the ideas we outlined in the introduction. We start by the definition of types followed by the language and its reduction semantics. The static semantics is the core of our work: we first present a declarative type system that deduces (possibly many) types for well-typed expressions and then the algorithms to decide whether an expression is well typed or not. 
\fi

\subsection{Types}
\begin{definition}[Types]\label{def:types}
The set of types \types{} is formed by the terms $t$ coinductively produced by the grammar:\svvspace{-1.45mm}
\[
\begin{array}{lrcl}
\textbf{Types} & t & ::= & b\alt t\to t\alt t\times t\alt t\vee t \alt \neg t \alt \Empty 
\end{array}
\]
and that satisfy the following conditions
\begin{itemize}[nosep]
\item (regularity) every term has a finite number of different sub-terms;
\item (contractivity) every infinite branch of a term contains an infinite number of occurrences of the
arrow or product type constructors.\svvspace{-1mm}
\end{itemize}
\end{definition}
We use the following abbreviations: $
    t_1 \land t_2 \eqdef \neg (\neg t_1 \vee \neg t_2)$, 
    $t_ 1 \setminus t_2 \eqdef t_1 \wedge \neg t_2$, $\Any \eqdef \neg \Empty$.
$b$ ranges over basic types
(e.g., \Int, \Bool),
$\Empty$ and $\Any$ respectively denote the empty (that types no value)
and top (that types all values) types. Coinduction accounts for
recursive types and the condition on infinite branches bars out
ill-formed types such as 
$t = t \lor t$ (which does not carry any information about the set
denoted by the type) or $t = \neg t$ (which cannot represent any
set). 
\iflongversion
It also ensures that the binary relation $\vartriangleright
\,\subseteq\!\types{\times}\types$ defined by $t_1 \lor t_2 \vartriangleright
t_i$, $t_1 \land t_2 \vartriangleright
t_i$, $\neg t \vartriangleright t$ is Noetherian.
This gives an induction principle on $\types$ that we
will use without any further explicit reference to the relation.\footnote{In a nutshell, we can do proofs by induction on the structure of unions and negations---and, thus, intersections---but arrows, products, and basic types are the base cases for the induction.} 
\fi
We refer to $ b $, $\times$, and $ \to $ as \emph{type constructors}
and to $ \lor $, $ \land $, $ \lnot $, and $ \setminus $
as \emph{type connectives}.

The subtyping relation for these types, noted $\leq$, is the one defined
by~\citet{Frisch2008} 
\iflongversion
and detailed description of the algorithm to
decide this relation can be found in~\cite{Cas15}. For the reader's
convenience we succinctly recall the definition of the subtyping
relation in the next subsection but it is possible to skip this
subsection at first reading and jump directly to Subsection~\ref{sec:syntax}, since to understand
the rest of the paper
\else
to which the reader may refer for the formal
definition (we recall it in \Appendix\ref{sec:subtyping} for the
reader's convenience).
A detailed description of the algorithm to
decide this relation can be found in~\cite{Cas15}.
For this presentation
\fi
it suffices to consider that
types are interpreted as sets of \emph{values} (i.e., either
constants, $\lambda$-abstractions, or pairs of values: see
Section~\ref{sec:syntax} right below) that have that type, and that subtyping is set
containment (i.e., a type $s$ is a subtype of a type $t$ if and only if $t$
contains all the values of type $s$). In particular, $s\to t$
contains all $\lambda$-abstractions that when applied to a value of
type $s$, if their computation terminates, then they return a result of
type $t$ (e.g., $\Empty\to\Any$ is the set of all
functions\footnote{\label{allfunctions}Actually, for every type $t$,
all types of the form $\Empty{\to}t$ are equivalent and each of them
denotes the set of all functions.} and $\Any\to\Empty$ is the set
of functions that diverge on every argument). Type connectives
(i.e., union, intersection, negation) are interpreted as the
corresponding set-theoretic operators (e.g.,~$s\vee t$ is the
union of the values of the two types). We use $\simeq$ to denote the
symmetric closure of $\leq$: thus $s\simeq t$ (read, $s$ is equivalent to $t$) means that $s$ and $t$ denote the same set of values and, as such, they are semantically the same type.
\iflongversion
All the above is formalized as follows.
\subsection{Subtyping}
\label{sec:subtyping}
\noindent
Subtyping is defined by giving a set-theoretic interpretation of the
types of Definition~\ref{def:types} into a suitable domain $\Domain$: 
\begin{definition}[Interpretation domain~\cite{Frisch2008}]\label{def:interpretation}
The \emph{interpretation domain} $ \Domain $ is the set of finite terms $ d $
produced inductively by the following grammar\vspace{-2mm}
\begin{align*}
  d & \Coloneqq  c \mid (d, d) \mid \Set{(d, \domega), \dots, (d, \domega)}
    \\
    \domega & \Coloneqq d \mid \Omega
\end{align*}
where $ c $ ranges over the set $ \Constants $ of constants
and where $ \Omega $ is such that $ \Omega \notin \Domain $.
\end{definition}
The elements of $ \Domain $ correspond, intuitively,
to (denotations of) the results of the evaluation of expressions.
In particular, in a higher-order language,
the results of  computations can be functions which, in this model,
are represented by sets of finite relations
of the form $ \Set{(d_1, \domega_1), \dots, (d_n, \domega_n)} $,
where $ \Omega $ (which is  not in $ \Domain $)
can appear in second components to signify
that the function fails (i.e., evaluation is stuck) on the
corresponding input. This is implemented by using in the second
projection the meta-variable $\domega$ which ranges over
$ \Domain_\Omega = \Domain \cup \Set{\Omega} $ (we reserve
$d$ to range over $\Domain$, thus excluding $\Omega$).
This constant $\Omega$ is used to ensure that $\Any\to\Any$ is not a
supertype of all function types: if we used $d$ instead of $\domega$,
then every well-typed function could be subsumed to  $\Any\to\Any$
and, therefore, every application could be given the type $\Any$,
independently from its argument as long as this argument is typable (see Section 4.2 of~\cite{Frisch2008} for details).
The restriction to \emph{finite} relations corresponds to the intuition
that the denotational semantics of a function is given by the set of
its finite approximations, where finiteness is a restriction necessary
(for cardinality reasons) to give the
semantics to higher-order functions.

We define the interpretation $ \TypeInter{t} $ of a type $ t $
so that it satisfies the following equalities,
where  $ \Pf $ denotes the restriction of the powerset to finite
subsets and $  \ConstantsInBasicType{}$ denotes the function
that assigns to each basic type the set of constants of that type, so
  that  for every constant $c$ we have $
  c\in \ConstantsInBasicType(\basic {c})$
\ifelsevierstyle
(we use $\basic{c}$ to denote the
basic type of the constant $c$):
\else
($\basic c$ is
  defined in Section~\ref{sec:syntax}):
\fi
\begin{align*}
  \TypeInter{\Empty} & = \emptyset&
  \TypeInter{t_1 \lor t_2} & = \TypeInter{t_1} \cup \TypeInter{t_2} &
  \TypeInter{\lnot t} & = \Domain \setminus \TypeInter{t} 
  \\
  \TypeInter{b} & = \ConstantsInBasicType(b) &
  \TypeInter{t_1 \times t_2} & = \TypeInter{t_1} \times \TypeInter{t_2} \\
  \TypeInter{t_1 {\to} t_2} & =
    \{R \in \Pf(\Domain{\times}\Domain_\Omega) \mid\forall (d, \domega) \in R. \:d \in \TypeInter{t_1} \implies \domega \in \TypeInter{t_2}\}
   \span\span
   \span\span
\end{align*}
We cannot take the equations above
directly as an inductive definition of $ \TypeInter{} $
because types are not defined inductively but coinductively.
\iflongversion
However, recall that the contractivity condition of
Definition~\ref{def:types}  ensures that the binary relation $\vartriangleright
\,\subseteq\!\types{\times}\types$ defined by $t_1 \lor t_2 \vartriangleright
t_i$, $t_1 \land t_2 \vartriangleright t_i$, $\neg t \vartriangleright t$ is Noetherian which gives an induction principle  on $\types$ that we
use combined with  structural induction on $\Domain$ to give the
following definition
which validates these equalities.
\else
Notice however that the contractivity condition of
Definition~\ref{def:types}  ensures that the binary relation $\vartriangleright
\,\subseteq\!\types{\times}\types$ defined by $t_1 \lor t_2 \vartriangleright
t_i$, $t_1 \land t_2 \vartriangleright
t_i$, $\neg t \vartriangleright t$ is Noetherian.
This gives an induction principle\footnote{In a nutshell, we can do
proofs and give definitions by induction on the structure of unions and negations---and, thus, intersections---but arrows, products, and basic types are the base cases for the induction.}  on $\types$ that we
use combined with  structural induction on $\Domain$ to give the following definition,
which validates these equalities.
\fi
\begin{definition}[Set-theoretic interpretation of types~\cite{Frisch2008}]\label{def:interpretation-of-types}
We define a binary predicate $ (d : t) $
(``the element $ d $ belongs to the type $t$''),
where $ d \in \Domain $ and $ t \in \types $,
by induction on the pair $ (d, t) $ ordered lexicographically.
The predicate is defined as follows:
\begin{align*}
  (c : b)  &= c \in \ConstantsInBasicType(b) \\
  ((d_1, d_2) : t_1 \times t_2 )  &=
    (d_1 : t_1) \mathrel{\mathsf{and}} (d_2 : t_2) \\
  (\Set{(d_1, \domega_1),..., (d_n, \domega_n)} : t_1 \to t_2)  &=
   \forall i \in [1.. n] . \:
    \mathsf{if} \: (d_i : t_1) \mathrel{\mathsf{then}} (\domega_i : t_2) \\
  (d : t_1 \lor t_2)  &= (d : t_1) \mathrel{\mathsf{or}} (d : t_2) \\
  (d : \lnot t)  &= \mathsf{not} \: (d : t) \\
  (\domega : t)  &= \mathsf{false} & \text{ otherwise}
  \end{align*}
We define the \emph{set-theoretic interpretation}
$ \TypeInter{} : \types \to \Pd(\Domain) $
as $ \TypeInter{t} = \{d \in \Domain \mid (d : t)\} $.
\end{definition}
Finally,
we define the subtyping preorder and its associated equivalence relation
as follows.

\begin{definition}[Subtyping relation~\cite{Frisch2008}]\label{def:subtyping}
  We define the \emph{subtyping} relation $ \leq $
  and the \emph{subtyping equivalence} relation $ \simeq $
  as
  \(
    t_1 \leq t_2 \iffdef \TypeInter{t_1} \subseteq \TypeInter{t_2}\) and   
  \(t_1 \simeq t_2 \iffdef (t_1 \leq t_2) \mathrel{\mathsf{and}} (t_2 \leq t_1)
    \: .
  \)
\end{definition}

\fi

\subsection{Syntax}\label{sec:syntax}
The expressions $e$ and values $v$ of our language are inductively generated by the following grammars:\svvspace{-2mm}
\begin{equation}\label{expressions}
\begin{array}{lrclr}  
  \textbf{Expr} &e &::=& c\alt x\alt ee\alt\lambda^{\wedge_{i\in I}s_i\to t_i} x.e\alt \pi_j e\alt(e,e)\alt\tcase{e}{t}{e}{e}\\[.3mm]
  \textbf{Values} &v &::=& c\alt\lambda^{\wedge_{i\in I}s_i\to t_i} x.e\alt (v,v)\\
\end{array}
\end{equation}
for $j=1,2$. In~\eqref{expressions}, $c$ ranges over constants
(e.g., \texttt{true}, \texttt{false}, \texttt{1}, \texttt{2},
...) which are values of basic types%
\ifelsevierstyle\else
\ (we use $\basic{c}$ to denote the
basic type of the constant $c$)%
\fi
; $x$ ranges over variables; $(e,e)$
denotes pairs and $\pi_i e$ their projections; $\tcase{e}{t}{e_1}{e_2}$
denotes the type-case expression that evaluates either $e_1$ or $e_2$
according to whether the value returned by $e$ (if any) has the type $t$
or not; $\lambda^{\wedge_{i\in I}s_i\to t_i} x.e$ denotes the function of parameter $x$
and body $e$ annotated with the type $\wedge_{i\in I}s_i\to t_i$. An expression has an intersection type if and only if it
has all the types that compose the intersection. Therefore,
intuitively, $\lambda^{\wedge_{i\in I}s_i\to t_i} x.e$ is a well-typed
expression if for all $i{\in} I$ the hypothesis that $x$ is of type $s_i$
implies that the body $e$ has type $t_i$, that is to say, it is well
typed if $\lambda^{\wedge_{i\in I}s_i\to t_i} x.e$ has type $s_i\to
t_i$ for all $i\in I$.

\subsection{Dynamic semantics}\label{sec:opsem}

The dynamic semantics is defined as a classic left-to-right
call-by-value weak reduction for a $\lambda$-calculus with pairs, enriched with specific rules for type-cases. We have the following  notions of reduction:\svvspace{-1.2mm}
\[
\begin{array}{rcll}
  (\lambda^{\wedge_{i\in I}s_i\to t_i} x.e)\,v &\reduces& e\subst x v\\[-.4mm]
  \pi_i(v_1,v_2) &\reduces& v_i & i=1,2\\[-.4mm]
  \tcase{v}{t}{e_1}{e_2} &\reduces& e_1 &v\in \valsemantic t\\[-.4mm] 
  \tcase{v}{t}{e_1}{e_2} &\reduces& e_2 &v\not\in \valsemantic t\\[-1.3mm]
\end{array}
\]
where $\valsemantic t$ denotes, intuitively, the set of values that have type $t$.
Formally, $\valsemantic t=\{v\alt \exists t'\in \vtyof{v}.\ t'\leq t\}$ where $\vtyof{v}$ is inductively defined as: $\vtyof c \eqdef \{\basic{c}\}$,\,\,\,
$\vtyof{\lambda^{\wedge_{i\in I}s_i\to t_i} x.e} \eqdef \{t\alt t\simeq (\wedge_{i\in I}s_i\to t_i)\wedge(\wedge_{j\in J}s_j'\to t_j'), t\not\leq\Empty\}$,\,\,\,
$\vtyof{(v_1,v_2)} \eqdef \vtyof{v_1}\times\vtyof{v_2}$
\footnote{This definition may look
complicated but it is necessary to handle some corner cases for
negated arrow types (cf.\ rule \Rule{Abs-} in
Section~\ref{sec:static}). For instance, it states that $\lambda^{\Int{\to}\Int}x.x\in
\valsemantic{(\Int{\to}\Int)\wedge\neg(\Bool{\to}\Int)}$.}.\\

Contextual reductions are
defined by the following evaluation contexts:
\[
\Cx[] ::= [\,]\alt \Cx e\alt v\Cx \alt (\Cx,e)\alt (v,\Cx)\alt \pi_i\Cx\alt \tcase{\Cx}tee
\]
As usual we denote by $\Cx[e]$ the term obtained by replacing $e$ for
the hole in the context $\Cx$ and we have that $e\reduces e'$ implies
$\Cx[e]\reduces\Cx[e']$.

\subsection{Static semantics}\label{sec:static}

While the syntax and reduction semantics are, on the whole, pretty
standard, for what concerns the type system we will have to introduce several
unconventional features that we anticipated in
Section~\ref{sec:challenges} and are at the core of our work. Let
us start with the standard part, that is the typing of the functional
core and the use of subtyping, given by the following typing rules:\svvspace{-1mm}
\begin{mathpar}
  \Infer[Const]
      { }
      {\Gamma\vdash c:\basic{c}}
      { }
  \quad
  \Infer[App]
      {
        \Gamma \vdash e_1: \arrow {t_1}{t_2}\quad
        \Gamma \vdash e_2: t_1
      }
      { \Gamma \vdash {e_1}{e_2}: t_2 }
      { }
  \quad
  \Infer[Abs+]
      {{\scriptstyle\forall i\in I}\quad\Gamma,x:s_i\vdash e:t_i}
      {
      \Gamma\vdash\lambda^{\wedge_{i\in I}\arrow {s_i} {t_i}}x.e:\textstyle \bigwedge_{i\in I}\arrow {s_i} {t_i}
      }
      { }
\end{mathpar}
\begin{mathpar}
      \Infer[Sel]
  {\Gamma \vdash e:\pair{t_1}{t_2}}
  {\Gamma \vdash \pi_i e:t_i}
  { }
  \qquad
  \Infer[Pair]
  {\Gamma \vdash e_1:t_1 \and \Gamma \vdash e_2:t_2}
  {\Gamma \vdash (e_1,e_2):\pair {t_1} {t_2}}
  { }
  \qquad
    \Infer[Subs]
      { \Gamma \vdash e:t\\t\leq t' }
      { \Gamma \vdash e: t' }
      { }
  \qquad\svvspace{-3mm}
\end{mathpar}
These rules are quite standard and do not need any particular explanation besides those already given in Section~\ref{sec:syntax}. Just notice subtyping is embedded in the system by the classic \Rule{Subs} subsumption rule. Next we focus on the unconventional aspects of our system, from the simplest to the hardest.

The first unconventional aspect is that, as explained in
Section~\ref{sec:challenges}, our type assumptions are about
expressions. Therefore, in our rules the type environments, ranged over
by $\Gamma$, map \emph{expressions}---rather than just variables---into
types. This explains why the classic typing rule for variables is replaced by a more general \Rule{Env} rule defined below:\svvspace{-1mm}
\begin{mathpar}
  \Infer[Env]
      { }
      { \Gamma \vdash e: \Gamma(e) }
      { e\in\dom\Gamma }
  \qquad
  \Infer[Inter]
      { \Gamma \vdash e:t_1\\\Gamma \vdash e:t_2 }
      { \Gamma \vdash e: t_1 \wedge t_2 }
      { }\svvspace{-3mm}
\end{mathpar}
The \Rule{Env} rule is coupled with the standard intersection introduction rule \Rule{Inter}
which allows us to deduce for a complex expression the intersection of
the types recorded by the occurrence typing analysis in the
environment $\Gamma$ with the static type deduced for the same
expression by using the other typing rules. This same intersection
rule is also used to infer the second unconventional aspect of our
system, that is, the fact that $\lambda$-abstractions can have negated
arrow types, as long as these negated types do not make the type deduced for the function empty:\svvspace{-.5mm}
\begin{mathpar}
  \Infer[Abs-]
    {\Gamma \vdash \lambda^{\wedge_{i\in I}\arrow {s_i} {t_i}}x.e:t}
    { \Gamma \vdash\lambda^{\wedge_{i\in I}\arrow {s_i} {t_i}}x.e:\neg(t_1\to t_2)  }
    { ((\wedge_{i\in I}\arrow {s_i} {t_i})\wedge\neg(t_1\to t_2))\not\simeq\Empty }\svvspace{-1.2mm}
\end{mathpar}
\rev{
In Section~\ref{sec:challenges} we explained that in order for our
system to satisfy the property of type preservation, the type system
must be able to deduce negated arrow types for functions---e.g. the
type $(\Int\to\Int)\wedge\neg(\Bool\to\Bool)$ for $\lambda^{\Int\to\Int}
x.x$. We demonstrated this with the expression in
equation \eqref{bistwo}, for which type preservation holds only if we are
able to deduce for this expression the type
$(\Int\to t)\setminus(\Int\to\neg\Bool)$, that is, $(\Int\to t)\wedge\neg(\Int\to\neg\Bool)$.
}
But the sole rule \Rule{Abs+}
above does not allow us to deduce  negations of
arrows for $\lambda$-abstractions: the rule \Rule{Abs-} makes this
possible.
\rev{
This rule ensures that given a function $\lambda^t x.e$ (where $t$
is an intersection type), for every type $t_1\to t_2$, either
$t_1\to t_2$ can be obtained by subsumption from $t$ or $\neg(t_1\to
t_2)$ can be added to the intersection $t$. In turn this ensures
that, for any function and any type $t$ either the function has type
$t$ or it has type $\neg t$ (see~\citet[Sections 3.3.2 and
3.3.3]{Pet19phd} for a thorough discussion on this rule).
}
As an aside, note that this kind
of deduction is already present in the system by~\citet{Frisch2008}
though in that system this presence was motivated by the semantics of types rather than, as in our case,
by the soundness of the type system.

Rules \Rule{Abs+} and \Rule{Abs-} are not enough to deduce for
$\lambda$-abstractions all the types we wish. In particular, these
rules alone are not enough to type general overloaded functions. For
instance, consider this simple example of a function that applied to an
integer returns its successor and applied to anything else returns
\textsf{true}:\\[1mm]
\centerline{\(
\lambda^{(\Int\to\Int)\wedge(\neg\Int\to\Bool)} x\,.\,\tcase{x}{\Int}{x+1}{\textsf{true}}
\)}\\[.6mm]
Clearly, the expression above is well typed, but the rule \Rule{Abs+} alone
is not enough to type it. In particular, according to \Rule{Abs+} we
have to prove that under the hypothesis that $x$ is of type $\Int$ the expression
$(\tcase{x}{\Int}{x+1}{\textsf{true}})$ is of type $\Int$, too.  That is, that under the
hypothesis that $x$ has type $\Int\wedge\Int$ (we apply occurrence
typing) the expression $x+1$ is of type \Int{} (which holds) and that under the
hypothesis that $x$ has type $\Int\setminus\Int$, that is $\Empty$
(we apply once more occurrence typing), \textsf{true} is of type \Int{}
(which \emph{does not} hold). The problem is that we are trying to type the
second case of a type-case even if we know that there is no chance that, when $x$ is bound to an integer,
that case will be ever selected. The fact that it is never selected is witnessed
by the presence of a type hypothesis with  $\Empty$ type. To
avoid this problem (and type the term above) we add the rule
\Rule{Efq} (\emph{ex falso quodlibet}) that allows the system to deduce any type
for an expression that will never be selected, that is, for an
expression whose type environment contains an empty assumption:
\begin{mathpar}
  \Infer[Efq]
  { }
  { \Gamma, (e:\Empty) \vdash e': t }
  { }\svvspace{-3mm}
\end{mathpar}
Once more, this kind of deduction was already present in the system
by~\citet{Frisch2008} to type full fledged overloaded functions,
though it was embedded in the typing rule for the type-case.
Here we
need the rule \Rule{Efq}, which is more general, to ensure the
property of subject reduction.

Finally, there remains one last rule in our type system, the one that
implements occurrence typing, that is, the rule for the
type-case:\svvspace{-1mm}
\begin{mathpar}
    \Infer[Case]
        {\Gamma\vdash e:t_0\\
        \Gamma \evdash e t \Gamma_1 \\ \Gamma_1 \vdash e_1:t'\\
        \Gamma \evdash e {\neg t} \Gamma_2 \\ \Gamma_2 \vdash e_2:t'}
        {\Gamma\vdash \tcase {e} t {e_1}{e_2}: t'}
        { }\svvspace{-3mm}
\end{mathpar}
The rule \Rule{Case} checks whether the expression $e$, whose type is
being tested, is well-typed and then performs the occurrence typing
analysis that produces the environments $\Gamma_i$'s under whose
hypothesis the expressions $e_i$'s are typed. The production of these
environments is represented by the judgments $\Gamma \evdash e
{(\neg)t} \Gamma_i$. The intuition is that when $\Gamma \evdash e t
\Gamma_1$ is provable then $\Gamma_1$ is a version of $\Gamma$
extended with type hypotheses for all expressions occurring in $e$,
type hypotheses that can be deduced assuming that the test $e\in t$
succeeds. Likewise, $\Gamma \evdash e {\neg t} \Gamma_2$ (notice the negation on $t$) extends
$\Gamma$ with the hypothesis deduced assuming that $e\in\neg t$, that
is, for when the test $e\in t$ fails.

All it remains to do is to show how to deduce judgments of the form
$\Gamma \evdash e t \Gamma'$. For that we first define how
to denote occurrences of an expression. These are identified by paths in the
syntax tree of the expressions, that is, by possibly empty strings of
characters denoting directions starting from the root of the tree (we
use $\epsilon$ for the empty string/path, which corresponds to the
root of the tree).

Let $e$ be an expression and $\varpi\in\{0,1,l,r,f,s\}^*$ a
\emph{path}; we denote $\occ e\varpi$ the occurrence of $e$ reached by
the path $\varpi$, that is (for $i=0,1$, and undefined otherwise)\svvspace{-.4mm}
\[
\begin{array}{r@{\downarrow}l@{\quad=\quad}lr@{\downarrow}l@{\quad=\quad}lr@{\downarrow}l@{\quad=\quad}l}
e&\epsilon & e & (e_1,e_2)& l.\varpi & \occ{e_1}\varpi &\pi_1 e& f.\varpi & \occ{e}\varpi\\
e_0\,e_1& i.\varpi & \occ{e_i}\varpi \quad\qquad& (e_1,e_2)& r.\varpi & \occ{e_2}\varpi \quad\qquad&
\pi_2 e& s.\varpi & \occ{e}\varpi\\[-.4mm]
\end{array}
\]
To ease our analysis we used different directions for each kind of
term. So we have $0$ and $1$ for the function and argument of an
application, $l$ and $r$ for the $l$eft and $r$ight expressions forming a pair,
and $f$ and $s$ for the argument of a $f$irst or of a $s$econd projection. Note also that we do not consider occurrences
under $\lambda$'s (since their type is frozen in their annotations) and type-cases (since they reset the analysis).
The judgments  $\Gamma \evdash e t \Gamma'$ are then deduced by the following two rules:\svvspace{-1mm} \begin{mathpar}
    \Infer[Base]
      { }
      { \Gamma \evdash e t \Gamma }
      { }
    \qquad
    \Infer[Path]
      { \pvdash {\Gamma'} e t \varpi:t' \\ \Gamma \evdash e t \Gamma' }
      { \Gamma \evdash e t \Gamma',(\occ e \varpi:t') }
      { }\svvspace{-1.5mm}
\end{mathpar}
These rules describe how to produce by occurrence typing the type
environments while checking that an expression $e$ has type $t$. They state that $(i)$ we can
deduce from $\Gamma$ all the hypothesis already in $\Gamma$ (rule
\Rule{Base}) and that $(ii)$ if we can deduce a given type $t'$ for a particular
occurrence $\varpi$ of the expression $e$ being checked, then we can add this
hypothesis to the produced type environment (rule \Rule{Path}). The rule
\Rule{Path} uses a (last) auxiliary judgement $\pvdash {\Gamma}  e t
\varpi:t'$ to deduce the type $t'$ of the occurrence $\occ e \varpi$ when
checking $e$ against $t$ under the hypotheses $\Gamma$. This rule \Rule{Path} is subtler than it may appear at
first sight, insofar as the deduction of the type for $\varpi$ may already use
some hypothesis on $\occ e \varpi$ (in $\Gamma'$) and, from an
algorithmic viewpoint, this will imply the computation of a fix-point
(see Section~\ref{sec:typenv}). The last ingredient for our type system is the deduction of the
judgements of the form $\pvdash {\Gamma}  e t \varpi:t'$ where
$\varpi$ is a path to an expression occurring in $e$. This is given by the following set
of rules.
\begin{mathpar}
    \Infer[PSubs]
        { \pvdash \Gamma e t \varpi:t_1 \\ t_1\leq t_2 }
        { \pvdash \Gamma e t \varpi:t_2 }
        { }
        \quad
    \Infer[PInter]
        { \pvdash \Gamma e t \varpi:t_1 \\ \pvdash \Gamma e t \varpi:t_2 }
        { \pvdash \Gamma e t \varpi:t_1\land t_2 }
        { }
        \quad
    \Infer[PTypeof]
        { \Gamma \vdash \occ e \varpi:t' }
        { \pvdash \Gamma e t \varpi:t' }
        { }
\svvspace{-1.2mm}\\
    \Infer[PEps]
        { }
        { \pvdash \Gamma e t \epsilon:t }
        { }
        \qquad
    \Infer[PAppR]
        { \pvdash \Gamma e t \varpi.0:\arrow{t_1}{t_2} \\ \pvdash \Gamma e t \varpi:t_2'}
        { \pvdash \Gamma e t \varpi.1:\neg t_1 }
        { t_2\land t_2' \simeq \Empty  }
\end{mathpar}\begin{mathpar}\svvspace{-2mm}
    \Infer[PAppL]
        { \pvdash \Gamma e t \varpi.1:t_1 \\ \pvdash \Gamma e t \varpi:t_2 }
        { \pvdash \Gamma e t \varpi.0:\neg (\arrow {t_1} {\neg t_2}) }
        { }
        \qquad
    \Infer[PPairL]
        { \pvdash \Gamma e t \varpi:\pair{t_1}{t_2} }
        { \pvdash \Gamma e t \varpi.l:t_1 }
        { }
\svvspace{-1.2mm}\\
    \Infer[PPairR]
        { \pvdash \Gamma e t \varpi:\pair{t_1}{t_2} }
        { \pvdash \Gamma e t \varpi.r:t_2 }
        { }
        \qquad
    \Infer[PFst]
        { \pvdash \Gamma e t \varpi:t' }
        { \pvdash \Gamma e t \varpi.f:\pair {t'} \Any }
        { }
        \qquad
    \Infer[PSnd]
        { \pvdash \Gamma e t \varpi:t' }
        { \pvdash \Gamma e t \varpi.s:\pair \Any {t'} }
        { }\svvspace{-0.9mm}
\end{mathpar}
These rules implement the analysis described in
Section~\ref{sec:ideas} for functions and extend it to products.  Let
us comment each rule in detail. \Rule{PSubs} is just subsumption for
the deduction $\vdashp$. The rule \Rule{PInter} combined with
\Rule{PTypeof} allows the system to deduce for an occurrence $\varpi$
the intersection of the static type of $\occ e \varpi$ (deduced by
\Rule{PTypeof}) with the type deduced for $\varpi$ by the other $\vdashp$ rules. The
rule \Rule{PEps} is the starting point of the analysis: if we are assuming that the test $e\in t$ succeeds, then we can assume that $e$ (i.e.,
$\occ e\epsilon$) has type $t$ (recall that assuming that the test $e\in t$ fails corresponds to having $\neg t$ at the index of the turnstyle).
The rule \Rule{PAppR} implements occurrence typing for
the arguments of applications, since it states that if a function maps
arguments of type $t_1$ in results of type $t_2$ and an application of
this function yields results (in $t'_2$) that cannot be in $t_2$
(since $t_2\land t_2' \simeq \Empty$), then the argument of this application cannot be of type $t_1$. \Rule{PAppL} performs the
occurrence typing analysis for the function part of an application,
since it states that if an application has type $t_2$ and the argument
of this application has type $t_1$, then the function in this
application cannot have type $t_1\to\neg t_2$. Rules \Rule{PPair\_}
are straightforward since they state that the $i$-th projection of a pair
that is of type $\pair{t_1}{t_2}$ must be of type $t_i$. So are the last two
rules that essentially state that if $\pi_1 e$ (respectively, $\pi_2
e$) is of type $t'$, then the type of $e$ must be of the form
$\pair{t'}\Any$ (respectively, $\pair\Any{t'}$).

This concludes the presentation of all the rules of our type system
(they are summarized for the reader's convenience in \Appendix\ref{sec:declarative}), which satisfies
the property of safety, deduced, as customary, from the properties
of progress and subject reduction (\emph{cf.} \Appendix\ref{app:soundness}).\svvspace{-.5mm}
\begin{theorem}[type safety]
For every expression $e$ such that $\varnothing\vdash e:t$ either  $e$
diverges or there
exists a value $v$ of type $t$ such that $e\reduces^* v$.
\end{theorem}
\svvspace{-2.1mm}

\subsection{Algorithmic system}
\label{ssec:algorithm}
The type system we defined in the previous section implements the ideas we
illustrated in the introduction and it is safe. Now the problem is to
decide whether an expression is well typed or not, that is, to find an
algorithm that given a type environment $\Gamma$ and an expression $e$
decides whether there exists a type $t$ such that $\Gamma\vdash e:t$
is provable. For that we need to solve essentially two problems:
$(i)$~how to handle the fact that it is possible to deduce several
types for the same well-typed expression and $(ii)$~how to compute the
auxiliary deduction system $ \pvdash \Gamma e t$ for paths.

$(i)$. Multiple types have two distinct origins each requiring a distinct
technical solution.  The first origin is the presence of
structural
rules%
\footnote{\label{fo:rules}In logic, logical rules refer to a
  particular connective (here, a type constructor, that is, either
  $\to$, or $\times$, or $b$), while identity rules (e.g., axioms and
  cuts) and structural rules (e.g., weakening and contraction) do
  not.\svvspace{-3.3mm}}
such as \Rule{Subs} and \Rule{Inter}. We handle this presence
in the classic way: we define an algorithmic system that tracks the
minimum type of an expression; this system is obtained from the
original system by removing the two structural rules and by
distributing suitable checks of the subtyping relation in the
remaining rules. To do that in the presence of set-theoretic types we
need to define some operators on types, which are given in
Section~\ref{sec:typeops}. The second origin is the rule \Rule{Abs-}
by which it is possible to deduce for every well-typed lambda
abstraction infinitely many types, that is the annotation of the
function intersected with as (finitely) many negations of arrow types as
possible without making the type empty. We do not handle this
multiplicity directly in the algorithmic system but only in the proof
of its soundness by using and adapting the technique of \emph{type
  schemes} defined by~\citet{Frisch2008}. Type schemes are canonical
representations of the infinite sets of types of
$\lambda$-abstractions which can be used to define an algorithmic
system that can be easily proved to be sound. The simpler algorithm
that we propose in this section implies (i.e., it is less precise than) the one with type schemes (\emph{cf}.\
Lemma~\ref{soundness_simple_ts}) and it is thus sound, too. The algorithm of this
section is not only simpler but, as we discuss in Section~\ref{sec:algoprop},
is also the one that should be used in practice. This is why we preferred to
present it here and relegate the presentation of the system with type schemes to
\Appendix\ref{app:typeschemes}.

$(ii)$. For what concerns the use of the auxiliary derivation for the $\Gamma \evdash e t \Gamma' $ and $\pvdash \Gamma e t \varpi:t'$
judgments, we present in Section~\ref{sec:typenv} an algorithm that is sound and satisfies a limited form of
completeness. All these notions are then used in the algorithmic typing
system given in Section~\ref{sec:algorules}.


\subsubsection{Operators for  type constructors}\label{sec:typeops}

In order to define the algorithmic typing of expressions like
applications and projections we need to define the operators on
types we used in Section~\ref{sec:ideas}. Consider the classic rule \Rule{App} for applications. It essentially
does three things: $(i)$ it checks that the expression in the function
position has a functional
type; $(ii)$ it checks that the argument is in the domain of the
function, and $(iii)$ it returns the type of the application. In systems
without set-theoretic types these operations are quite
straightforward: $(i)$ corresponds to checking that the expression has
an arrow type, $(ii)$ corresponds to checking that the argument is in
the domain of the arrow deduced for the function, and $(iii)$ corresponds
to returning the codomain of that same arrow. With set-theoretic types
things get more difficult, since a function can be typed by, say, a
union of intersection of arrows and negations of types. Checking that
the function has a functional type is easy since it corresponds to
checking that it has a type subtype of $\Empty{\to}\Any$. Determining
its domain and the type of the application is more complicated and needs the operators $\dom{}$ and $\circ$ we informally described in Section~\ref{sec:ideas} where we also introduced the operator $\worra{}{}$. These three operators are used by our algorithm and formally defined as:\svvspace{-0.5mm}
\begin{eqnarray}
\dom t & = & \max \{ u \alt t\leq u\to \Any\}
\\[-1mm]
\apply t s & = &\,\min \{ u \alt t\leq s\to u\}
\\[-1mm]
\worra t s  & = &\,\min\{u \alt t\circ(\dom t\setminus u)\leq \neg s\}\label{worra}
\end{eqnarray}
In short, $\dom t$ is the largest domain of any single arrow that
subsumes $t$, $\apply t s$ is the smallest codomain of an arrow type
that subsumes $t$ and has domain $s$ and $\worra t s$ was explained
before.

We need similar operators for projections since the type $t$
of $e$ in $\pi_i e$ may not be a single product type but, say, a union
of products: all we know is that $t$ must be a subtype of
$\pair\Any\Any$. So let $t$ be a type such that $t\leq\pair\Any\Any$,
then we define:\svvspace{-0.7mm}
\begin{equation}
  \begin{array}{lcrlcr}
  \bpl t & = & \min \{ u \alt t\leq \pair u\Any\}\qquad&\qquad
  \bpr t & = & \min \{ u \alt t\leq \pair \Any u\}
\end{array}\svvspace{-0.7mm}
\end{equation}
All the operators above but $\worra{}{}$ are already present in the
theory of semantic subtyping: the reader can find how to compute them
in~\cite[Section
  6.11]{Frisch2008} (see also~\citep[\S4.4]{Cas15} for a detailed description). Below we just show our new formula that computes
$\worra t s$ for a $t$ subtype of $\Empty\to\Any$. For that, we use a
result of semantic subtyping that states that every type $t$ is
equivalent to a type in disjunctive normal form and that if
furthermore $t\leq\Empty\to\Any$, then $t \simeq \bigvee_{i\in
  I}\left(\bigwedge_{p\in P_i}(s_p\to t_p)\bigwedge_{n\in
  N_i}\neg(s_n'\to t_n')\right)$ with $\bigwedge_{p\in P_i}(s_p\to
t_p)\bigwedge_{n\in N_i}\neg(s_n'\to t_n') \not\simeq \Empty$ for all
$i$ in $I$. For such a $t$ and any type $s$ then we have:\svvspace{-1.0mm}
\begin{equation}\label{worralgo}
\worra t s  =  \dom t \wedge\bigvee_{i\in I}\left(\bigwedge_{\{P\subseteq P_i\alt s\leq \bigvee_{p \in P} \neg t_p\}}\left(\bigvee_{p \in P} \neg s_p\right) \right)\svvspace{-1.0mm}
\end{equation}
The formula considers only the positive arrows of each summand that
forms $t$ and states that, for each summand, whenever you take a subset
$P$ of its positive arrows that cannot yield results in
$s$ (since $s$ does not overlap the intersection of the codomains of these arrows), then
the success of the test cannot depend on these arrows and therefore
the intersection of the domains of these arrows---i.e., the values that would precisely select that set of arrows---can be removed from $\dom t$.  The proof
that this type satisfies \eqref{worra} is given in the
\Appendix\ref{app:worra}.

\subsubsection{Type environments for occurrence typing}\label{sec:typenv}

The second ingredient necessary to the definition of our algorithmic systems is the algorithm for the
deduction of $\Gamma \evdash e t \Gamma'$, that is an algorithm that
takes as input $\Gamma$, $e$, and $t$, and returns an environment that
extends $\Gamma$ with hypotheses on the occurrences of $e$ that are
the most general that can be deduced by assuming that $e\,{\in}\,t$ succeeds. For that we need the notation $\tyof{e}{\Gamma}$ which denotes the type deduced for $e$ under the type environment $\Gamma$ in the algorithmic type system of Section~\ref{sec:algorules}.
That is, $\tyof{e}{\Gamma}=t$ if and only if $\Gamma\vdashA e:t$ is provable.

We start by defining the algorithm for each single occurrence, that is for the deduction of $\pvdash \Gamma e t \varpi:t'$. This is obtained by defining two mutually recursive functions $\constrf$ and $\env{}{}$:\svvspace{-1.3mm}
  \begin{eqnarray}
    \constr\epsilon{\Gamma,e,t} & = & t\label{uno}\\[\sk]
    \constr{\varpi.0}{\Gamma,e,t} & = & \neg(\arrow{\env {\Gamma,e,t}{(\varpi.1)}}{\neg \env {\Gamma,e,t} (\varpi)})\label{due}\\[\sk]
    \constr{\varpi.1}{\Gamma,e,t} & = & \worra{{\tyof{\occ e{\varpi.0}}\Gamma}}{\env {\Gamma,e,t} (\varpi)}\label{tre}\\[\sk]
    \constr{\varpi.l}{\Gamma,e,t} & = & \bpl{\env {\Gamma,e,t} (\varpi)}\label{quattro}\\[\sk]
    \constr{\varpi.r}{\Gamma,e,t} & = & \bpr{\env {\Gamma,e,t} (\varpi)}\label{cinque}\\[\sk]
    \constr{\varpi.f}{\Gamma,e,t} & = & \pair{\env {\Gamma,e,t} (\varpi)}\Any\label{sei}\\[\sk]
    \constr{\varpi.s}{\Gamma,e,t} & = & \pair\Any{\env {\Gamma,e,t} (\varpi)}\label{sette}\\[.8mm]
    \env {\Gamma,e,t} (\varpi) & = & {\constr \varpi {\Gamma,e,t} \wedge \tyof {\occ e \varpi} \Gamma}\label{otto}
  \end{eqnarray}\svvspace{-5mm}\\
All the functions above are defined if and only if the initial path
$\varpi$ is valid for $e$ (i.e., $\occ e{\varpi}$ is defined) and $e$
is well-typed (which implies that all $\tyof {\occ e{\varpi}} \Gamma$
in the definition are defined)%
\iflongversion%
.\footnote{Note that the definition is
  well-founded.  This can be seen by analyzing the rule
  \Rule{Case\Aa} of Section~\ref{sec:algorules}: the definition of $\Refine {e,t} \Gamma$ and
  $\Refine {e,\neg t} \Gamma$ use $\tyof{\occ e{\varpi}}\Gamma$, and
  this is defined for all $\varpi$ since the first premisses of
  \Rule{Case\Aa} states that $\Gamma\vdash e:t_0$ (and this is
  possible only if we were able to deduce under the hypothesis
  $\Gamma$ the type of every occurrence of $e$.)\svvspace{-3mm}}
\else
; the well foundness of the definition can be deduced by analysing the rule~\Rule{Case\Aa} of Section~\ref{sec:algorules}.
\fi
Each case of the definition of the $\constrf$ function corresponds to the
application of a logical rule
(\emph{cf.} definition in Footnote~\ref{fo:rules})
in
the deduction system for $\vdashp$: case \eqref{uno} corresponds
to the application of \Rule{PEps}; case \eqref{due} implements \Rule{Pappl}
straightforwardly; the implementation of rule \Rule{PAppR} is subtler:
instead of finding the best $t_1$ to subtract (by intersection) from the
static type of the argument, \eqref{tre} finds directly the best type for the argument by
applying the $\worra{}{}$ operator to the static type of the function
and the refined type of the application. The remaining (\ref{quattro}--\ref{sette})
cases are the straightforward implementations of the rules
\Rule{PPairL}, \Rule{PPairR}, \Rule{PFst}, and \Rule{PSnd},
respectively.

The other recursive function, $\env{}{}$, implements the two structural
rules \Rule{PInter} and \Rule{PTypeof} by intersecting the type
obtained for $\varpi$ by the logical rules, with the static type
deduced by the type system for the expression occurring at $\varpi$. The
remaining structural rule, \Rule{Psubs}, is accounted for by the use
of the operators $\worra{}{}$ and $\boldsymbol{\pi}_i$ in
the definition of $\constrf$.

It remains to explain how to compute the environment $\Gamma'$ produced from $\Gamma$ by the deduction system for $\Gamma \evdash e t \Gamma'$. Alas, this is the most delicate part of our algorithm.
In a nutshell, what we want to do is to define a function
$\Refine{\_,\_}{\_}$ that takes a type environment $\Gamma$, an
expression $e$ and a type $t$ and returns the best type environment
$\Gamma'$ such that $\Gamma \evdash e t \Gamma'$ holds. By the best
environment we mean the one in which the occurrences of $e$ are
associated to the largest possible types (type environments are
hypotheses so they are contravariant: the larger the type the better
the hypothesis).  Recall that in Section~\ref{sec:challenges} we said
that we want our analysis to be able to capture all the information
available from nested checks. If we gave up such a kind of precision
then the definition of $\Refinef$ would be pretty easy: it must map
each subexpression of $e$ to the intersection of the types deduced by
$\vdashp$ (i.e., by $\env{}{}$) for each of its occurrences. That
is, for each expression $e'$ occurring in $e$, $\Refine {e,t}\Gamma$
would be the type environment that maps $e'$ into $\bigwedge_{\{\varpi \alt
  \occ e \varpi \equiv e'\}} \env {\Gamma,e,t} (\varpi)$. As we
explained in Section~\ref{sec:challenges} the intersection is needed
to apply occurrence typing to expressions such as
$\tcase{(x,x)}{\pair{t_1}{t_2}}{e_1}{e_2}$ where some
expressions---here $x$---occur multiple times.

In order to capture most of the type information from nested queries
the rule \Rule{Path} allows the deduction of the type of some
occurrence $\varpi$ to use a type environment $\Gamma'$ that may
contain information about some suboccurrences of $\varpi$. On the
algorithm this would correspond to applying the $\Refinef$ defined
above to an environment that already is the result of $\Refinef$, and so on. Therefore, ideally our
algorithm should compute the type environment as a fixpoint of the
function $X\mapsto\Refine{e,t}{X}$. Unfortunately, an iteration of $\Refinef$ may
not converge. As an example, consider the (dumb) expression $\tcase
{x x}{\Any}{e_1}{e_2}$. If $x:\Any\to\Any$, then when refining the ``then'' branch, every iteration of
$\Refinef$ yields for $x$ a type strictly more precise than the type deduced in the
previous iteration (because of the $\varpi.0$ case).

The solution we adopt in practice is to bound the  number of iterations to some number $n_o$. This is obtained by the following definition of $\Refinef$\svvspace{-1mm}
\[
\begin{array}{rcl}
  \Refinef_{e,t} \eqdef (\RefineStep{e,t})^{n_o}\\[-2mm]
\text{where }\RefineStep {e,t}(\Gamma)(e') &=&  \left\{\begin{array}{ll}
        \bigwedge_{\{\varpi \alt \occ e \varpi \equiv e'\}}
        \env {\Gamma,e,t} (\varpi) & \text{if } \exists \varpi.\ \occ e \varpi \equiv e' \\
        \Gamma(e') & \text{otherwise, if } e'\in\dom\Gamma\\
        \text{undefined} & \text{otherwise}
      \end{array}\right.
\end{array}\svvspace{-1.5mm}
\]
Note in particular that $\Refine{e,t}\Gamma$  extends  $\Gamma$ with hypotheses on the expressions occurring in $e$, since
$\dom{\Refine{e,t}\Gamma}$ $=$ $\dom{\RefineStep {e,t}(\Gamma)} = \dom{\Gamma} \cup \{e' \alt \exists \varpi.\ \occ e \varpi \equiv e'\}$.

In other terms, we try to find a fixpoint of $\RefineStep{e,t}$ but we
bound our search to $n_o$ iterations. Since $\RefineStep {e,t}$ is
monotone (w.r.t.\ the subtyping pre-order extended to type environments pointwise), then
  every iteration yields a better solution.
\iflongversion
While this is unsatisfactory from a formal point of view, in practice
the problem is a very mild one. Divergence may happen only when
refining the type of a function in an application: not only such a
refinement is meaningful only when the function is typed by a union
type, but also we had to build the expression that causes the
divergence in quite an \emph{ad hoc} way which makes divergence even
more unlikely: setting an $n_o$ twice the depth of the syntax tree of
the outermost type case should be more than enough to capture all realistic
cases. For instance, all examples given in Section~\ref{sec:practical}
can be checked (or found to be ill-typed) with $n_o = 1$.
\fi

\subsubsection{Algorithmic typing rules}\label{sec:algorules}
We now have all the definitions we need for our typing algorithm%
\iflongversion%
, which is defined by the following rules.
\else%
:
\fi
\begin{mathpar}
  \Infer[Efq\Aa]
  { }
  { \Gamma, (e:\Empty) \vdashA e': \Empty }
  { \begin{array}{c}\text{\tiny with priority over}\\[-1.8mm]\text{\tiny all the other rules}\end{array}}
  \qquad
  \Infer[Var\Aa]
      { }
      { \Gamma \vdashA x: \Gamma(x) }
      { x\in\dom\Gamma}
  \svvspace{-2mm}\\
  \Infer[Env\Aa]
      { \Gamma\setminus\{e\} \vdashA e : t }
      { \Gamma \vdashA e: \Gamma(e) \wedge t }
      { \begin{array}{c}e\in\dom\Gamma \text{ and }\\[-1mm] e \text{ not a variable}\end{array}}
  \qquad
  \Infer[Const\Aa]
      { }
      {\Gamma\vdashA c:\basic{c}}
      {c\not\in\dom\Gamma}
  \svvspace{-2mm}\\
\ifsubmission\else
\end{mathpar}
\begin{mathpar}
\fi%
  \Infer[Abs\Aa]
      {\Gamma,x:s_i\vdashA e:t_i'\\ t_i'\leq t_i}
      {
      \Gamma\vdashA\lambda^{\wedge_{i\in I}\arrow {s_i} {t_i}}x.e:\textstyle\wedge_{i\in I} {\arrow {s_i} {t_i}}
      }
      {\lambda^{\wedge_{i\in I}\arrow {s_i} {t_i}}x.e\not\in\dom\Gamma}
  \svvspace{-2mm}\\
  \Infer[App\Aa]
      {
        \Gamma \vdashA e_1: t_1\\
        \Gamma \vdashA e_2: t_2\\
        t_1 \leq \arrow \Empty \Any\\
        t_2 \leq \dom {t_1}
      }
      { \Gamma \vdashA {e_1}{e_2}: t_1 \circ t_2 }
      { {e_1}{e_2}\not\in\dom\Gamma}
  \svvspace{-2mm}\\
  \Infer[Case\Aa]
        {\Gamma\vdashA e:t_0\\
        \Refine {e,t} \Gamma \vdashA e_1 : t_1\\
        \Refine {e,\neg t} \Gamma \vdashA e_2 : t_2}
        {\Gamma\vdashA \tcase {e} t {e_1}{e_2}: t_1\vee t_2}
        { \tcase {e} {t\!} {\!e_1\!}{\!e_2}\not\in\dom\Gamma}
  \svvspace{-2mm}  \\
  \Infer[Proj\Aa]
  {\Gamma \vdashA e:t\and \!\!t\leq\pair{\Any\!}{\!\Any}}
  {\Gamma \vdashA \pi_i e:\bpi_{\mathbf{i}}(t)}
  {\pi_i e{\not\in}\dom\Gamma}\hfill
  \Infer[Pair\Aa]
  {\Gamma \vdashA e_1:t_1 \and \!\!\Gamma \vdashA e_2:t_2}
  {\Gamma \vdashA (e_1,e_2):{t_1}\times{t_2}}
  {(e_1,e_2){\not\in}\dom\Gamma}
\end{mathpar}
The side conditions of the rules ensure that the system is syntax
directed, that is, that at most one rule applies when typing a term:
priority is given to \Rule{Eqf\Aa} over all the other rules and to
\Rule{Env\Aa} over all remaining logical rules. The subsumption rule
is no longer in the system; it is replaced by: $(i)$ using a union
type in \Rule{Case\Aa}, $(ii)$ checking in \Rule{Abs\Aa} that the body
of the function is typed by a subtype of the type declared in the
annotation, and $(iii)$ using type operators and checking subtyping in
the elimination rules \Rule{App\Aa,Proj\Aa}. In particular, for
\Rule{App\Aa} notice that it checks that the type of the function is a
functional type, that the type of the argument is a subtype of the
domain of the function, and then returns the result type of the
application of the two types. The intersection rule is (partially)
replaced by the rule \Rule{Env\Aa} which intersects the type deduced
for an expression $e$ by occurrence typing and stored in $\Gamma$ with
the type deduced for $e$ by the logical rules: this is simply obtained
by removing any hypothesis about $e$ from $\Gamma$, so that the
deduction of the type $t$ for $e$ cannot but end by a logical rule. Of
course, this does not apply when the expression $e$ is a variable,
since an hypothesis in $\Gamma$ is the only way to deduce the type of
a variable, which is why the algorithm reintroduces the classic rule
for variables. Finally, notice that there is no counterpart for the
rule \Rule{Abs-} and that therefore it is not possible to deduce
negated arrow types for functions. This means that the algorithmic
system is not complete as we discuss in details in the next section.

\subsubsection{Properties of the algorithmic system}\label{sec:algoprop}
\rev{
In what follow we will use $\Gamma\vdashA^{n_o} e:t$ to stress the
fact that the judgment $\Gamma\vdashA e:t$ is provable in the
algorithmic system where $\Refinef_{e,t}$ is defined as
$(\RefineStep{e,t})^{n_o}$; we will omit the index $n_o$---thus keeping
it implicit---whenever it does not matter in the context.
}

The algorithmic system above is sound with respect to the deductive one of Section~\ref{sec:static}
\begin{theorem}[Soundness]\label{th:algosound}
For every $\Gamma$, $e$, $t$, $n_o$, if $\Gamma\vdashA^{n_o}  e: t$, then $\Gamma \vdash e:t$.
\end{theorem}
The proof of this theorem (see \Appendix\ref{sec:proofs_algorithmic_without_ts}) is obtained by
defining an algorithmic system $\vdashAts$ that uses type schemes,
that is, which associates each typable term $e$ with a possibly
infinite set of types $\ts$ (in particular a $\lambda$-expression
$\lambda^{\wedge_{i\in I}\arrow {s_i} {t_i}}x.e$ will
be associated to a set of types of the form $\{s\alt
    \exists s_0 = \bigwedge_{i=1..n} \arrow {t_i} {s_i}
    \land \bigwedge_{j=1..m} \neg (\arrow {t_j'} {s_j'}).\
    \Empty \not\simeq s_0 \leq s \}$) and proving that,  if
    $\Gamma\vdashA e: t$ then  $\Gamma\vdashAts e: \ts$ with
    $t\in\ts$: the soundness of $\vdashA$ follows from the soundness
    of $\vdashAts$.

Completeness needs a more detailed explanation. The algorithmic
system $\vdashA$
is not complete w.r.t.\ the language presented in
Section~\ref{sec:syntax} because it cannot deduce negated arrow
types for functions.
However, no practical programming language with structural subtyping would implement the full
language of Section~\ref{sec:syntax}, but rather restrict all
expressions of the form $\tcase{e}{t}{e_1}{e_2}$ so that the type $t$ tested in them is either
non functional (e.g., products, integer, a record type, etc.) or it is
$\Empty\to\Any$ (i.e., the expression can just test whether $e$ returns a function
or not).\footnote{
\rev{
Of course, there exist languages in which it is
possible to check whether some value has a type that has functional
subcomponents---e.g., to test whether an object is of some class
that possesses some given methods, but that is a case of nominal rather than
structural subtyping, which in our framework corresponds to testing
whether a value has some basic type.
}
}
There are multiple reasons to impose such a restriction, the
most important ones can be summarized as follows:
\begin{enumerate}[left=0pt .. 12pt]
  \item For explicitly-typed languages it may yield conterintutive results,
    since for instance $\lambda^{\Int\to\Int}x.x\in\Bool\to\Bool$
    should fail despite the fact that identity functions maps Booleans
    to Booleans.
  \item For implicitly-typed languages it yields a semantics that
    depends on the inference algorithm, since $(\lambda
    y.(\lambda x.y))3\in 3{\to} 3$ may either fail or not according to
    whether the type deduced for the result of the expression is either
    $\Int{\to}\Int$ or $3{\to}3$ (which are both valid but incomparable).

  \item For gradually-typed languages it would yield a problematic system as
    we explain in Section~\ref{sec:gradual}.
\end{enumerate}
Now, if we apply this restriction to the language of
Section~\ref{sec:syntax}, then the algorithmic system of
section~\ref{sec:algorules} is complete. Let say that an expression $e$
is \emph{positive} if it never tests a functional type more precise
than $\Empty\to\Any$ (see
\Appendix\ref{sec:proofs_algorithmic_without_ts} for the formal
definition). Then we have:
\begin{theorem}[Completeness for Positive Expressions]
  For every type environment $\Gamma$ and \emph{positive} expression $e$, if
  $\Gamma\vdash e: t$, then there exist $n_o$ and  $t'$ such that $\Gamma\vdashA^{n_o} 
  e: t'$.
\end{theorem}\noindent
We can use the algorithmic system $\vdashAts$ defined for the proof
of Theorem~\ref{th:algosound} to give a far more precise characterization than the above
of the terms for which our algorithm is complete: positivity is a practical but rough approximation. The system $\vdashAts$
copes with negated arrow types, but it still is not
complete essentially for two reasons: $(i)$ the recursive nature of
rule \Rule{Path} and $(ii)$ the use of nested \Rule{PAppL} that yields
a precision that the algorithm loses by using type schemes in defining of \constrf{} (case~\eqref{due} is the critical
one). Completeness is recovered by $(i)$ limiting the depth of the
derivations and $(ii)$ forbidding nested negated arrows on the
left-hand side of negated arrows.\svvspace{-.7mm}
\begin{definition}[Rank-0 negation]
A derivation of $\Gamma \vdash e:t$ is \emph{rank-0 negated} if
\Rule{Abs--} never occurs in the derivation of a left premise of a
\Rule{PAppL} rule.\svvspace{-.7mm}
\end{definition}
\noindent The use of this terminology is borrowed from the ranking of higher-order
types, since, intuitively, it corresponds to typing a language in
which  in the types used in dynamic tests, a negated arrow never occurs  on the
left-hand side of another negated arrow.
\begin{theorem}[Rank-0 Completeness]
For every $\Gamma$, $e$, $t$, if $\Gamma \vdash e:t$ is derivable by a rank-0 negated derivation, then there exists $n_o$ such that $\Gamma\vdashAts^{n_o}  e: t'$ and $t'\leq t$.
\end{theorem}
\noindent This last result is only of theoretical interest since, in
practice, we expect to have only languages with positive
expressions. This is why for our implementation we use the library of CDuce~\cite{BCF03}
in which type schemes are absent and functions are typed only
by intersections of positive arrows. We present the implementation in
Section~\ref{sec:practical}, but before we study some extensions.

\section{Extensions}
\label{sec:extensions}
As we recalled in the introduction, the main application of occurrence
typing is to type dynamic languages. In this section we explore how to
extend our work to encompass three features that are necessary to type
these languages.

First, we consider record types and record expressions which, in dynamic
languages, are used to implement objects. In particular, we extend our
system to cope with typical usage patterns of objects employed in these
languages such as adding, modifying, or deleting a field, or dynamically
testing its presence to specify different behaviors.

Second, in order
to precisely type applications in dynamic languages it is
crucial to refine the type of some functions to account for their
different behaviors with specific input types. But current approaches
are bad at it: they require the programmer to explicitly specify a
precise intersection type for these functions and, even with such specifications, some common cases
fail to type (in that case the only solution is to hard-code the
function and its typing discipline into the language). We show how we
can use the work developed in the previous sections to infer precise
intersection types for functions. In our system, these functions do not
require any type annotation or just an annotation for
the function parameters, whereas some of them fail to type in
current alternative approaches even when they are given the full intersection type
specification.

Finally, to type dynamic languages it is often necessary to make
statically-typed parts of a program coexist with dynamically-typed
ones. This is the aim of gradually typed systems that we explore in
the third extension of this section.

\subsection{Record types}
\label{ssec:struct}
\iflongversion
The previous analysis already covers a large gamut of realistic
cases. For instance, the analysis already handles list data
structures, since products and recursive types can encode them as
right-associative nested pairs, as it is done in the language
CDuce (e.g., $X = \textsf{Nil}\vee(\Int\times X)$ is the
type of the lists of integers): see Code 8 in Table~\ref{tab:implem} of Section~\ref{sec:practical} for a concrete example. Even more, thanks to the presence of
union types it is possible to type heterogeneous lists whose
content is described by regular expressions on types as proposed
by~\citet{hosoya00regular}. However, this is not enough to cover
records and, in particular, the specific usage patterns in dynamic
languages of records, whose field are dynamically tested, deleted,
added, and modified. This is why we extend here our work to records,
building on the record types as they are defined in CDuce.

The extension we present in this section is not trivial. Although we use the record \emph{types} as they are
defined in CDuce we cannot do the same for CDuce record
\emph{expressions}. The reasons why we cannot use the record
expressions of CDuce and we have to define and study new ones are
twofold. On the one hand  we want to capture the typing of record field extension and field
deletion, two operation widely used in dynamic language; on the other
hand we need to have very simple expressions formed by elementary
sub-expressions, in order to limit the combinatorics of occurrence
typing. For this reason we build our records one field at a time,
starting from the empty record and adding, updating, or deleting
single fields.
\else
 Since the main application of occurrence
typing is to type dynamic languages, then it is important to show that our
work can handle records. We cannot handle them via a simple encoding into
pairs, since we need a precise type analysis for operations such
as field addition, update, and deletion. We therefore extend our approach to deal
with the records (types and expressions) as they are to be found in the aforementioned
CDuce language, where these operations can be precisely typed.
\fi

Formally, CDuce record \emph{types} can be embedded in our types by adding the
following two type constructors:\\[1.4mm]
\centerline{\(\textbf{Types} \quad t ~ ::= ~ \record{\ell_1=t \ldots \ell_n=t}{t}\alt \Undef\)}\\[1.4mm]
\rev{
where $\ell$ ranges over an infinite set of labels $\Labels$ and $\Undef$
is a special singleton type whose only value is a constant
$\undefcst$ which is not in $\Domain$ (for that it is a constant akin
to $\Omega$): as a consequence $\Undef$ and $\Any$ are distinct types,
the interpretation of the former being the constant $\undefcst$ while
the interpretation of the latter being the set of all the other values.
}
The type
$\record{\ell_1=t_1 \ldots \ell_n=t_n}{t}$ is a \emph{quasi-constant
  function} that maps every $\ell_i$ to the type $t_i$ and every other
$\ell \in \Labels$ to the type $t$ (all the $\ell_i$'s must be
distinct). Quasi constant functions are the internal representation of
record types in CDuce. These are not visible to the programmer who can use
only two specific forms of quasi constant functions, open record types and closed record types (as for OCaml object types), provided by the
following syntactic sugar:%
\iflongversion%
\footnote{Note that in the definitions ``$\ldots{}$'' is meta-syntax to denote the presence of other fields while in the open records ``{\large\textbf{..}}'' is the syntax that distinguishes them from closed ones.}
\fi
\begin{itemize}[nosep]
\item $\crecord{\ell_1=t_1, \ldots, \ell_n=t_n}$ for $\record{\ell_1=t_1 \ldots \ell_n=t_n}{\Undef}$ \hfill(closed records).
\item $\orecord{\ell_1=t_1, \ldots, \ell_n=t_n}$ for $\record{\ell_1=t_1 \ldots \ell_n=t_n}{\Any \vee \Undef}$\hfill (open records).
\end{itemize}
plus the notation $\mathtt{\ell \eqq} t$ to denote optional fields,
which corresponds to using in the quasi-constant function notation the
field $\ell = t \vee \Undef$%
\iflongversion
.
\else
\ (note that ``$\ldots{}$'' is meta-syntax while ``{\large\textbf{..}}'' is syntax).
\fi

For what concerns \emph{expressions}, we cannot use CDuce record expressions
as they are, but instead we must adapt them to our analysis. So as anticipated, we
consider records that are built starting from the empty record expression \erecord{} by adding, updating, or removing fields:\svvspace{-0.75mm}
\[
\begin{array}{lrcl}
\textbf{Expr} & e & ::= & \erecord {} ~\alt~ \recupd e \ell e ~\alt~ \recdel e \ell ~\alt~ e.\ell
\end{array}\svvspace{-.75mm}
\]
in particular $\recdel e \ell$ deletes the field $\ell$ from $e$, $\recupd e \ell e'$ adds the field $\ell=e'$ to the record $e$ (deleting any existing $\ell$ field), while $e.\ell$ is field selection with the reduction:
\(\erecord{...,\ell=e,...}.\ell\ \reduces\ e\).

To define record type subtyping and record expression type inference
we need three operators on record types: $\proj \ell t$ which returns
the type of the field $\ell$ in the record type $t$, $t_1+t_2$ which
returns the record type formed by all the fields in $t_2$ and those in
$t_1$ that are not in $t_2$, and $\recdel t\ell$ which returns the
type $t$ in which the field $\ell$ is
\iflongversion
undefined. They are formally defined as follows (see~\citet{alainthesis} for more details):
\begingroup
\allowdisplaybreaks
\begin{eqnarray}
\proj \ell t & = & \left\{\begin{array}{ll}\min \{ u \alt t\leq \orecord{\ell=u}\} &\text{if } t \leq \orecord{\ell = \Any}\\ \text{undefined}&\text{otherwise}\end{array}\right.\\[2mm]
  t_1 + t_2 & = & \min\left\{
      u \quad\bigg|\quad \forall \ell \in \Labels.\left\{\begin{array}{ll}
      \proj \ell u \geq \proj \ell {t_2} &\text{ if } \proj \ell {t_2} \leq \neg \Undef\\
      \proj \ell u \geq \proj \ell {t_1} \vee (\proj \ell {t_2} \setminus \Undef) &\text{ otherwise}
      \end{array}\right\}
    \right\}\\[2mm]
  \recdel t \ell & = & \min \left\{ u \quad\bigg|\quad  \forall \ell' \in \Labels. \left\{\begin{array}{ll}
    \proj {\ell'} u \geq \Undef &\text{ if }\ell' = \ell\\
    \proj {\ell'} u \geq \proj {\ell'} t &\text{ otherwise}
    \end{array}\right\}\right\}
\end{eqnarray}
\endgroup
\else
undefined
(see \Appendix\ref{app:recop} for the formal definition and~\citet{alainthesis} for more details).
\fi
Then two record types $t_1$ and $t_2$ are in subtyping relation, $t_1 \leq t_2$, if and only if for all $\ell \in \Labels$ we have $\proj \ell {t_1} \leq \proj \ell {t_2}$. In particular $\orecord{\!\!}$ is the largest record type.

Expressions are then typed by the following rules (already in algorithmic form).\svvspace{-.1mm}
\begin{mathpar}
    \Infer[Record]
        {~}
        {\Gamma \vdash \erecord {}:\crecord {}}
        {}
~
  \Infer[Update]
        {\Gamma \vdash e_1:t_1\and t_1\leq\orecord {\!\!} \and \Gamma \vdash e_2:t_2}
        {\Gamma \vdash \recupd{e_1}{\ell}{e_2}: t_1 + \crecord{\ell=t_2}}
        {\recupd{e_1\!\!}{\!\!\ell}{e_2} \not\in\dom\Gamma}
        \svvspace{-1.9mm}
\\        
  \Infer[Delete]
        {\Gamma \vdash e:t\and t\leq\orecord {\!\!}}
        {\Gamma \vdash \recdel e \ell: \recdel t \ell}
        {\recdel e \ell\not\in\dom\Gamma}

  \Infer[Proj]
        {\Gamma \vdash e:t\and t\leq\orecord{\ell = \Any}}
        {\Gamma \vdash e.\ell:\proj \ell {t}}
        {e.\ell\not\in\dom\Gamma}\svvspace{-2mm}
\end{mathpar}
To extend occurrence typing to records we add the following values to paths: $\varpi\in\{\ldots,a_\ell,u_\ell^1,u_\ell^2,r_\ell\}^*$, with 
\(e.\ell\downarrow a_\ell.\varpi =\occ{e}\varpi\), 
\(\recdel e \ell\downarrow r_\ell.\varpi = \occ{e}\varpi\), and
\(\recupd{e_1}\ell{e_2}\downarrow u_\ell^i.\varpi = \occ{e_i}\varpi\)
and add the following rules for the new paths:
\begin{mathpar}\svvspace{-8.7mm}
    \Infer[PSel]
        { \pvdash \Gamma e t \varpi:t'}
        { \pvdash \Gamma e t {\varpi.a_\ell}:\orecord {\ell:t'} }
        { }
        
    \Infer[PDel]
        { \pvdash \Gamma e t \varpi:t'}
        { \pvdash \Gamma e t \varpi.r_\ell: (\recdel {t'} \ell) + \crecord{\ell \eqq \Any}}
        { }
        \svvspace{-3mm}\\
    \Infer[PUpd1]
        { \pvdash \Gamma e t \varpi:t'}
        { \pvdash \Gamma e t \varpi.u_\ell^1: (\recdel {t'} \ell) + \crecord{\ell \eqq \Any}}
        { }
        
    \Infer[PUpd2]
        { \pvdash \Gamma e t \varpi:t}
        { \pvdash \Gamma e t \varpi.u_\ell^2: \proj \ell t'}
        { }
\end{mathpar}
Deriving the algorithm from these rules is then straightforward:\\[1.8mm]
\hspace*{-2mm}\(
\begin{array}{llll}
  \constr{\varpi.a_\ell}{\Gamma,e,t} =  \orecord {\ell: \env {\Gamma,e,t} (\varpi)}\hspace*{-1mm}&
  \constr{\varpi.r_\ell}{\Gamma,e,t} =  \recdel {(\env {\Gamma,e,t} (\varpi))} \ell + \crecord{\ell \eqq \Any}\\
   \constr{\varpi.u_\ell^2}{\Gamma,e,t}  =  \proj \ell {(\env {\Gamma,e,t} (\varpi))} &
 \constr{\varpi.u_\ell^1}{\Gamma,e,t}  =  \recdel {(\env {\Gamma,e,t} (\varpi))} \ell + \crecord{\ell \eqq \Any}\\[1.8mm]
\end{array}
\)

Notice that the effect of doing $\recdel t \ell + \crecord{\ell \eqq
  \Any}$ corresponds to setting the field $\ell$ of the (record) type
$t$ to the type $\Any\vee\Undef$, that is, to the type of all
undefined fields in an open record. So \Rule{PDel} and \Rule{PUpd1}
mean that if we remove, add, or redefine a field $\ell$ in an expression $e$
then all we can deduce for $e$ is that its field $\ell$ is undefined: since the original field was destroyed we do not have any information on it apart from the static one.
For instance, consider the test:\\[1mm]
\centerline{\(\ifty{\recupd{x}{a}{0}}{\orecord{a=\Int, b=\Bool}\vee \orecord{a=\Bool, b=\Int}}{x.b}{\False}\)}\\[1mm]
By  $\constr{\varpi.u_\ell^1}{\Gamma,e,t}$---i.e., by \Rule{Ext1}, \Rule{PTypeof}, and \Rule{PInter}---the type for $x$ in the positive branch is $((\orecord{a=\Int, b=\Bool}\vee \orecord{a=\Bool, b=\Int}) \land \orecord{a=\Int}) + \crecord{a\eqq \Any}$.
It is equivalent to the type $\orecord{b=\Bool}$, and thus we can deduce that $x.b$ has the type $\Bool$.

\beppe{Compare with path expressions of ~\citet{THF10} }

\subsection{Refining function types}\label{sec:refining}
\newcommand{\negspace}{\svvspace{-.5mm}}
As we explained in the introduction, both TypeScript and Flow deduce for the first definition of the function \code{foo} in~\eqref{foo} the type
\code{(number$\vee$string) $\to$ (number$\vee$string)}, while the more precise type\svvspace{-3pt}
\begin{equation}\label{tyinter}
\code{(number$\to$number)$\,\wedge\,$(string$\to$string)}\svvspace{-3pt}
\end{equation}
can be deduced by these languages only if they are instructed to do so: the
programmer has to explicitly annotate \code{foo} with the
type \eqref{tyinter}: we did it in \eqref{foo2} using Flow---the TypeScript annotation for it is much heavier. But this seems like overkill, since a simple
analysis of the body of \code{foo} in \eqref{foo} shows that its execution may have
two possible behaviors according to whether the parameter \code{x} has
type \code{number} or not (i.e., or \code{(number$\vee$string)$\setminus$number}, that is \code{string}), and this is
should be enough for the system to deduce the type \eqref{tyinter}
even in the absence the annotation given in \eqref{foo2}.
In this section we show how to do it by using the theory of occurrence
typing we developed in the first part of the paper. In particular, we
collect the different types that are assigned to the parameter of a function
in its body, and use this information to partition the domain of the function
and to re-type its body. Consider a more involved example in a pseudo
TypeScript that uses our syntax for type-cases
\begin{alltt}\color{darkblue}\morecompact
  function (x \textcolor{darkred}{: \(\tau\)}) \{
    return (x \(\in\) Real) ? ((x \(\in\) Int) ? x+1 : sqrt(x)) : !x;  \refstepcounter{equation}                           \mbox{\color{black}\rm(\theequation)}\label{foorefine}
  \}
\end{alltt}
where we assume that \code{Int} is a
subtype of \code{Real}. When $\tau$ is \code{Real$\vee$Bool}  we want to deduce for this function the
type
\code{$(\Int\to\Int)\wedge(\Real\backslash\Int\to\Real)\wedge(\Bool\to\Bool)$}.
When $\tau$ is \Any,
then the function must be rejected (since it tries to type
\code{!x} under the assumption that \code x\ has type
\code{$\neg\Real$}). Notice that typing the function under the
hypothesis that $\tau$ is \Any,
allows us to capture user-defined discrimination as defined
by~\citet{THF10} since, for instance
\begin{alltt}\color{darkblue}\morecompact
  let is_int x = (x\(\in\)Int)? true : false
   in if is_int z then z+1 else 42
\end{alltt}
is well typed since the function \code{is\_int} is given type
$(\Int\to\True)\wedge(\neg\Int\to\False)$. We propose a more general
approach than the one by~\citet{THF10} since we allow the programmer to hint a particular type for the
argument and let the system deduce, if possible, an intersection type for the
function.

We start by considering the system where $\lambda$-abstractions are
typed by a single arrow and later generalize it to the  case of
intersections of arrows. First, we define the auxiliary judgement
\(
\Gamma \vdash e\triangleright\psi
\)
where $\Gamma$ is a typing environement, $e$ an expression and $\psi$
a mapping from variables to sets of types. Intuitively $\psi(x)$ denotes
the set that contains the types of all the occurrences of $x$ in $e$. This
judgement can be deduced by the following deduction system
that collects type information on the variables that are $\lambda$-abstracted
(i.e., those in the domain of $\Gamma$, since lambdas are our only
binders):\svvspace{-1.5mm}
\begin{mathpar}
\Infer[Var]
    {
    }
    { \Gamma \vdash x \triangleright\{x \mapsto \{ \Gamma(x) \} \} }
    {}
\hfill
\Infer[Const]
{
}
{ \Gamma \vdash c \triangleright \varnothing }
{}
\hfill
\Infer[Abs]
    {\Gamma,x:s\vdash e\triangleright\psi}
    {
    \Gamma\vdash\lambda x:s.e\triangleright\psi\setminus\{x\}
    }
    {}
\svvspace{-2.3mm}\\
\Infer[App]
    {
      \Gamma \vdash e_1\triangleright\psi_1 \\
      \Gamma\vdash e_2\triangleright\psi_2
    }
    { \Gamma \vdash {e_1}{e_2}\triangleright\psi_1\cup\psi_2 }
    {}
\hfill    
  \Infer[Pair]
        {\Gamma \vdash e_1\triangleright\psi_1 \and \Gamma \vdash e_2\triangleright\psi_2}
        {\Gamma \vdash (e_1,e_2)\triangleright\psi_1\cup\psi_2}
        {}
\hfill
 \Infer[Proj]
        {\Gamma \vdash e\triangleright\psi}
        {\Gamma \vdash \pi_i e\triangleright\psi}
        {}
\svvspace{-2.3mm}\\        
\Infer[Case]
      {\Gamma\vdash e\triangleright\psi_\circ\\
        \Gamma \evdash e t \Gamma_1\\ \Gamma_1\vdash e\triangleright\psi_1\\ \Gamma_1\vdash e_1\triangleright\psi_1'\\
        \Gamma \evdash e {\neg t} \Gamma_2 \\ \Gamma_2\vdash e\triangleright\psi_2\\ \Gamma_2\vdash e_2\triangleright\psi_2'}
      {\Gamma\vdash \ifty{e}{t}{e_1}{e_2}\triangleright\psi_\circ\cup\psi_1\cup\psi_1'\cup\psi_2\cup\psi_2'}
      {}\svvspace{-1.4mm}
\end{mathpar}
Where $\psi\setminus\{x\}$ is the function defined as $\psi$ but undefined on $x$ and $\psi_1\cup \psi_2$ denotes component-wise union%
\iflongversion%
, that is :
    \begin{displaymath}
      (\psi_1\cup \psi_2)(x) = \left\{\begin{array}{ll}
                                        \psi_1 (x) & \text{if~} x\notin
                                                     \dom{\psi_2}\\
                                        \psi_2 (x) & \text{if~} x\notin
                                                     \dom{\psi_1}\\
                                        \psi_1(x)\cup\psi_2(x) & \text{otherwise}
                                      \end{array}\right.
    \end{displaymath}
\noindent
\else.~\fi
All that remains to do is to replace the rule [{\sc Abs}+] with the
following rule\svvspace{-.8mm}
\begin{mathpar}
  \Infer[AbsInf+]
  {\Gamma,x:s\vdash e\triangleright\psi
    \and
    \Gamma,x:s\vdash e:t
    \and
    T = \{ (s,t) \} \cup \{ (u,w) ~|~
      u\in\psi(x) \land \Gamma,x:u\vdash e:w \}}
    {
    \Gamma\vdash\lambda x{:}s.e:\textstyle\bigwedge_{(u,w) \in T}u\to w
    }
    {}\svvspace{-2.5mm}
  \end{mathpar}
  Note the invariant that the domain of
  $\psi$ is always conatined in the
domain of $\Gamma$ restricted to variables.
Simply put, this rule first collects all possible types that are deduced
for a variable $x$ during the typing of the body of the $\lambda$ and then uses them to re-type the body
 under this new refined hypothesis for the type of
$x$. The re-typing ensures that  the type safety property
carries over to this new rule.

This system is enough to type our case study \eqref{foorefine} for the case $\tau$
defined as \code{Real$\vee$Bool}. Indeed, the analysis of the body yields
$\psi(x)=\{\Int,\Real\setminus\Int\}$ for the branch \code{(x\,$\in$\,Int) ? x+1 : sqrt(x)} and, since
\code{$(\Bool\vee\Real)\setminus\Real = \Bool$}, yields
$\psi(x)=\{\Bool\}$ for the branch \code{!x}. So the function
will be checked for the input types $\Int$, $\Real\setminus\Int$, and
\Bool, yielding the expected result.

It is not too difficult to generalize this rule when the lambda is
typed by an intersection type:\svvspace{-.8mm}
\begin{mathpar}
  \Infer[AbsInf+] {\forall i\in I\hspace*{2mm}\Gamma,x:s_i\vdash
    e\triangleright\psi_i
    \and
    \Gamma,x:s_i\vdash
    e : t_i
    \and
    T_i = \{ (u, w) ~|~ u\in\psi_i(x) \land \Gamma, x:u\vdash e : w\}
  } {\textstyle \Gamma\vdash\lambda^{\bigwedge_{i\in
        I}s_i\to t_i} x.e:\bigwedge_{i\in I}(s_i\to
        t_i)\land\bigwedge_{(u, w)\in T_i}(u\to w) } {}\svvspace{-3mm}
\end{mathpar}
For each arrow declared in the interface of the function, we
first typecheck the body of the function as usual (to check that the
arrow is valid) and collect the refined types for the parameter $x$.
Then we deduce all possible output types for this refined set of input
types and add the resulting arrows to the type deduced for the whole
function (see
\iflongversion
Section~\ref{sec:practical}
\else
\Appendix\ref{app:optimize}
\fi
for an even more precise rule).

\kim{We define the rule on the type system not the algorithm. I think
  we could do the same (by collecting type scheme) in $\psi$ but we
  then need to choose a candidate in the type scheme \ldots }

In summary, in order to type a
function we use the type-cases on its parameter to partition the
domain of the function and we type-check the function on each single partition rather
than on the union thereof. Of course, we could use much a finer
partition: the finest (but impossible) one is to check the function
against the singleton types of all its inputs. But any finer partition
would return, in many cases,  not a much better information, since most
partitions would collapse on the same return type: type-cases on the
parameter are the tipping points that are likely to make a difference, by returning different
types for different partitions thus yielding more precise typing.

Even though type cases in the body of a
function are tipping points that may change the type of the result
of the function, they are not the only ones: applications of overloaded functions play exactly the same role. We
therefore add to our deduction system a last further  rule:\\[2mm]
\centerline{\(
\Infer[OverApp]
    {
      \Gamma \vdash e : \textstyle\bigvee \bigwedge_{i \in I}t_i\to{}s_i\\
      \Gamma \vdash x : t\\
      \Gamma \vdash e\triangleright\psi_1\\
      \Gamma \vdash x\triangleright\psi_2\\
    }
    { \Gamma \vdash\textstyle
      {e}{~x}\triangleright\psi_1\cup\psi_2\cup\bigcup_{i\in I}\{
      x\mapsto t\wedge t_i\} }
    {(t\wedge t_i\not\simeq\Empty)}
    \)}\\[2mm]
Whenever a function parameter is the argument of an
overloaded function, we record as possible types for this parameter
all the domains $t_i$ of the arrows that type the overloaded
function, restricted (via intersection) by the static type $t$ of the parameter and provided that the type is not empty ($t\wedge t_i\not\simeq\Empty$). We show the remarkable power of this rule on some practical examples in Section~\ref{sec:practical}.

\ignore{\color{darkred} RANDOM THOUGHTS:
A delicate point is the definition of the union
$\psi_1\cup\psi_2$. The simplest possible definition is the
component-wise union. This definition is enough to avoid the problem
of typecases on casted values and it is also enough to type our case
study, as we showed before. However if we want a more precise typing
discipline we may want to consider a more sophisticated way of
combining the information collected by the various $\psi$. For
instance, consider the following pair
\[\code{( \ite x s {e_1}{e_2} , \ite x t {e_3}{e_4} )}\]
Then we have $x:\Any\vdash \code{( \ite x s {e_1}{e_2} , \ite x t
  {e_3}{e_4} )}\triangleright x\mapsto\{s,\neg s, t, \neg t\}$. However if this
code is the body of a function with parameter $x$, then it may be
tempting to try to produce a finer-grained analysis: for example,
instead of checking as input type just $s$ and $t$ one could check
instead $s\setminus t$, $t\setminus s$, and $s\wedge t$, whenever these
three types are not empty. This can be obtained by defining the union
operation as follows:
\[ (\psi_0\cup\psi_1)(x)=\{ t \alt  \exists t_1\in\psi_1(x), t_2\in\psi_2(x), t=t_1\setminus t_2  \text{ or }  t=t_2\setminus t_1 \text{ or }  t=t_1\wedge t_2\text{ and } t\not\simeq\Empty\}\]
Do we really gain in precision? I think the gain is minimum. All we may obtain just come from a polymorphic use of the variable, but we can hardly gain more. Probably it is not worth the effort. As a concrete case consider
\[\code{function  x \{ ( \ite x {\texttt{String|Bool}} {x}{x} , \ite x {\texttt{Bool|Int}} {x}{x} )} \}\]

So what?

A simple solution would be to define the union as follows
\begin{equation}
    (\psi_0\cup\psi_1)(x)=
    \left\{\begin{array}{ll}
    \psi_i(x) &\text{ if }\psi_i(x)\subseteq\psi_{((i+1)\,\text{mod}\,2)}(x)\\
    \psi_0(x)\cup\psi_1(x) &\text{ otherwise}
    \end{array}\right.
  \end{equation}
   and  $\psi_1(x)\subseteq\psi_2(x)\iff \forall\tau\in\psi_1(x),\exists \tau'\in\psi_2(x),\tau\leq\tau'$

This would be interesting to avoid to use as domains those in
$\psi_\circ$ that would be split in two in the $\psi_1$ and $\psi_2$
of the ``if'' ... but it is probably better to check it by modifying
the rule of ``if'' (that is, add $\psi_\circ$ only for when it brings new
information).
}

\subsection{Integrating gradual typing}\label{sec:gradual}
Gradual typing is an approach proposed by~\citet{siek2006gradual} to
combine the safety guarantees of static typing with the programming
flexibility of dynamic typing. The idea is to introduce an \emph{unknown} 
(or \emph{dynamic}) type, denoted $\dyn$, used to inform the compiler that
some static type-checking can be omitted, at the cost of some additional
runtime checks. The use of both static typing and dynamic typing in a same
program creates a boundary between the two, where the compiler automatically
adds---often costly~\cite{takikawa2016sound}---dynamic type-checks to ensure 
that a value crossing the barrier is correctly typed.

Occurrence typing and gradual typing are two complementary disciplines
which have a lot to gain to be integrated, although we are not aware
of any study in this sense. We explore this integration for the
formalism of Section~\ref{sec:language} for which the integration of
gradual typing was first defined by~\citet{CL17} and sucessively considerably
improved by~\citet{castagna2019gradual} (see~\citet{Lanvin21phd} for a
comprehensive presentation).

In a sense, occurrence typing is a
discipline designed to push forward the frontiers beyond which gradual
typing is needed, thus reducing the amount of runtime checks needed. For 
instance, the JavaScript code of~\eqref{foo} and~\eqref{foo2} in the introduction can also be
typed by using gradual typing:
\begin{alltt}\color{darkblue}\morecompact
  function foo(x\textcolor{darkred}{ : \pmb{\dyn}}) \{
      return (typeof(x) === "number")? x+1 : x.trim();  \refstepcounter{equation}                                \mbox{\color{black}\rm(\theequation)}\label{foo3}
  \}\negspace
\end{alltt}
``Standard'' or ``safe'' gradual typing inserts two dynamic checks since it compiles the code above into:
\begin{alltt}\color{darkblue}\morecompact
  function foo(x) \{
      return (typeof(x) === "number")? (\textcolor{darkred}{\Cast{number}{{\color{darkblue}x}}})+1 : (\textcolor{darkred}{\Cast{string}{{\color{darkblue}x}}}).trim();
  \}\negspace
\end{alltt}
where {\Cast{$t$}{$e$}} is a type-cast that dynamically checks whether the value returned by $e$ has type $t$.\footnote{Intuitively, \code{\Cast{$t$}{$e$}} is
  syntactic sugar for \code{(typeof($e$)==="$t$")\,?\,$e$\,:\,(throw "Type
    error")}. Not exactly though, since to implement compilation \emph{à la} sound gradual typing it is necessary to use casts on function types that need special handling.}
We already saw  that thanks to occurrence typing we can annotate the parameter \code{x} by \code{number|string} instead of \dyn{} and avoid the insertion of any cast. 
But occurrence typing can be used also on the gradually typed code above in order to statically detect the insertion of useless casts. Using
occurrence typing to type the gradually-typed version of \code{foo} in~\eqref{foo3}, allows the system to avoid inserting the first cast
\code{\Cast{number}{x}} since, thanks to occurrence typing, the
occurrence of \code{x} at issue is given type \code{number} (but the
second cast is still necessary though). But removing only this cast is far
from being satisfactory, since when this function is applied to an integer
there are some casts that still need to be inserted outside  the function.
The reason is that the compiled version of the function
has type \code{\dyn$\to$number}, that is, it expects an argument of type
\dyn, and thus we have to apply a cast (either to the argument or
to the function) whenever this is not the case. In particular, the
application \code{foo(42)} will be compiled as
\code{foo(\Cast{\dyn}{42})}. Now, the main problem with such a cast is not
that it produces some unnecessary overhead by performing useless
checks (a cast to \dyn{} can easily be detected and safely ignored at runtime). 
The main problem is that the combination of such a cast with type-cases 
will lead to unintuitive results under the standard operational
semantics of type-cases and casts.
Indeed, consider the standard semantics
of the type-case \code{(typeof($e$)==="$t$")} which consists in
reducing $e$ to a value and checking whether the type of the value is a
subtype of $t$. In standard gradual semantics, \code{\Cast{\dyn}{42}} is a value. 
And this value is of type \code{\dyn}, which is not a subtype of \code{number}. 
Therefore the check in \code{foo} would fail for \code{\Cast{\dyn}{42}}, and so
would the whole function call.
Although this behavior is type safe, this violates the gradual 
guarantee~\cite{siek2015refined} since giving a \emph{more precise} type to
the parameter \code{x} (such as \code{number}) would make the function succeed,
as the cast to \code{$\dyn$} would not be inserted.
A solution is to modify the semantics of type-cases, and in particular of 
\code{typeof}, to strip off all the casts in values, even nested ones.
While this adds a new overhead at runtime, this is preferable to losing the
gradual guarantee, and the overhead can be mitigated by having a proper
representation of cast values that allows to strip all casts at once.

However, this problem gets much more complex when considering functional values.
In fact, as we hinted in Section~\ref{ssec:algorithm}, there is no way to
modify the semantics of type cases to preserve both the gradual guarantee and
the soundness of the system in the presence of arbitrary type cases.
For example, consider the function
$f = \lambda^{(\Int \to \Int) \to \Int} g. \tcase{g}{(\Int \to \Int)}{g\
1}{\code{true}}$. This function is well-typed since the type of the
parameter guarantees that only the first branch can be taken, and thus that
only an integer can be returned. However, if we apply this function
to $h = \MCast{\Int \to \Int}{(\lambda^{\dyn \to \dyn} x.\ x)}$, the type case
strips off the cast around $h$ (to preserve the gradual guarantee), then
checks if $\lambda^{\dyn \to \dyn} x.\ x$ has type $\Int \to \Int$. 
Since $\dyn \to \dyn$ is not a subtype of $\Int \to \Int$, the check fails
and the application returns $\code{true}$, which is unsound.
Therefore, to preserve soundness in the presence of gradual types, type cases
should not test functional types other than $\Empty \to \Any$, which is
the same restriction as the one presented by~\citet{siek2016recursive}.

While this solves the problem of the gradual guarantee, it is clear that
it would be much better if the application \code{foo(42)} were compiled as is,
without introducing the cast \code{\Cast{\dyn}{42}}, thus getting rid of the
overhead associated with removing this cast in the type case.
This is where the previous section about refining function types comes in handy.
To get rid of all superfluous casts, we have to fully exploit the information 
provided to us by occurrence typing and deduce for the function in~\eqref{foo3} the type
\code{(number$\to$number)$\wedge$((\dyn\textbackslash
  number)$\to$string)}, so that no cast is inserted when the
function is applied to a number. 
To achieve this, we simply modify the typing rule for functions that we defined
in the previous section to accommodate for gradual typing. Let $\sigma$ and $\tau$ range over \emph{gradual types}, that is the types produced by the grammar in Definition~\ref{def:types} to which we add \dyn{} as basic type (see~\citet{castagna2019gradual} for the definition of the subtyping relation on these types). For every gradual type
$\tau$, define $\tauUp$ as the (non gradual) type obtained from $\tau$ by replacing all
covariant occurrences of \dyn{} by \Any{} and all contravariant ones by \Empty. The
type $\tauUp$ can be seen as the \emph{maximal} interpretation of $\tau$, that is,
every expression that can safely be cast to $\tau$ is of type $\tauUp$. In
other words, if a function expects an argument of type $\tau$ but can be 
typed under the hypothesis that the argument has type $\tauUp$, then no casts
are needed, since every cast that succeeds will be a subtype of
$\tauUp$. Taking advantage of this property, we modify the rule for
functions as: \svvspace{-2mm}
%
\[
  \textsc{[AbsInf+]}
  \frac
  {
    \begin{array}{c}
    \hspace*{-8mm}T = \{ (\sigma', \tau') \} \cup \{ (\sigma,\tau) ~|~ \sigma \in \psi(x) \land \Gamma, x: \sigma \vdash e: \tau \} \cup \{ (\sigmaUp,\tau) ~|~ \sigma \in \psi(x) \land \Gamma, x: \sigmaUp \vdash e: \tau \}\\
    \Gamma,x:\sigma'\vdash e\triangleright\psi \qquad \qquad \qquad\Gamma,x:\sigma'\vdash e:\tau'
    \end{array}
  }
  {
    \Gamma\vdash\lambda x:\sigma'.e:\textstyle\bigwedge_{(\sigma,\tau) \in T}\sigma\to \tau
  }\svvspace{-2mm}
\]
The main idea behind this rule is the same as before: we first collect all the
information we can into $\psi$ by analyzing the body of the function. We then
retype the function using the new hypothesis $x : \sigma$ for every 
$\sigma \in \psi(x)$. Furthermore, we also retype the function using the hypothesis
$x : \sigmaUp$: as explained before the rule, whenever this typing suceeds it eliminates unnecessary gradual types and, thus, unecessary casts.
Let us see how this works on the function \code{foo} in \eqref{foo3}. First, we deduce
the refined hypothesis 
$\psi(\code x) = \{\,\code{number}{\land}\dyn\;,\;\dyn \textbackslash \code{number}\,\}$.
Typing the function using this new hypothesis but without considering the
maximal interpretation would yield
$(\dyn \to \code{number}\vee\code{string}) \land ((\code{number} \land \dyn) \to \code{number})
\land ((\dyn \textbackslash \code{number}) \to \code{string})$. However, as
we stated before, this would introduce an unnecessary cast if the function 
were to be applied to an integer.\footnote{%
Notice that considering
$\code{number} \land \dyn\simeq \code{number}$ is
not an option, since it would force us to choose between having
the gradual guarantee or having, say,
$\code{number} \land \code{string}$ be more precise than
$\code{number} \land \dyn$.\svvspace{-2mm}}
Hence the need for the second part of
Rule~\textsc{[AbsInf+]}: the maximal interpretation of $\code{number} \land \dyn$
is $\code{number}$, and it is clear that, if $\code x$ is given type \code{number},
the function type-checks, thanks to occurrence typing. Thus, after some
routine simplifications, we can actually deduce the desired type
$(\code{number} \to \code{number}) \land ((\dyn \textbackslash \code{number}) \to \code{string})$.

\beppe{Problem: how to compile functions with intersection types, since their body is typed several distinct types. I see two possible solutions: either merge the casts of the various typings (as we did in the compilation of polymorphic functions for semantic subtyping) or allow just one gradual type in the intersection when a function is explicitly typed (reasonable since, why would you use more gradual types in an intersection?)}

\section{Implementation}
\label{sec:practical}
\rev{We present in this section preliminary results obtained by our
implementation. After giving some technical highlights, we focus on
demonstrating the behavior of our typing algorithm on meaningful
examples. We also provide an in-depth comparison with the fourteen
examples of \cite{THF10}}.

\subsection{Implementation details}

We have implemented the algorithmic system $\vdashA$ we presented in Section~\ref{sec:algorules}. Besides the type-checking
algorithm defined on the base language, our implementation supports
the record types and expressions of Section \ref{ssec:struct} and the refinement of
function types 
\iflongversion
described in Section \ref{sec:refining}. Furthermore, our implementation uses for the inference of arrow types
the following improved rule:
\[\begin{array}{l}
   {\small\Rule{AbsInf++}}\\
  \frac
  {  \begin{array}{c}
    T = \{ (s\setminus\bigvee_{s'\in\psi(x)}s',t) \} \cup \{ (s',t') ~|~
      s'\in\psi(x) \land \Gamma,x:s'\vdash e:t' \}\\
\textstyle\Gamma,x:s\vdash e\triangleright\psi
    \qquad
    \Gamma,x:s\setminus\bigvee_{s'\in\psi(x)}s'\vdash e:t
    \end{array}
    }
    {
    \textstyle\Gamma\vdash\lambda x{:}s.e:{\bigwedge_{(s',t') \in T}s'\to t'}
    }
\end{array}\]
instead of the simpler \Rule{AbsInf+} given in
Section \ref{sec:refining}. The difference of this new rule with
respect to \Rule{AbsInf+} is that the typing of the body
is made under the hypothesis $x:s\setminus\bigvee_{s'\in\psi(x)}s'$,
that is, the domain of the function minus all the input types
determined by the $\psi$-analysis. This yields an even better refinement
of the function type that makes a difference for instance with the
inference for the function \code{xor\_} (see Code 3
in Table~\ref{tab:implem}): the old rule would have returned a less precise type. The rule above is defined for functions annotated by a single arrow type:
the extension to annotations with intersections of multiple arrows is similar to the one we did in the
simpler setting of Section~\ref{sec:refining}.

\else
of Section \ref{sec:refining} with the rule of
Appendix~\ref{app:optimize}.
\fi

The implementation is rather crude and consists of 2000 lines of OCaml code,
including parsing, type-checking of programs, and pretty printing of
types. \rev{CDuce is used as a library to provide set-theoretic types and
semantic subtyping. The implementation faithfully transcribes in OCaml
the algorithmic system $\vdashA$ as well as all the type operations
defined in this work. One optimization that our implementation
features (with respect to the formal presentation) is the use of a
memoization environment in the code of the $\Refine {e,t}{\Gamma}$
function, which allows the inference to avoid unnecessary traversals
of $e$.
Lastly, while our prototype allows the user to specify a particular
value for the $n_o$ parameter we introduced in
Section~\ref{sec:typenv}, a value of $1$ for $n_o$ is sufficient
to check all examples we present in the rest of the section.
}

\subsection{Experiments}
We
demonstrate the output of our type-checking implementation in
Table~\ref{tab:implem} and Table~\ref{tab:implem2}. Table~\ref{tab:implem} lists
some examples, none of which can be typed by current systems. Even though some
systems such as Flow and TypeScript can type some of these examples by adding
explicit type annotations, the code 6, 7, 9, and 10 in Table~\ref{tab:implem}
and, even more,  the \code{and\_} and \code{xor\_} functions given
in \eqref{and+} and \eqref{xor+} later in this
section are out of reach of current systems, even when using the right explicit
annotations. 

It should be noted that for all the examples we present, the
time for the type inference process is less than 5ms, hence we do not report
precise timings in the table. These and other examples can be tested in the
online toplevel available at
\url{https://occtyping.github.io/}%
\lstset{language=[Objective]Caml,columns=fixed,basicstyle=\linespread{0.43}\ttfamily\scriptsize,aboveskip=-0.5em,belowskip=-1em,xleftmargin=-0.5em}
\begin{table}
   {\scriptsize
  \begin{tabular}{|@{\,}c@{\,}|p{0.53\textwidth}@{}|@{\,}p{0.42\textwidth}@{\,}|}
\hline
    & Code & Inferred type\\
    \hline
    1 &
          \begin{lstlisting}
let basic_inf = fun (y : Int | Bool) ->
  if y is Int then incr y else lnot y\end{lstlisting}
         &\vfill
           $(\Int\to\Int)\land(\Bool\to\Bool)$
    \\\hline
    2 &
\begin{lstlisting}
let any_inf = fun (x : Any) ->
  if x is Int then incr x else
  if x is Bool then lnot x else x
\end{lstlisting} &\vfill
 $(\Int\to\Int)\land(\lnot\Int\to\lnot\Int)\;\land$\newline
 $(\Bool\to\Bool)\land(\lnot(\Int\vee\Bool)\to\lnot(\Int\vee\Bool))$\\
    \hline

    3 &\begin{lstlisting}
let is_int = fun (x : Any) ->
 if x is Int then true else false

let is_bool = fun (x : Any) ->
 if x is Bool then true else false

let is_char = fun (x : Any) ->
 if x is Char then true else false
\end{lstlisting}
&\smallskip
  $(\Int\to\Keyw{True})\land(\lnot\Int\to\Keyw{False})$\newline
  ~\newline\smallskip
  $(\Bool\to\Keyw{True})\land(\lnot\Bool\to\Keyw{False})$\newline
  ~\newline
  $(\Char\to\Keyw{True})\land(\lnot\Char\to\Keyw{False})$
\\\hline
    4 &
          \begin{lstlisting}
let not_ = fun (x : Any) ->
   if x is True then false else true
 \end{lstlisting} &\vfill
 $(\Keyw{True}\to\Keyw{False})\land(\lnot\Keyw{True}\to\Keyw{True})$\\\hline
    5 &
          \begin{lstlisting}
let or_ = fun (x : Any) -> fun (y: Any) ->
 if x is True then true
 else if y is True then true else false
\end{lstlisting}
 &\vfill
   $(\True\to\textsf{Any}\to\True)\land(\lnot\Keyw{True}\to\True\to\True)\;\land$\newline
   $(\lnot\True\to\lnot\True\to\False)$
\\\hline
    6 &
          \begin{lstlisting}
let and_ = fun (x : Any) -> fun (y : Any) ->
  if not_ (or_ (not_ x) (not_ y)) is True
  then true else false
\end{lstlisting}
 &\vfill
   $(\True\to((\lnot\True\to\False)\land(\True\to\True))$\newline
   $ \land(\lnot\True\to\textsf{Any}\to\False)$
\\\hline
    7 &
\begin{lstlisting}
let f = fun (x : Any) -> fun (y : Any) ->
  if and_ (is_int x) (is_bool y) is True
  then 1 else
     if or_ (is_char x) (is_int y) is True
     then 2 else 3
\end{lstlisting}&
$(\Int \to (\Int \to 2) \land (\lnot\Int \to 1 \lor 3) \land (\Bool \to 1) \land$\newline
 \hspace*{1cm}$(\lnot(\Bool \lor \Int) \to 3) \land
          (\lnot\Bool \to 2\lor3))$~~$\land$\newline
$(\Char \to (\Int \to 2) \land (\lnot\Int \to 2) \land (\Bool \to 2) \land$\newline
\hspace*{1cm}$(\lnot(\Bool \lor \Int) \to 2) \land
          (\lnot\Bool \to 2))$~~$\land$\newline
$(\lnot(\Int \lor \Char) \to (\Int \to 2) \land
          (\lnot\Int \to 3) \land$\newline
\hspace*{0.2cm}$(\Bool \to 3) \land(\lnot(\Bool \lor \Int) \to 3)
          \land (\lnot\Bool \to 2
                  \lor 3))$\newline
                  $\land \ldots$ (two other redundant cases omitted)
\\
& \begin{lstlisting}
let test_1 = f 3 true
let test_2 = f (42,42) 42
let test_3 = f nil nil
\end{lstlisting}&
                  $\p 1$\newline
                  $\p 2$\newline
                  $\p 3$
    \\\hline
    8 &
\begin{lstlisting}
atom nil
type Document = { nodeType=9 ..}
and Element = { nodeType=1, childNodes=NodeList ..}
and Text = { nodeType=3,
              isElementContentWhiteSpace=Bool ..}
and Node = Document | Element | Text
and NodeList = Nil | (Node, NodeList)

let is_empty_node = fun (x : Node) ->
  if x.nodeType is 9 then false
  else if x is { nodeType=3 ..} then
    x.isElementContentWhiteSpace
  else if x.childNodes is Nil then true else false
\end{lstlisting} &\vspace{12mm}
$(\Keyw{Document}\to\False)~\land$\newline
$(\orecord{ \texttt{nodeType}\,{=}\,1, \texttt{childNodes}\,{=}\,\Keyw{Nil} }\to\True) ~\land$\newline
$(\orecord{ \texttt{nodeType}\,{=}\,1,
                   \texttt{childNodes}\,{=}\,(\Keyw{Node},\Keyw{NodeList})
                   }\to\False) ~\land$\newline
$(\Keyw{Text}\to\Bool)~\land\ldots$ (omitted redundant arrows)
    \\\hline
    9 & \begin{lstlisting}
let xor_ = fun (x : Any) -> fun (y : Any) ->
 if and_ (or_ x y) (not_ (and_ x y)) is True
 then true else false
\end{lstlisting} &\vfill
$\True\to((\True\to\False)\land(\lnot\True\to\True))~\land$\newline
$(\lnot\True\to((\True\to\True) \land (\lnot\True\to\False))$
    \\\hline
 10 & \begin{lstlisting}
(* f, g have type: (Int->Int) & (Any->Bool) *)
let example10 = fun (x : Any) ->
  if (f x, g x) is (Int, Bool) then 1 else 2
\end{lstlisting} &\vfill
$(\Int\to\textsf{Empty})\land(\neg\Int\to{}2)$\newline
\texttt{Warning: line 4, 39-40: unreachable expression}
    \\\hline
11 & \begin{lstlisting}
let typeof = fun (x:Any) -> 
    if x is Int then "number"
    else if x is Char then "string"
    else if x is Bool then "boolean" else "object"

let test = fun (x:Any) ->
    if typeof x is "number" then incr x
    else if typeof x is "string" then charcode x
    else if typeof x is "boolean" then int_of_bool x
    else 0\end{lstlisting} &\vfill\smallskip
  $(\Int\to\textsf{"number"}) \wedge$\newline
  $(\Char\to\textsf{"string"})\wedge$\newline
  $(\Bool\to\textsf{"boolean"})\wedge$\newline
  $(\lnot(\Bool{\vee}\Int{\vee}\Char)\to\textsf{"object"})\wedge  \ldots$\newline
      (two other redundant cases omitted)
  \newline~\newline
  $(\Int \to \Int) \wedge (\Char \to \Int)  \wedge
  (\Bool \to \Int) \wedge $\newline
  $(\lnot(\Bool{\vee} \Int {\vee} \Char) \to 0)\wedge  \ldots$\newline
      (two other redundant cases omitted) 
      \\\hline
  12 & \begin{lstlisting}
atom null
type Object = Null | { prototype = Object ..}
type ObjectWithPropertyL = { l = Any ..}
   | { prototype = ObjectWithPropertyL ..}

let has_property_l = fun (o:Object) ->
    if o is ObjectWithPropertyL then true else false

let has_own_property_l = fun (o:Object) ->
    if o is { l=Any ..} then true else false

let get_property_l = fun (self:Object->Any) o ->
    if has_own_property_l o is True then o.l
    else if o is Null then null
    else self (o.prototype)
\end{lstlisting} &\vfill\medskip\smallskip
$(\Keyw{ObjectWithPropertyL}\to\True)$ $\land$
$(\Keyw{X1}\to\False) \texttt{ where}$\newline
$\Keyw{X1}\,=\,(\Keyw{Nil}\,|\,\orecord{\texttt{l}\,{=}\,?\Keyw{Empty},\, \texttt{prototype}\,{=}\,\Keyw{X1}})$\newline
~\newline
$(\orecord{\Keyw{l}\,=\,\Keyw{Any}, \texttt{prototype}\,=\,\Keyw{Object}}\to\True)$ $\land$\newline
$((\Keyw{Nil}\,|\,\orecord{\texttt{l}\,=\,?\Keyw{Empty},\, \texttt{prototype}\,=\,\Keyw{Object}})\to\False)$\newline
\newline
~\newline
$\Keyw{Object}\to\Keyw{Any}$
    \\\hline

\end{tabular}
}
\caption{Types inferred by the implementation}
\ifsubmission%
\svvspace{-10mm}
\fi%
\label{tab:implem}
\end{table}

In Table~1, the second column gives a code fragment and the third
column the type deduced by our implementation as is (we
pretty printed it but we did not alter the output). Code~1 is a
straightforward function similar to our introductory
example \code{foo} in (\ref{foo}) and (\ref{foo2}) where \code{incr}
is the successor function and \code{lneg} the logical negation for
Booleans. Here the
programmer annotates the parameter of the function with a coarse type
$\Int\vee\Bool$. Our implementation first type-checks the body of the
function under this assumption, but doing so it collects that the type of
$\code{x}$ is specialized to \Int{} in the ``then'' case and to \Bool{}
in the ``else'' case. The function is thus type-checked twice more
under each hypothesis for \code{x}, yielding the precise type
$(\Int\to\Int)\land(\Bool\to\Bool)$. Note that w.r.t.\
rule \Rule{AbsInf+} of Section~\ref{sec:refining}, the
rule  \Rule{AbsInf++} we use in the  implementation improves the output of the computed
type. Indeed, using rule~[{\sc AbsInf}+] we would have obtained the
type
$(\Int\to\Int)\land(\Bool\to\Bool)\land(\Bool\vee\Int\to\Bool\vee\Int)$
with a redundant arrow. Here we can see that, since we deduced
the first two arrows $(\Int\to\Int)\land(\Bool\to\Bool)$, and since
the union of their domain exactly covers the domain of the third arrow, then
the latter is not needed. Code~2 shows what happens when the argument
of the function is left unannotated (i.e., it is annotated by the top
type \Any, written ``\code{Any}'' in our implementation). Here
type-checking and refinement also work as expected, but the function
only type checks if all cases for \code{x} are covered (which means
that the function must handle the case of inputs that are neither in \Int{}
nor in \Bool).

The following examples paint a more interesting picture. First
(Code~3) it is
easy in our formalism to program type predicates such as those
hard-coded in the $\lambda_{\textit{TR}}$ language of \citet{THF10}. Such type
predicates, which return \code{true} if and only if their input has
a particular type, are just plain functions with an intersection
type inferred by the system of Section~\ref{sec:refining}. We next define Boolean connectives as overloaded
functions. The \code{not\_} connective (Code~4) just tests whether its
argument is the Boolean \code{true} by testing that it belongs to
the singleton type \True{} (the type whose only value is
\code{true}) returning \code{false} for it and \code{true} for
any other value (recall that $\neg\True$ is equivalent to
$\texttt{Any\textbackslash}\True$). It works on values of any type,
but we could restrict it to Boolean values by simply annotating the
parameter by \Bool{} (which, in the CDuce's types that our system uses,  is syntactic sugar for
\True$\vee$\False) yielding the type
$(\True{\to}\False)\wedge(\False{\to}\True)$.
The \code{or\_} connective (Code~5) is straightforward as far as the
code goes, but we see that the overloaded type precisely captures all
possible cases: the function returns \code{false} if and only if both
arguments are of type $\neg\True$, that is, they are any value
different from \code{true}. Again we use a generalized version of the
\code{or\_} connective that accepts and treats any value that is not
\code{true} as \code{false} and again, we could easily restrict the
domain to \Bool{} if desired.\\
\indent
To showcase the power of our type system, and in particular of
the ``$\worra{}{}$''
type operator, we define \code{and\_} (Code~6) using De Morgan's
Laws instead of
using a direct definition. Here the application of the outermost \code{not\_} operator is checked against type \True. This
allows the system to deduce that the whole \code{or\_} application
has type \False, which in turn leads to \code{not\_\;x} and
\code{not\_\;y} to have type $\lnot \True$ and therefore both \code{x}
and \code{y} to have type \True. The whole function is typed with
the most precise type (we present the type as printed by our
implementation, but the first arrow of the resulting type is
equivalent to
$(\True\to\lnot\True\to\False)\land(\True\to\True\to\True)$).

All these type predicates and Boolean connectives can be used together
to write complex type tests, as in Code~7. Here we define a function
\code{f} that takes two arguments \code{x} and \code{y}. If
\code{x} is an integer and \code{y} a Boolean, then it returns the
integer \code{1}; if \code{x} is a character or
\code{y} is an integer, then it returns \code{2}; otherwise the
function returns \code{3}. Our system correctly deduces a (complex)
intersection type that covers all cases (plus several redundant arrow
types). That this type is as precise as possible can be shown by the fact that
when applying
\code{f} to arguments of the expected type, the \emph{type
statically deduced} for the
whole expression is the singleton type \code{1}, or \code{2},
or \code{3}, depending on the type of the arguments.

Code~8 allows us to demonstrate the use and typing of record paths. We
model, using open records, the type of DOM objects that represent XML
or HTML documents. Such objects possess a common field
\code{nodeType} containing an integer constant denoting the kind of
the node (e.g., \p{1} for an element node, \p{3} for a text node, \ldots). Depending on the kind, the object will have
different fields and methods. It is common practice to perform a test
on the value of the \code{nodeType} field. In dynamic languages such
as JavaScript, the relevant field can directly be accessed
after having checked for the appropriate \code{nodeType}, whereas
in statically typed languages such as Java, a downward cast
from the generic \code{Node} type to the expected precise type of
the object is needed. We can see that using the record expressions presented in
Section~\ref{ssec:struct} we can deduce the correct type for
\code{x} in all cases. Of particular interest is the last case,
since we use a type case to check the emptiness of the list of child
nodes. This splits, at the type level, the case for the \Keyw{Element}
type depending on whether the content of the \code{childNodes} field
is the empty list or not.

Code~9 shows the usefulness of the rule \Rule{OverApp}.
Consider the definition of the \code{xor\_} operator.
Here the rule~[{\sc AbsInf}+] is not sufficient to precisely type the
function, and using only this rule would yield a type
 $\Any\to\Any\to\Bool$.
\iflongversion
Let us follow the behavior of the
 ``$\worra{}{}$'' operator. Here the whole \code{and\_} is requested
 to have type \True, which implies that \code{or\_ x y} must have
 type \True. This can always happen, whether \code{x} is \True{} or
 not (but then depends on the type of \code{y}). The ``$\worra{}{}$''
 operator correctly computes that the type for \code{x} in the
 ``\emph{then}'' branch is $\True\vee\lnot\True\lor\True\simeq\Any$,
and a similar reasoning holds for \code{y}.
\fi
However, since \code{or\_} has type
$(\True\to\Any\to\True)\land(\Any\to\True\to\True)\land
   (\lnot\True\to\lnot\True\to\False)$
then the rule \Rule{OverApp} applies and \True, \Any, and $\lnot\True$ become candidate types for
\code{x}, which allows us to deduce the precise type given in the table. Finally, thanks to rule \Rule{OverApp} it is not  necessary to use a type case to force refinement. As a consequence, we can define the functions \code{and\_} and \code{xor\_} more naturally as:
\begin{alltt}\color{darkblue}
  let and_ = fun (x : Any) -> fun (y : Any) -> not_ (or_ (not_ x) (not_ y)) \refstepcounter{equation}          \mbox{\color{black}\rm(\theequation)}\label{and+}
  let xor_ = fun (x : Any) -> fun (y : Any) -> and_ (or_ x y) (not_ (and_ x y)) \refstepcounter{equation}      \mbox{\color{black}\rm(\theequation)}\label{xor+}
\end{alltt}
for which the very same types as in Table~\ref{tab:implem} are deduced.

As for Code~10 (corresponding to our introductory
example~\eqref{nest1}), it illustrates the need for iterative refinement of
type environments, as defined in Section~\ref{sec:typenv}. As
explained, a single pass analysis would deduce
for {\tt x}
a type \Int{} from the {\tt f\;x} application and \Any{} from the {\tt g\;x}
application. Here by iterating a second time, the algorithm deduces
that {\tt x} has type $\Empty$ (i.e., $\textsf{Empty}$), that is, that the first branch can never
be selected (and our implementation warns the user accordingly). In hindsight, the only way for a well-typed overloaded function to have
type $(\Int{\to}\Int)\land(\Any{\to}\Bool)$ is to diverge when the
argument is of type \Int: since this intersection type states that
whenever the input is \Int, {\em both\/} branches can be selected,
yielding a result that is at the same time an integer and a Boolean.
This is precisely reflected by the case $\Int\to\Empty$ in the result.
Indeed our {\tt example10} function can be applied to an integer, but
at runtime the application of {\tt f\,x} will diverge.

Code~11 implements the typical type-switching pattern used in JavaScript. While
languages such as Scheme and Racket hard-code specific type predicates for each
type---predicates that our system does not need to hard-code since they can be
directly defined (cf. Code~3)---, JavaScript hard-codes a \code{typeof} function
that takes an expression and returns a string indicating the type of the
expression. Code~11 shows that \code{typeof} can be encoded and precisely typed
in our system. Indeed, constant strings are simply encoded as fixed list of
characters (themselves encoded as pairs as usual, with special atom \code{nil}
representing the empty list). Thanks to our precise tracking of singleton types
both in the result type of \code{typeof} and in the type case of
\code{test}, we can deduce for the latter a precise type (the given in
Table~\ref{tab:implem} is equivalent to
$(\textsf{Any}\to\Int)\wedge(\lnot(\Bool{\vee} \Int {\vee} \Char) \to 0)$).

Code~12 simulates the behavior of JavaScript property resolution, by looking
for a property \code{l} either in the object \code{o} itself or in the
chained list of its \code{prototype} objects. In this example, we first model
prototype-chaining by defining a type \code{Object} that can be either the
atom \code{Null} or any record with a \code{prototype} field which contains
(recursively) an \code{Object}. To ease the reading, we defined a recursive
type \code{ObjectWithPropertyL} which is either a record with a field
\code{l} or a record with a prototype of type \code{ObjectWithPropertyL}. We
can then define two predicate functions \code{has\_property\_l} and
\code{has\_own\_property\_l} that test whether an object has a property
through its prototype or directly. Lastly, we can define a function
\code{get\_property\_l} which directly accesses the field if it is present, or
recursively search for it through the prototype chain; the recursive
search is implemented by calling the explicitly-typed
parameter \code{self} which, in our syntax, refers to the function itself. Of
particular interest is the type deduced for the two predicate functions. Indeed,
we can see that \code{has\_own\_property\_l} is given an overloaded type whose
first argument is in each case a recursive record type that describes precisely
whether \code{l} is present at some point in the list or not (recall that
in a record type a field such as $\orecord{ \ell=?\textsf{Empty} }$, indicate that field \code{$\ell$}
is surely absent). 
Notice that in our language a fixpoint combinator can be defined as follows
\begin{alltt}\color{darkblue}
type X = X -> \textit{S} -> \textit{T}
let z = fun (((\textit{S} -> \textit{T}) -> \textit{S} -> \textit{T} ) -> (\textit{S} -> \textit{T})) f ->
      let delta = fun ( X -> (\textit{S} -> \textit{T}) ) x ->
         f ( fun (\textit{S} -> \textit{T}) v -> ( x x v ))
       in delta delta   
\end{alltt}
which applied to any function \code{f:(\textit{S}$\to$\textit{T})$\to$\textit{S}$\to$\textit{T}}
returns  a function \code{(z\,f):\textit{S}$\to$\textit{T}} such that
 for every non
diverging expression \code{$e$} of type \code{\textit{S}}, the
expression \code{(z\,f)$e$} (which is of type \code{\textit{T}}) reduces to \code{f((z\,f)$e$)}.
It is then clear that definition of \code{get\_property\_l} in Code 12, is nothing but syntactic sugar for
\begin{alltt}\color{darkblue}
let get_property_l =
  let aux = fun (self:Object->Any) -> fun (o:Object)->
    if has_own_property_l o is True then o.l
    else if o is Null then null
    else self (o.prototype)
  in z aux
\end{alltt}
where  \code{\textit{S}} is \code{Object} and  \code{\textit{T}} is \code{Any}.

\subsection{Comparison}

\lstset{language=[Objective]Caml,columns=fixed,basicstyle=\linespread{0.43}\ttfamily\scriptsize,aboveskip=-0.5em,belowskip=-1em,xleftmargin=-0.5em}
\begin{table}
  {\scriptsize
    \begin{tabular}{|@{\,}c@{\,}|p{0.54\textwidth}@{}|@{\,}p{0.41\textwidth}@{\,}|}
      \hline
                             & Code                    & Inferred type            \\
      \hline
      1                      &
      \begin{lstlisting}
(* Assumes: add1 : Int -> Int *)
let example1 = fun (x:Any) ->
     if x is Int then add1 x else 0
\end{lstlisting} & \vfill
      $\Int \to \Int$
      \\\hline
      2                      &
      \begin{lstlisting}
(* Assumes strlen: String -> Int *)
 let example2 = fun (x:String|Int) ->
  if x is Int then add1 x else strlen x
\end{lstlisting} & \vfill
      $(\Int \to \Int) \land (\String \to \Int)$                                  \\\hline
      3                      & \begin{lstlisting}
let example3 = fun (x: Any) ->
    if x is (Any \ False) then (x,x) else false  \end{lstlisting}  & \smallskip
      $(\False \to \False) \land (\lnot \False \lnot (\lnot \False,\lnot\False))$
      \\\hline
      4                      &
      \begin{lstlisting}
(*Uses `is_int` from Table 1.3 and `or_` from Table 1.5,
  assumes f : (Int|String) -> Int *)
let is_string = fun (x : Any) ->
  if x is String then true else false

let example4 = fun (x : Any) ->
   if or (is_int x) (is_string x) is True then x else 'A'

 \end{lstlisting} & \vfill
  $(\String\to\Keyw{True})\land(\lnot\String\to\Keyw{False})$\newline
  ~\newline
  $(\Int \to \Int) \land  (\String \to \String)  \land$\newline
  $(\lnot \Int \to (\String \lor \Keyw{'A'})) \land $\newline
  $(\lnot \String \to (\Int \lor \Keyw{'A'})) \land $\newline  
  $(\lnot(\String \lor \Int) \to \Keyw{'A'})$
               \\\hline
      5                      &
      \begin{lstlisting}
(*Uses `and_` from Table 1.6,
  assumes strlen : String -> Int *)
let example5 = fun (x : Any) -> fun (y : Any) ->
   if and_ (is_int x) (is_string y) is True then
    add x (strlen y) else 0
\end{lstlisting}
                             & \smallskip\vfill
$(\Int \to \String \to \Int) \land (\Int \to \lnot \String \to 0))~\land$\newline
$(\lnot \Int \to \String \to 0) \land (\lnot \String \to 0)$ 
      \\\hline
      6                      &
\begin{lstlisting}
let example6 = fun (x : Int|String) -> fun (y : Any) ->
if and_ (is_int x) (is_string y) is True then
  add x (strlen y) else strlen x
\end{lstlisting}
                             & \vfill 
Type error for \texttt{strlen x},
\texttt{x} has type $\Int\lor\String$. \\\hline
      7                      &
      \begin{lstlisting}
let example7 = fun (x : Any) -> fun (y : Any) ->
if  
 (if (is_int x) is True then (is_string y) else false)
  is True then
  add x (strlen y) else 0
\end{lstlisting} &\medskip\vfill
$(\Int \to \String \to \Int) \land (\Int \to \lnot \String \to 0))~\land$\newline
$(\lnot \Int \to \String \to 0) \land (\lnot \String \to 0)$ 
(identical to example 5)
      \\\hline
 8 & \begin{lstlisting}
let example8 = fun (x : Any) ->
 if or_ (is_int x) (is_string x) is True then true
 else false
\end{lstlisting}  &\smallskip\vfill
$(\Int \to \True)\land(\String \to \True)~\land$\newline
$(\lnot(\String\lor\Int)\to \False)$
      \\\hline
 9    &
      \begin{lstlisting}
let example9 = fun (x : Any) -> 
if 
  (if is_int x is True then is_int x else is_string x)
  is True then  f x else 0
\end{lstlisting} &\smallskip\vfill
$(\Int \to \Int) \land (\String \to \Int)~\land$\newline
$(\lnot (\String\lor\Int)\to 0)$
      \\\hline
10   & \begin{lstlisting}
let example10 = fun (p : (Any,Any)) ->
  if is_int (fst p) is True then add1 (fst p) else 7
\end{lstlisting} & \vfill
$((\Int,\textsf{Any}) \to \Int) \land ((\lnot (\Int,\textsf{Any}) \to 7)$
      \\\hline
      11                     & \begin{lstlisting}
let example11 = fun (p : (Any, Any)) ->
  if and_ (is_int (fst p)) (is_int (snd p)) is True
  then g p else no   
\end{lstlisting} & \vfill\smallskip
   $((\Int, \Int) \to \Int) \land
    (((\textsf{Any}, \lnot \Int)\lor(\lnot\Int, \textsf{Any}))\to\Keyw{no})$
      \\\hline
      12                     & \begin{lstlisting}
let example12 =
   fun (p : (Any, Any)) ->
   if is_int (fst p) is True then true else false
\end{lstlisting} &\medskip \vfill
$((\Int,\textsf{Any}) -> \True) \land ((\lnot \Int,\textsf{Any}) -> \False) )$
      \\\hline
      13                     & \begin{lstlisting}
let example13 = fun (x : Any) -> fun (y : Any) ->
  if and_ (is_int x) (is_string y) is True then 1
  else if is_int x is True then 2
  else 3
\end{lstlisting} & \vfill
$(\Int \to \String \to 1) \land (\Int \to \lnot\String \to 2)) \land$\newline
$(\lnot \Int \to \textsf{Any} \to 3)$
              \\\hline
14   & \begin{lstlisting}
let example14_alt = fun (input : Int | String) ->
 fun (extra : (Any, Any)) ->
  if and2_((is_int input),(is_int (fst extra))) is True
    then add input (fst extra)
  else if (is_int input,is_int (fst extra)) is (Any,True)
    then add (strlen input) (fst extra)
  else 0\end{lstlisting}
  & \medskip\medskip\vfill
$(\Int \to ((\Int,\textsf{Any}) \to \Int) \wedge ((\lnot\Int,\textsf{Any}) \to 0))\wedge$\newline
$(\String \to ((\Int,\textsf{Any}) \to \Int) \wedge
((\lnot\Int,\textsf{Any}) \to 0))$
  \\\hline
    \end{tabular}
  }
  \caption{Comparison with the 14 examples of \citet{THF10}}
  \ifsubmission%
    \svvspace{-10mm}
  \fi%
  \label{tab:implem2}
\end{table}

In Table~\ref{tab:implem2}, we reproduce in our syntax the 14
archetypal examples of
\citet{THF10} (we tried to complete such examples with neutral code when they
were incomplete in the original paper). Of these 14 examples,
Example~1 to 13 depict combinations of type predicates (such
as \code{is\_int}) used either directly or through Boolean
predicates (such as the \code{or\_} function previously
defined). Note that for all examples for which there was no explicit
indication in the original version, we \emph{infer} the type of the
function whereas in~\cite{THF10} the same examples are always in a
context where the type of identifiers is known or the input type of
function is fully annotated.  Notice also that for Example~6, the goal
of the example is to show that indeed, the function is ill-typed
(which our typechecker detects accurately).



The original Example~14 of~\citet{THF10} is the only case of their
work that our
system cannot directly capture. It can be written in our syntax as:
\begin{alltt}\color{darkblue}
  let example14 = fun (input : Int|String) ->
    fun (extra : (Any, Any)) ->
      if and2_(is_int input , is_int(fst extra)) is True then
         add input (fst extra)  \refstepcounter{equation}                                                     \mbox{\color{black}\rm(\theequation)}\label{ex14}
      else if is_int(fst extra) is True then
         add (strlen input) (fst extra)
      else 0
\end{alltt}
where \code{and2\_} is the uncurried version of the \code{and\_}
function we defined in \eqref{and+} and \code{is\_int} is the fuction
defined in the third row of Table~\ref{tab:implem}.
Our system rejects the expression above, while the system by~\citet{THF10}
correctly infers the function always return an integer. The reason
why our system rejects it is because the type it deduces for the
occurrence of \code{input} in the 6th line of the code
is \code{Int|String} rather than \code{String} as required by the
application of \code{strlen}. The general reason for this failure is
that, contrary to~\cite{THF10}, our system does not implement an
analysis of the flow of type information. In particular, since the
variable \code{input} does not occur in the condition of the
second \code{if}, then its type is not refined (as it could be).
Indeed, if the first test fails, it is either
because  \code{fst\,extra} is not an integer
(i.e., \code{is\_int(fst\,extra)} is not \True) or
because \code{input} is not an integer. Therefore, in
our setting, the type information propagated to the second test for
the pair of the arguments in the first test is :
$\code{(is\_int input , is\_int(fst\,extra))} \in \lnot (\True, \True)$, that
is $\code{(input, is\_int(fst\,extra))} \in
(\lnot\Int, \True)\lor(\Int, \False)$. Since the second test checks
whether \code{is\_int(fst\,extra)} holds or not, then we could deduce
that the following occurrence of \code{input} is of type
$\lnot\Int$. But since \code{input} does not occur in the test, this
refinement of the type of \code{input} is not done. Instead, the type deduced
for \code{input} in the second branch is $(\String\lor\Int) \land
(\lnot\Int \lor \Int) = \String\lor\Int$ which is not precise
enough to type the application \code{strlen\;input}. It not difficult
to patch, alas unsatisfactorly, this example in our system: it suffices to test
dummily in the second \code{if} the whole argument of \code{and2\_},
without really checking its first component:
\begin{alltt}\color{darkblue}
  let example14_alt = fun (input : Int|String) ->
    fun (extra : (Any, Any)) ->
      if and2_(is_int input , is_int(fst extra)) is True then
         add input (fst extra)
      else if (is_int input , is_int(fst extra)) is (Any,True) then
         add (strlen input) (fst extra)
      else 0
\end{alltt}
Even if the type of \code{is\_int\,input} is not really tested (any
result will produce the same effect) its presence in the test triggers the
refinement of the type of the last occurrence of \code{input}, which
type checks with the (quite precise) type shown in the entry 14 of
Table~\ref{tab:implem2}, type that is equivalent to $\Int\vee\String \to ((\Int,\textsf{Any}) \to \Int) \wedge ((\lnot\Int,\textsf{Any}) \to 0)$.
Lifting this limitation through a control-flow
analysis is part of our future work.

\rev{
In our system, however, it is possible to express dependencies between
different arguments of a function by uncurrying the function and typing its
arguments by a union of products. To understand this point, consider this
simple example:}
\begin{alltt}\color{darkblue}
  let sum = fun (x : Int|String) -> fun (y : Int|String) ->
      if x is String then concat x y else add x y
\end{alltt}
\rev{
The definition above does not type-check in any available system, and rightly
does so since nothing ensures
that \code{x} and \code{y} will be either both strings (so
that \code{concat} does not fail) or both integers (so that \code{add}
does not fail). It is however possible to state this dependency
between the type of the two arguments by
uncurring the function and using a union type:
}
\begin{alltt}\color{darkblue}
  let sum = fun (x : (Int,Int)|(String,String))
      if fst x is String then concat(fst x)(snd x) else add(fst x)(snd x)
\end{alltt}
\rev{
this function type-checks in our system (and, of course, in Typed
Racket as well) but the corresponding type-annontated version in JavaScript
}
\begin{alltt}\color{darkblue}
  function sum (x : [string,string]|[number,number]) \{
    if (typeof x[0] === "string") \{
       return x[0].concat(x[1]); 
    \} else \{
       return x[0] + x[1];
    \}
  \}
\end{alltt}
\rev{
is rejected both by Flow and TypeScript since their type analyses fail to detect the
dependency of the types of the two projections.
}

Although these experiments are still preliminary, they show how the
combination of occurrence typing and set-theoretic types, together
with the type inference for overloaded function types presented in
Section~\ref{sec:refining} goes beyond what languages like
TypeScript and Flow do, since they can only infer single arrow types.
Our refining of overloaded
functions is also future-proof and resilient to extensions: since it ``retypes'' functions
 using information gathered by the typing of occurrences in the body,
 its precision will improve with any improvement of
 our occurrence typing framework.


\section{Related work}
\label{sec:related}
\kim{
I don't really know about these:\\
Occurrence typing is for variables in TypeScript, also paths in Flow. STRESS THE USE OF INTERSECTIONS AND UNIONS. ACTUALLY SPEAK OF UNION AND NEGATION (UNION FOR ALTERNATIVES, NEGATION FOR TYPING THE ELSE BRANCH AND THEREFORE WE GET INTERSECTION FOR FREE)
 extend to selectors, logical connectives, paths, and user
defined predicates.

[IMPORTANT: check the differences with what we do here]

State what we capture already, for instance lists since we have product and recursive types.
}
\kim{Need to clean-up and insert at the right place.
}

Occurrence typing was introduced by \citet{THF08} and further advanced
in \cite{THF10} in the context of the Typed Racket language. This
latter work in particular is close to ours, with some key differences.
\citet{THF10} define $\lambda_{\textit{TR}}$, a core calculus for
Typed Racket. In this language types are annotated by two logical
propositions that record the type of the input depending on the
(Boolean) value of the output. For instance, the type of the
{\tt number?} function states that when the output is {\tt true}, then
the argument has type {\tt Number}, and when the output is {\tt false}, the
argument does not. Such information is used selectively
in the ``then'' and ``else'' branches of a test.
\rev{
Since \citet{THF10} focus their analysis on a particular set of pure
operations, the approach works also in the presence of side-effects. Although
the choices made by our and their approach seem poles apart
(Boolean output of few pure operations vs.\ any output of every
expression), they share some similar techniques. For instance, our
deduction system for $\vdashp$ plays a similar role as
the proof systems and \textsf{update} function of \citet[Figures 4,
7 \& 9]{THF10}. In that framework, in order to type a variable
(judgement ``$\Gamma \vdash x:\tau$'') one needs to prove
that the logical formula $\tau_x$ holds (under the hypotheses of
$\Gamma$). This atomic formula may not be directly available in
$\Gamma$ but may
be proven by a combination of logical deduction rules (Figure~4 of~\cite{THF10}), or
by recursively exploring a path leading to $x$ (Figure~7 and ~9 of~\cite{THF10}) a
path being a sequence of \textbf{cdr} or \textbf{car} applications,
much like our $f$ and $s$ components of paths. This idea is also
present in our deduction system for $\vdashp$ with differences pertaining to our type
framework and design choices: type restrictions can be encoded using
set-theoretic intersections and negations (instead of meta-functions working on the
syntax of types) and our richer language of paths components.
}
One area where their
work goes further than ours is that the type information also flows
outside of the tests to the surrounding context. In contrast, our type
system only refines the type of variables strictly in the branches of a
test.
\rev{This is particularly beneficial when typing functions since the
logical propositions of Tobin-Hochstadt and Felleisen can record
dependencies on expressions other than the input of a
function. Consider for instance the following example (due
to~\cite{kent19phd}) in JavaScript \code{function
is-y-a-number(x) \{ return(typeof(y)\,===\,"number") \}} which defines
a functions that
disregards its argument and returns whether the variable \code{y} is
an integer or not.\footnote{Although such a function may appear
nonsensical, \citet{kent19phd} argues that it corresponds a programming
pattern that may appear in Typed Racked due to the expansion of some
sophisticated macro definitions.} While our approach cannot deduce for this function but the
type $\Any\to\Bool$, the logical approach of Tobin-Hochstadt and
Felleisen can record in the type of \code{is-y-a-number} the fact that
when the function returns \texttt{true}, then \code{y} is a number, and
the opposite when it returns \texttt{false}. In our approach, the only
possibility to track such a dependency is that the variable \code{y}
is the parameter of an outer function to which our analysis could give
an overloaded type by splitting the type \texttt{Any} of \texttt{y}
into \texttt{Number} and \texttt{$\neg$Number}. Under the hypothesis
of \texttt{y} being of type \texttt{Number} the type inferred
for \code{is-y-a-number} will then be $\Any\to\True$, and
$\Any\to\False$ otherwise, thus capturing the wanted dependency.
Although the approach of using logical proposition has the undeniable
advantage over ours of providing more a flow sensitive analysis, we
believe that using semantic subtyping as a foundation as we do
has also several merits over the logical proposition approach.
}
First, in our case, type
predicates are not built-in. A user may define any type predicate she
wishes by using an overloaded function, as we have shown in
Section~\ref{sec:practical}. Second, in our setting, {\em types\/}
play the role of formulæ. Using set-theoretic types, we can express
the complex types of variables without resorting to a meta-logic. This
allows us to type all but two of the key examples of~\citet{THF10} (the notable
exceptions being Example~9 and 14 in their paper, which use the
propagation of type information outside of the branches of a
test). While Typed Racket supports structured data types such as pairs
and records only unions of such types can be expressed at the level of
types, and even for those, subtyping is handled axiomatically. For
instance, for pairs, the subtyping rule presented in \cite{THF10} is
unable to deduce that
$(\texttt{number}\times(\texttt{number}\cup\texttt{bool}))\cup
(\texttt{bool}\times (\texttt{number}\cup\texttt{bool}))$ is a subtype
of (and actually equal to)
$((\texttt{number}\cup\texttt{bool})\times\texttt{number})\cup
((\texttt{number}\cup\texttt{bool})\times\texttt{bool})$ (and likewise
for other type constructors combined with union types). For record
types, we also type precisely the deletion of labels, which, as far as
we know no other system can do. On the other hand, the propagation of
logical properties defined in \cite{THF10} is a powerful tool, that
can be extended to cope with sophisticated language features such as
the multi-method dispatch of the Closure language \cite{Bonn16}.

\kim{Remove or merge :
Also, while they
extend their core calculus with pairs, they only provide a simple {\tt
cons?} predicate that allows them to test whether some value is a
pair. It is therefore unclear whether their systems allows one to
write predicates over list types (e.g., test whether the input
is a list of integers), which we can easily do thanks to our support
for recursive types.}

\rev{
For what concerns the first work by~\citet{THF08} it is interesting to
compare it with our work because the comparison shows two rather
different approaches to deal with the property of type
preservation. \citet{THF08} define a first type system that does not
satisfy type-preservation. The reason for that is that this first type
system checks all the branches of a type-case expression,
independently from whether they are selectable or not; this may
result in a well-typed expression to reduce to an expression that 
is not well-typed because it contains a type-case expression with
a branch that, due to the reduction, became both non-selectable
and ill-typed (see \cite[Section 3.3]{THF08}). To obviate this problem
they introduce a \emph{second} type system that extends the previous
one with some auxiliary typing rules that type type-case expressions
by skipping the typing of non-selectable branches. They use this
second type system only to prove type preservation and obtain, thus,
the soundness of their type system. In our work, instead, we prefer to
start directly with a system that satisfies type preservation. Our
system does not have the problem of the first system of~\cite{THF08}
thanks to the
presence of the \Rule{Efq} rule, that we included for that very purpose,
that is, to skip non-selectable branches during typing. The choice of one
or the other approach is mostly a matter of taste and, in this specific case,
boils down to deciding whether some typing problems must be
signaled at compile time by an error or a warning. The approach
of~\citet{THF08} ensures that every subexpression of a program is
well-typed and, if not, it generates a type-error. Our approach allows
some subexpressions of a program to be ill-typed, but only if they
occur in dead branches of type-cases: in that case any reasonable
implementation would flag a warning to signal the presence of the dead
branches. The very same reasons that explain the presence in our system
of \Rule{Eqf}, explain why from the beginning we included in our
system the typing rule \Rule{Abs-} that deduces negated arrow types:
we wanted a system that satisfied type preservation (albeit, for a
parallel reduction: cf: \Appendix\ref{app:parallel}). We then defined
an algorithmic system that is not \emph{complete} with respect to the
type-system but from which it inherits its soundness. Of course, we
could have proceeded as~\citet{THF08} did: start directly with a
type-system corresponding to the algorithm (i.e., omit the
rule \Rule{Abs-}) and later extend this system with the rule to infer
negated arrows, the only purpose of this extension being to prove type preservation. We preferred not
to, not only because we favor type preserving systems, but also
because in this way we were able to characterize different subsystems that
are complete with respect to the algorithmic system, thus exploring
different language designs and arguing about their usefulness.
}

\rev{
Highly related to our work is Andrew M.\ Kent's PhD.\
dissertation~\cite{kent19phd}, in particular its Chapter 5 whose title
is ``A set-theoretic foundation for occurrence typing'' where he endows the logical techniques of~\cite{THF10} with the set-theoretic types of semantic
subtyping~\cite{Frisch2008}. Kent's work builds on the approach
developed for Typed Racket that, as recalled above, consists in enriching the
types of the expressions with information to track under which hypotheses an expression
returns false or not (it considers every non false value to be
``truthy''). This tracking is performed by recording in the type of
the expression two logical propositions that hold when the expression
evaluates to false or not, respectively. The work in~\citet[Chapter
5]{kent19phd} uses set-theoretic types to express type predicates (a
predicate that holds only for a type $t$ has type
$p:(t\to\texttt{True})\land(\neg t\to\texttt{False})$) as well as to
express in a more compact (and, sometimes, more precise) way the types
of several built-in Typed Racket functions. It also uses the
properties of set-theoretic types to deduce the logical types (i.e.,
the propositions that hold when an expressions produces \texttt{false}
or not) of arguments of function applications.  To do that it defines
a type operator called \emph{function application inversion}, that
determines the largest subset of the domain of a function for which an
application yields a result of a given type $t$, and then uses it for
the special cases when the type $t$ is either \texttt{False} or \texttt{$\lnot$False} so
as to determine the logical type of the argument. For instance, this
operator can be used to deduce that if the application \texttt{boolean?\,x}
yields false, then the logical proposition \texttt{x${\in}\neg$Bool} holds true. The definition
of our \emph{worra} operator that we gave in equation~\eqref{worra} is, in
its spirit, the same as Kent's \emph{function application inversion}
operator (more precisely, the same as the operator \textsf{pred} Kent defines in Figure
5.7 of his dissertation), even though the two operators were defined independently from
each other. The exact definitions however are slightly different, since
the algorithm given in~\citet[Figure 5.2]{kent19phd}
for \emph{function application inversion} is sound only for functions
whose type is an intersection of arrows, whereas our definition
of worra, given in~\eqref{worralgo}, is sound and complete for any function, in
particular, for functions that have a union type (for which Kent's
definition may yield unsound results). Apart from these technical
issues, the main
difference of Kent's approach with respect to ours is that, since it
builds on the logical propositions approach, then it focus the use of
set-theoretic types and of the worra (or application inversion)
operator to determine when an expression yields a result of
type \texttt{False} or \texttt{$\lnot$False}. We have instead a more
holistic approach since, not only our analysis strives to infer type
information by analyzing all types of results (and not
just \texttt{False} or \texttt{$\lnot$False}), but also it tries to
perform this analysis for all possible expressions (and not just for a
restricted set of expressions). For instance, we use the operator
worra also to refine the type of the function in an application (see
discussion in Section~\ref{sec:ideas}) while in Kent's approach
the analysis of an application \texttt{f\,x} refines the properties of
the argument \texttt{x} but not of the function \texttt{f}; and when such an application is the argument of a type
test, such as in \texttt{number?\,(f\,x)}, then in Kent's approach it is no longer possible
to refine the information on the argument \texttt{x}. The latter is
not is a flaw of the approach but a design choice: as we explain at
the end of this section, the approach of Type Racket not only focuses on the  inference of two logical
propositions according to the truthy or false value of an expression,
but also it does it only for a selected set of \emph{pure} expressions of
the language, to cope with the possible presence of side effects (and
applications do not belong to this set since they can be be impure).
That said, the very fact of focusing on truthy vs.\ false results may
make Kent's analysis fail even for pure Boolean tests where it would
be naively expected to work. For example, consider
the polymorphic function that when applied to two integers returns
whether they have the same parity and false
otherwise: \code{have\_same\_parity:$(\Int\To\Int\To\Bool)\land(\neg\Int\To\ANY\To\False)\land(\ANY\To\neg\Int\To\False)$}. We
can imagine to use this function to implicitly test whether two
arguments are both integers, as in the body of the following function:%
}
\begin{alltt}\color{darkblue}
  let f = fun (x : Any) -> fun (y : Any) ->
    if have\_same\_parity x y is True then add x y else 0
\end{alltt}
\rev{
While our approach can correctly deduce for this function the type
$\ANY\To\ANY\To\Int$, Kent's approach fails to type check it since
to type the ``then'' branch requires to deduce that  the
application \code{have\_same\_parity\;x} returns  the constant
function \code{true} only if \code{x} is an integer.
Finally, Kent's approach inherits all the advantages and
disadvantages that the logical propositions approach has with respect
to ours (e.g., flow sensitive analysis vs.\ user-defined type
predicates) that we already discussed at the beginning of this
section.
}

Another direction of research related to ours is the one on semantic
types. In particular, several attempts have been made recently to map
types to first order formul\ae. In that setting, subtyping between
types translates to logical implication between formul\ae.
\citet{Bierman10} introduce Dminor, a data-oriented
language featuring a {\tt SELECT}-like construct over
collections.  Types are mapped to first order formulæ and an SMT-solver is
then used to (try to) prove their satisfiability. The refinement
types they present go well beyond what can be expressed with the set-theoretic
types we use (as they allow almost any pure expression to occur in
types). However, the system forgoes any notion (or just characterization) of completeness
and the
subtyping algorithm is largely dependent on the subtle behavior of
the SMT solver (which may timeout or give an incorrect model that
cannot be used as a counter-example to explain the type-error).
As with our work, the typing rule for the {\tt if~$e$ then $e_1$ else
$e_2$} construct
of Dminor refines the type of each branch by remembering that $e$
(resp. $\lnot e$) is true in $e_1$ (resp. $e_2$) and this information
is not propagated to the outer context.
A similar approach is taken by \citet{Chugh12}, and extended to so-called nested refinement types. In these types, an arrow type may
appear in a logical formula (whereas previous work only allowed formul\ae
on  ``base types''). This is done in the context of a dynamic
language and their approach is extended with polymorphism, dynamic
dispatch and record types.
\rev{
A problem that is faced by refinement type systems is the one of
propagating in the branches of a test the very precise information
learned from the test (usually that some equality between terms holds).  A solution
that is for instance chosen by~\citet{OTMW04} and \citet{KF09} is to devise a
meta-function that recursively explores both a type and an expression
and constructs a more precise \emph{dependent} type. In the dependent
type, fresh variables are introduced to name sub-expressions and
record the new constraints.
This
process---called in the cited works \emph{selfification}---roughly
corresponds to  our $\constrf$ and \Refinef{} functions (see Section~\ref{sec:typenv}).
Another approach is the one followed by \citet{RKJ08} which is
completely based on a program transformation, namely, it consists in putting the term
in \emph{A-normal form} as defined by~\citet{SF92}. Using a program
transformation, every destructor application (function application,
projection, \ldots) is given a name through a let-binding. The problem
of tracking precise type information for every sub-expression is
therefore reduced to the one of keeping precise typing information for
a variable.  While this solution seems appealing, it is not completely
straightforward in our case. Indeed, to retain the same degree of
precision, one would need to identify $\alpha$-equivalent sub-expressions
so that they share the same binding, something that a plain A-normalization
does not provide (and which, actually, must not provide, since in that
case the transformation may not preserve the reduction semantics).

Among the work on refinement types, some have studied the extensions
of a refinement type-system with intersection types. For
instance, \cite{BHM14} studies a type system with refinement types,
polymorphism and full union and intersection (but no negation). While
the goal of their type-system is to verify secure protocol
implementations, the core language RCF$^{\forall}_{\land\lor}$ they
present, as well as the associated type-system is a $\lambda$-calculus
with pattern-matching, let bindings, and a refining test for equality
(as well as protocol-oriented constructs such as channel creation,
message passing, and expression forking). While on the surface their
types resemble ours, they follow another direction. First, their
language is fully annotated (meaning that, for instance, polymorphic terms must
be explicitly instantiated and intersection types must also be specified
through an annotation). Second, since the subtyping relation they
provide is syntactic, it cannot in general take into account the
distributivity of logical connectives with respect to type constructors. This
limitation is however not a problem since the main goal of their
subtyping relation is to propagate a \emph{kinding} information that
they use to characterize the level of knowledge an attacker may have
about a particular value.
Another work adding intersection types to refinement types
is \cite{PAF15} in the context of liquid types. This work introduces
intersection (but not union nor negations) to liquid types, with a
particular focus on intersection of arrow types. This work uses a
syntactic subtyping relation to push down intersection of types into
the logical formulas of types. Once the formulas have been propagated,
they are offloaded to an SMT solver to decide the base case of the
subtyping relation. Of particular interest is their type-inference
algorithm. Contrary to ours, their inference is based on algorithm
$\mathcal{W}$, using the polymorphic type deduced as a template for an
intersection. They can therefore infer intersection arrow types that
are several distinct instances of the same polymorphic type.

}


\citet{Kent16} bridge the gap between prior work on occurrence typing
and SMT-based (sub-)typing. They introduce the $\lambda_{RTR}$ core calculus, an
extension of $\lambda_{TR}$ of~\cite{THF10} where the logical formulæ embedded in
types are not limited to built-in type predicates, but accept
predicates of arbitrary theories. This allows them to provide some
form of dependent typing (and in particular they provide an
implementation supporting bitvector and linear arithmetic theories).
The cost of this expressive power in types is however paid by the
programmer, who has to write logical
annotations (to help the external provers). Here, types and formul{\ae}
remain segregated. Subtyping of ``structural'' types is checked by
syntactic rules (as in \cite{THF10}) while logical formul{\ae} present
in type predicates are verified by the SMT solver.

\citet{Cha2017} present the design and implementation of Flow by formalizing a relevant
fragment of the language. Since they target an industrial-grade
implementation, they must account for aspects that we could afford to
postpone to future work, notably side effects and responsiveness of
the type checker on very large code base. The degree of precision of
their analysis is really impressive and they achieve most of what we
did here and, since they perform flow analysis and use an effect
system (to track mutable variables), even more. However, this results
in a specific and very complex system. Their formalization includes
only union types (though, Flow accepts also intersection types as
we showed in \eqref{foo2}) which are used in \emph{ad hoc} manner by the type
system, for instance to type record types. This allows Flow to perform
an analysis similar to the one we did for Code 8 in
Table~\ref{tab:implem}, but also has as a consequence that in some
cases unions do not behave as expected. In contrast, our approach is
more classic and foundational: we really define a type system, typing
rules look like classic ones and are easy to understand, unions are
unions of values (and so are intersections and negations), and the
algorithmic part is---excepted for fix points---relatively simple
(algorithmically Flow relies on constraint generation and
solving). This is the reason why our system seems more adapted to study
and understand occurrence typing and to extend it with additional
features (e.g., gradual typing and polymorphism) and we are eager to
test how much of their analysis we can capture and enhance by
formalizing it in our system.
\nocite{typescript,flow}
\rev{
More generally, we believe that what sets our work apart in the
palimpsest of the research on occurrence typing is that we have a
type-theoretic foundational approach striving as much as possible to
explain occurrence typing by extending prior (unrelated but standard)
work while keeping prior results. In that respect, we think that our
approach is not satisfactory, yet, because it uses non standard
type-environments that map expressions rather than variables to types:
but all the rest is standard type-theory. And even on the latter
aspect it must be recognized that the necessity of tracking types not
only for variables but also for more structured expressions is
something that shows up, in different forms, in several other
approaches. For instance, in the approach defined for Typed
Racket~\cite{THF10} the type-system associates to an expression a
quadruple formed by its type, two logical propositions, and an object
which is a pointer to the environment for the type hypothesis about
the expression and, as such, it plays the role of our extended type
environments. Likewise, the \emph{selfification} of \cite{OTMW04}
and \cite{KF09}, propagates the precise type constraints learned
during a test. One difference with our approach is that with refinement types the
information can be
kept at the level of types, since dependent types contain terms and can
introduce variables, while in our approach the mapping is kept separate in a
type environment.  In summary the tracking of types for structured expressions
seems an aspect common to different approaches to occurrence types,
nevertheless we are confident that even this last non-standard aspect
of our system can be removed and that occurrence typing can be
explained in a pure standard type-theoretic setting.

On the practical side, while languages such as Flow and
Typed Racket are the golden standard of occurrence typing, it may be
worth citing that there exist
other programming languages that implement some much more simplistic forms
of occurrence typing. Languages such as Kotlin~\cite{Kotlin} and
Dart~\cite{googledart} enforce null safety by performing occurrence
typing whenever the tested expression is a
variable. CDuce~\cite{cduce} implements a slightly more sophisticated
form of this simplistic occurrence typing since it is able to refine
in the branches of a test the type of all variables that occur in the
tested expression as long as they are subexpressions of non-functional
values: so for instance for an expression of the form $\tcase{(x, (fz,
y))}{(\Int{\times}(\Int{\times}\Int))}{e_1}{e_2}$ CDuce is able to
to specialize in $e_1$
the types of $x$ and $y$ (to \Int) but not those of $f$ or $z$ (since
they occur in an application).
Likewise, Kotlin also supports dynamically testing the type of an
object (using the \texttt{\color{darkblue}is} operator similar to
Java's \texttt{\color{darkblue}instanceOf}) and refining the type of the tested
variable in the corresponding branch of a test, without having to
resort to a manual down-cast. As expected, Kotlin can only refine the
type of variables it can statically determine to be immutable, namely
local variables introduced by an immutable \texttt{\color{darkblue}val} binding and
mutable references introduced by a \texttt{\color{darkblue}var} binding, provided
the reference is not modified between the type test and its
occurrences in the
branch.

This work already has a follow-up, which was recently presented at the
POPL conference~\cite{CLNL22}. Both this work and the system
in~\cite{CLNL22} use the characteristics of semantic subtyping to
improve occurrence typing. Both works obtain this improvement by precisely
tracking the type of each occurrence of an expression. However, they
use rather different techniques to track the occurrences of an
expression and associate them with types. In this work, we do it by
enriching type environments so that they map occurrences of
expressions (expressed in terms of paths) to types. In~\cite{CLNL22},
instead, the different occurrences of the same expression are tracked by
using explicit bindings. In practice, in~\cite{CLNL22} every
expression is transformed into an intermediate representation---dubbed
maximal-sharing canonical form (MSC-form)--- that consists of a list
of bindings from variables to expressions whose proper subexpressions
are all variables. This form is called \emph{maximal sharing} because all
occurrences of a given expression are mapped by the same binding. In
other terms, for each subexpression, there is a unique variable and a
unique binding that tracks it. The advantages of using bindings
instead of enhanced type environments and paths are twofold. First, the
definition of the type system is standard: type environments map
variables to types, and occurrence typing is expressed by combining the
typing rules for type-case expressions with the standard
union-elimination rule by~\citet{MacQueen1986}. Second, MSC-forms
relate via a binding all occurrences of a given expression; so, in
particular, they may relate occurrences that are inside a type-case
with occurrences that are outside it. This allows the system
of~\cite{CLNL22} to capture and analyze the flows of information
between different expressions, a kind of analysis that makes the strength of the
approaches heralded by Flow and Typed Racket and which constitutes one
of the main limitations of the approach presented here.

We end this presentation of related work with a discussion on side
effects. Although in our system we did not take into account
side-effects---and actually our system works because all the expressions
of our language are pure---it is interesting to see how the different
approaches of occurrence typing position themselves with respect to the problem of
handling side effects, since this helps to better place our work in the taxonomy
of the current literature. As Sam Tobin-Hochstadt insightfully
noticed, one can distinguish the approaches that use types to reason about
the dynamic behavior of programs according to the set of expressions
that are taken into account by the analysis. In the case of occurrence
typing, this set is often determined by the way impure expressions are
handled. On the one end of the spectrum lies our approach: our
analysis takes into account \emph{all} expressions but, in its current
formulation, it works only for pure languages. On the other end of the
spectrum we find the approach of Typed Racket whose analysis reasons
about a limited and predetermined set of \emph{pure} operations: all data structure
accessors. Somewhere in-between lies the approach of the
Flow language which, as hinted above, implements a complex effect systems to determine
pure expressions. While the system presented here does not work for
impure languages, we argue that its foundational nature predisposes it
to be adapted to handle impure expressions as well, by adopting existing
solutions or proposing new ones. For instance, it is not hard to
modify our system so that it takes into account only a set of
predetermined pure expressions, as done by Typed Racket: it suffices
to modify the definition of $\Gamma \evdash e
{(\neg)t} \Gamma'$ (cf.\ Section~\ref{sec:static}) so that $\Gamma'$
extends $\Gamma$ with type hypotheses for all expressions occurring in
$e$ that are also in the set of predetermined pure expressions
(instead of extending it for all subexpressions of $e$, \emph{tout court}).
However, such a solution would be marginally interesting since by excluding
from the analysis all applications we would lose most of the
advantages of our approach with respect to the one with logical
propositions. Thus a more interesting solution would be to use some
external static analysis tools---e.g., to graft
the effect system of~\citet{Cha2017} on ours---to detect impure
expressions. The idea would be to mark different occurrences of a same
impure expression using different marks. These marks would essentially
be used to verify the presence of type hypotheses for a given
expression in a type environment $\Gamma$; the idea being that
expressions with different marks are to be considered as different
expressions and, therefore, would not share the same type
hypothesis. For instance, consider the test
$\ifty{f\,x}\Int{\,...\,}{\,...}$: if $f x$ were flagged as impure,
then an occurrence of $f x$ in the ``then'' branch would not be
supposed to be of type $\Int$ since it would be typed in an
environment $\Gamma$ containing a binding for an $f\,x$ expression
having a mark different from the one in the ``then'' branch: the
regular typing rules would apply for $f\, x$ in that case. This
would certainly improve our analysis, but we believe that ultimately
our system should not resort to external static analysis tools to detect
impure expressions but, rather, it has to integrate this analysis with
the typing one, so as to mark \emph{only} those impure expressions
whose side-effects may affect the semantics of some type-cases. For
instance, consider a JavaScript object \code{obj} that we modify as
follows: \code{obj["key"] = 3}. If the field \code{"key"} is already
present in \code{obj} with type \Int{} and we do not test it more than
about this type, then it is not necessary to mark different
occurrences of \code{obj} with different marks, since the result of
the type-case will not be changed by the assignment; the same holds
true if the field is absent but type-cases do not discriminate on its
presence. Otherwise, some occurrences of \code{obj} must use different
marks: the analysis will determine which ones. We leave this study for
future work.
}

\section{Future work and conclusion}
\label{sec:conclusion}
In this work we presented the core of our analysis of occurrence
typing, extended it to record types and proposed a couple of novel
applications of the theory, namely the reconstruction of
intersection types for unannotated functions and a static analysis to reduce the number of
casts inserted when compiling gradually-typed programs.
One of the by-products of our work is the ability to define type
predicates such as those used in \cite{THF10} as plain functions and
have the inference procedure deduce automatically the correct
overloaded function type. More generally, our approach surpasses 
current ones in that it can deduce precise (overloaded) types for
functions that in all other approaches either require 
the programmer to specify the full precise type (e.g., the
function \code{foo} we defined in \eqref{foo} and \eqref{foo2}  in
our introduction) or cannot be typed at all (the \code{and\_} and \code{xor\_} functions given
in \eqref{and+} and \eqref{xor+} are the most eloquent examples).

There is still a lot of work to do to fill the gap with real-world
programming languages. For example, our analysis cannot handle flow of
information, as we discussed for the function \code{example14} in Section~\ref{sec:practical}.
In particular, the result of a type test can flow only to the branches
but not outside the test. As a consequence the current system cannot
type a let binding such as \code{
  let x = (y\(\in\)Int)?{`}yes:`no in (x\(\in\)`yes)?y+1:not(y)%
}
which is clearly safe when  $y:\Int\vee\Bool$. Nor can this example be solved by partial evaluation since we do not handle nesting of tests in the condition\code{(\,((y\(\in\)Int)?{`}yes:`no)\(\in\)`yes\,)\,?\,y+1\,:\,not(y)},
and both are issues that the system by~\citet{THF10} can handle. We think that it is possible
to reuse some of their ideas to perform an information flow analysis on top of
our system to remove these limitations.
\iflongversion%
Some of the extensions we hinted to in Section~\ref{sec:practical}
warrant a formal treatment. In particular, the rule [{\sc OverApp}]
only detects the application of an overloaded function once, when
type-checking the body of the function against the coarse input type
(i.e., $\psi$ is computed only once). But we could repeat this
process whilst type-checking the inferred arrows (i.e., we would
enrich $\psi$ while using it to find the various arrow types of the
lambda abstraction). Clearly, if untamed, such a process may never
reach a fix point. Studying whether this iterative refining can be
made to converge and, foremost, whether it is of use in practice is among our objectives.
\fi%

But the real challenges that lie ahead are the handling of side
effects and the addition of polymorphic types.
\rev{
Our analysis works in 
pure languages and we already discussed at length at the end of the
previous section our plans to extend it to
cope with side-effects. However, the
ultimate solution of integrating type and effect analysis in a unique tool
is not more defined than that.
}
For polymorphism, instead, we can easily adapt
the main idea of this work to the polymorphic setting. Indeed, the
main idea  is to remove from the type of an expression all
the results of the expression that would make some test fail (or
succeed, if we are typing a negative branch). This is done by
applying an intersection to the type of the expression, so as to keep
only the values that may yield success (or failure) of the test. For
polymorphism the idea is the same, with the only difference that
besides applying an intersection we can also apply an
instantiation. The idea is to single out the two most general type
substitutions for which some test may succeed and fail, respectively, and apply these
substitutions to refine the types of the corresponding occurrences
in the ``then'' and ``else'' branches. Concretely, consider the test
$x_1x_2\in t^\circ$ where $t^\circ$ is a closed type and $x_1$, $x_2$ are
variables of type $x_1: s\to t$ and $x_2: u$ with $u\leq s$. For the
positive branch we first check whether there exists a type
substitution $\sigma$ such that $t\sigma\leq\neg t^\circ$. If it does not
exists, then this means that for all possible assignments of
polymorphic type variables of $s\to t$, the test may succeed, that is,
the success of the test does not depend on the particular instance of
$s\to t$ and, thus, it is not possible to pick some substitution for
refining the occurrence typing. If it exists, then 
we find a type substitution $\sigma_\circ$ such that $t^\circ\leq
t\sigma_\circ$ and we refine for the
positive branch the types of $x_1$, of $x_2$, and of $x_1x_2$ by applying $\sigma_\circ$ to their types. While the
idea is clear%
\iflongversion
,
\else
\ (see \Appendix\ref{app:roadmap} for a more detailed explanation),
\fi
the technical details are quite involved, especially if we also want
functions with intersection types and/or  gradual
typing. Nevertheless, our approach has an edge on systems that do not
account for polymorphism.
\iflongversion
This needs a whole gamut of non trivial research that we plan to
develop in the near future building on the work on polymorphic types
for semantic subtyping~\cite{CX11} and the research on the definition of
polymorphic languages with set-theoretic types
by~\citet{polyduce2,polyduce1,CPN16} and \citet{Pet19phd}.
\fi

\subsubsection*{Acknowledgments}   
\noindent The authors thank Paul-André Melliès for his help on type
ranking and Sam Tobin-Hochstadt and the other reviewers for their
feedback and useful insight.
  This research was partially supported by Labex DigiCosme (project ANR-11-LABEX-0045-
  DIGICOSME) operated by ANR as part of the program «Investissement d'Avenir» Idex
  Paris-Saclay (ANR-11-IDEX-0003-02) and by a Google PhD fellowship for the second author.

 \bibliographystyle{ACM-Reference-Format}
 \bibliography{main}

\pagebreak

\appendix

\iflongversion\else
\section{Definition of the Subtyping Relation}
\label{sec:subtyping}

\newpage
\fi

\section{Proof of Type Soundness}
\label{sec:proofs}
\renewcommand{\ite}[4]{\ensuremath{(#1{\in}#2)\,\texttt{\textup{?}}\,#3\,\texttt{\textup{:}}\,#4}}
We give in this section the complete formalization of the declarative type system as
well as the proof of its type safety.

\newtheorem{property}{Property}

    \subsection{The declarative type system}\label{sec:declarative}

    \begin{mathpar}
        \Infer[Env]
      { }
      { \Gamma \vdash e: \Gamma(e) }
      { e\in\dom\Gamma }
  \qquad
  \Infer[Inter]
      { \Gamma \vdash e:t_1\\\Gamma \vdash e:t_2 }
      { \Gamma \vdash e: t_1 \wedge t_2 }
      { }
      \qquad
      \Infer[Subs]
      { \Gamma \vdash e:t\\t\leq t' }
      { \Gamma \vdash e: t' }
      { }
      \\
      \Infer[Const]
      { }
      {\Gamma\vdash c:\basic{c}}
      { }
  \quad
  \Infer[App]
      {
        \Gamma \vdash e_1: \arrow {t_1}{t_2}\quad
        \Gamma \vdash e_2: t_1
      }
      { \Gamma \vdash {e_1}{e_2}: t_2 }
      { }
          \\
          \Infer[Abs+]
            {{\scriptstyle\forall i\in I}\quad\Gamma,x:s_i\vdash e:t_i}
            {
            \Gamma\vdash\lambda^{\wedge_{i\in I}\arrow {s_i} {t_i}}x.e:\textstyle \bigwedge_{i\in I}\arrow {s_i} {t_i}
            }
            { }
          \\
          \Infer[Abs-]
          {\Gamma \vdash \lambda^{\wedge_{i\in I}\arrow {s_i} {t_i}}x.e:t}
          { \Gamma \vdash\lambda^{\wedge_{i\in I}\arrow {s_i} {t_i}}x.e:\neg(t_1\to t_2)  }
          { ((\wedge_{i\in I}\arrow {s_i} {t_i})\wedge\neg(t_1\to t_2))\not\simeq\Empty }
          \\
        \Infer[Case]
        {\Gamma\vdash e:t_0\\
        \Gamma \evdash e t \Gamma_1 \\ \Gamma_1 \vdash e_1:t'\\
        \Gamma \evdash e {\neg t} \Gamma_2 \\ \Gamma_2 \vdash e_2:t'}
        {\Gamma\vdash \tcase {e} t {e_1}{e_2}: t'}
        { }
      \\
      \Infer[Efq]
      { }
      { \Gamma, (e:\Empty) \vdash e': t }
      { }
      \qquad
      \Infer[Proj]
  {\Gamma \vdash e:\pair{t_1}{t_2}}
  {\Gamma \vdash \pi_i e:t_i}
  { }
  \qquad
  \Infer[Pair]
  {\Gamma \vdash e_1:t_1 \and \Gamma \vdash e_2:t_2}
  {\Gamma \vdash (e_1,e_2):\pair {t_1} {t_2}}
  { }
    \end{mathpar}

    \begin{center} \line(1,0){300} \end{center}

    \begin{mathpar}
        \Infer[Base]
        { }
        { \Gamma \evdash e t \Gamma }
        { }
        \qquad
        \Infer[Path]
        { \pvdash {\Gamma'} e t \varpi:t' \\ \Gamma \evdash e t \Gamma' }
        { \Gamma \evdash e t \Gamma',(\occ e \varpi:t') }
        { }
    \end{mathpar}

    \begin{center} \line(1,0){300} \end{center}

    \begin{mathpar}
        \Infer[PSubs]
        { \pvdash \Gamma e t \varpi:t_1 \quad t_1\leq t_2 }
        { \pvdash \Gamma e t \varpi:t_2 }
        { }
        \quad
    \Infer[PInter]
        { \pvdash \Gamma e t \varpi:t_1 \\ \pvdash \Gamma e t \varpi:t_2 }
        { \pvdash \Gamma e t \varpi:t_1\land t_2 }
        { }
        \quad
    \Infer[PTypeof]
        { \Gamma \vdash \occ e \varpi:t' }
        { \pvdash \Gamma e t \varpi:t' }
        { }
        \\
    \Infer[PEps]
        { }
        { \pvdash \Gamma e t \epsilon:t }
        { }
        \qquad
    \Infer[PAppR]
        { \pvdash \Gamma e t \varpi.0:\arrow{t_1}{t_2} \\ \pvdash \Gamma e t \varpi:t_2'}
        { \pvdash \Gamma e t \varpi.1:\neg t_1 }
        { t_2\land t_2' \simeq \Empty  }
        \\
    \Infer[PAppL]
        { \pvdash \Gamma e t \varpi.1:t_1 \\ \pvdash \Gamma e t \varpi:t_2 }
        { \pvdash \Gamma e t \varpi.0:\neg (\arrow {t_1} {\neg t_2}) }
        { }
        \qquad
    \Infer[PPairL]
        { \pvdash \Gamma e t \varpi:\pair{t_1}{t_2} }
        { \pvdash \Gamma e t \varpi.l:t_1 }
        { }
        \\
    \Infer[PPairR]
        { \pvdash \Gamma e t \varpi:\pair{t_1}{t_2} }
        { \pvdash \Gamma e t \varpi.r:t_2 }
        { }
        \qquad
    \Infer[PFst]
        { \pvdash \Gamma e t \varpi:t' }
        { \pvdash \Gamma e t \varpi.f:\pair {t'} \Any }
        { }
        \qquad
    \Infer[PSnd]
        { \pvdash \Gamma e t \varpi:t' }
        { \pvdash \Gamma e t \varpi.s:\pair \Any {t'} }
        { }
        \qquad
    \end{mathpar}

    \newpage

    \subsection{Parallel semantics}\label{app:parallel}
    One technical difficulty in the proof of the subject reduction
    property is that, when reducing an expression $e$ into $v$ in a type
    case, the expression $e$ disappears ($e$ is not a
    sub-expression of the test anymore) and, thus, we can no longer refine the
    expression $e$ in the ``then'' and ``else'' branches (which might
    contain occurrences of $e$).  To circumvent this issue, we introduce a notion of parallel
    reduction which essentially reduces all occurrences of a
    sub-expression appearing in a type cases also in the ``then'' and
    ``else'' branch at the same time. 


    The idea is to label each step of reduction done by a context rule
    with the inner \textit{notion of reduction} (defined below) that caused the context to reduce. In
    case of a reduction of the expression tested in the type case,
    that same reduction is applied in parallel to both branches.
    The semantics based on parallel
    reduction is given below where expressions and values are defined as in Section~\ref{sec:syntax}. The contexts, however, are not exactly those in Section~\ref{sec:opsem} since there are two differences: $(i)$ we remove the test expression context, since this
    requires a specific rule (rule \Rule{$\tau\kappa$}) that performs the parallel reduction and $(ii)$ context holes are present only at top-level since the parallel reduction will handle the nesting of contexts by applying the rule \Rule{$\kappa$} below multiple times. This yields the following definition:\\[1mm]
   \[
      \begin{array}{lrcl}
      \textbf{Context} & \Cx[] & ::= & e [] \alt [] v \alt (e,[]) \alt ([],v) \alt \pi_i[]
      \end{array}
      \]
    For convenience, we denote $e\xleadsto{e\mapsto e'}e'$ by
    $e\idleadsto e'$ and by $e\uleadsto e'$ a step of reduction of the
    parallel semantics, regardless of the value on the top of the
    arrow.

\begin{mathpar}    
\textbf{Notions of reduction:}\\
  \Infer[$\beta$]
      { }
      {(\lambda^tx.e) v \idleadsto e\subst x v}
      {}
      \qquad
      \Infer[$\pi$]
      { }
      {\pi_i (v_1,v_2) \idleadsto v_i}
      {}\\
      \Infer[$\tau_1$]
      {\mbox{ }}
      {\ite v t {e_1} {e_2} \idleadsto e_1}
      {v \in \valsemantic t}
      \qquad
      \Infer[$\tau_2$]
      {\mbox{ }}
      {\ite v t {e_1} {e_2} \idleadsto e_2}
      {v \not\in \valsemantic t}\\
\end{mathpar}
\begin{mathpar}    
\textbf{Context reductions:}\\
      \Infer[$\kappa$]
      {e \xleadsto{e_r\mapsto e_r'} e'}
      {C[e] \xleadsto{e_r\mapsto e_r'} C[e']}
      {}\\
      \Infer[$\tau\kappa$]
      {e \xleadsto{e_r\mapsto e_r'} e'}
      {\ite e t {e_1} {e_2} \idleadsto \ite {e\subst{e_r}{e_r'}} t {e_1\subst {e_r} {e_r'}} {e_2\subst {e_r} {e_r'}}}{}
\end{mathpar}
where 
    \[\valsemantic t = \{v \alt \vvdash v : t\}\]
with
    \begin{mathpar}
      \Infer[Subsum]
          {\vvdash v:t'\\t'\leq t}
          {\vvdash v:t}
          {}
      \qquad
      \Infer[Const]
          { }
          {\vvdash c:\basic{c}}
          {}
      \\
      \Infer[Pair]
        {\vvdash v_1:t_1 \and \vvdash v_2:t_2}
        {\vvdash (v_1,v_2):\pair {t_1} {t_2}}
        { }
      \qquad
      \Infer[Abs]
          {t=(\wedge_{i\in I}\arrow {s_i} {t_i})\land (\wedge_{j\in J} \neg (\arrow {s'_j} {t'_j}))\\t\not\leq \Empty}
          {\vvdash\lambda^{\wedge_{i\in I}\arrow {s_i} {t_i}}x.e:t}
          {}
      \\
      \end{mathpar}

      Here is a couple of examples of reduction using the parallel semantics:

      \begin{mathpar}
        \Infer[$\tau\kappa$]
        {
          \Infer[$\kappa$]
          {\Infer[$\beta$] { } {(\lambda x.\ x+1)\ 1\idleadsto 2} {}}
          {((\lambda x.\ x+1)\ 1, \true)\xleadsto{(\lambda x.\ x+1)1\ \mapsto\ 2} (2, \true)}
          {}
        }
        {\ite {((\lambda x.\ x+1)\ 1, \true)} {\pair \Int \Bool} {(\lambda x.\ x+1)\ 1} {0} \idleadsto \ite {(2, \true)} {\pair \Int \Bool} {2} {0}}
        {}
      \end{mathpar}
      and
      \begin{mathpar}
        \qquad\qquad\Infer[$\tau_1$]
        {~}
        {\ite {(2, \true)} {\pair \Int \Bool} {2} {0} \idleadsto 2}
        {(2, \true) \in \valsemantic {\pair \Int \Bool}}
      \end{mathpar}
      Notice that the rule \Rule{$\kappa$} applies a substitution from an expression to an expressions (rather than from a variable to an expressions). This is formally defined as follows:
      \begin{definition}[Expression substitutions]
        Expression substitutions, ranged over by $\rho$, map an expression into another expression. The application of an expressions substitution $\rho$ to an expression $e$, noted $e\rho$ is the capture avoiding replacement defined as follows:
        \begin{itemize}
        \item If  $e'\equiv_\alpha e''$,  then $e''\subst{e'}e = e$.\vspace{1mm}
        \item If  $e'\not\equiv_\alpha e''$,  then $e''\subst{e'}e$ is inductively defined as \svvspace{-1.5mm} 
        \begin{align*}
          c\subst{e'}e   & = c\\
          x\subst{e'}{e} & =  x\\
          (e_1e_2)\subst{e'}{e} & =  (e_1\subst{e'}{e})(e_2\subst{e'}{e})\\
          (\lambda^{\wedge_{i\in I}s_i\to t_i} x.e)\subst{e'}{e}  & = \lambda^{\wedge_{i\in I}s_i\to t_i} x.(e\subst{e'}{e}) &\text{if } x\not\in\fv(e)\cup\fv(e')\\
          (\pi_i e)\subst{e'}{e}   & = \pi_i (e \subst{e'}{e})\\
         (e_1,e_2)\subst{e'}{e} & =  (e_1\subst{e'}{e},e_2\subst{e'}{e})\\    
         (\tcase{e_1}{t}{e_2}{e_3})\subst{e'}{e}  & = \tcase{e_1\subst{e'}{e}}{t}{e_2\subst{e'}{e}}{e_3\subst{e'}{e}}
        \end{align*}
        \end{itemize}
      \end{definition}
      Notice that the expression substitutions are up to
      alpha-renaming and perform only one pass.  For instance, if our
      substitution is $\rho=\subst {(\lambda^t x. x) y} {y}$, we have
      $((\lambda^t x. x)((\lambda^t z. z) y))\rho = (\lambda^t x. x)
      y$. The environments operate up to alpha-renaming, too.
   
      Finally notice that according to the definition above the rule \Rule{$\tau\kappa$} could be equivalently written as follows:
      \begin{mathpar}    
      \Infer[$\tau\kappa$]
      {e \xleadsto{\rho} e'}
      {\ite e t {e_1} {e_2} \idleadsto (\ite e t {e_1} {e_2})\rho}{}
      \end{mathpar}
      
      All the proofs below will use the parallel semantics instead of the standard semantics (of Section~\ref{sec:opsem}).
      However, the safety of the type system for the standard semantics can be deduced from the safety of the type system for the parallel semantics,
      using the following lemma:

      \begin{lemma}\label{semanticsimpl}
        $\forall e, v.\ e\uleadsto^* v \Rightarrow e\reduces^* v$
      \end{lemma}
      \begin{proof}
        This is a known result for the $\lambda-$calculus (even
        extended with conditional, and basic types), obtained using
        the Tait and Martin-Löf technique (\cite{Bar84}). See for instance
        \cite{Taka89} and \cite{Levy2017}.
        The additional substitutions made by the rule \Rule{$\tau\kappa$}
        will be performed later with the standard semantics.
      \end{proof}
      \subsection{Proofs for the declarative type system}\label{app:soundness}

      In this section, the only environments that we consider are well-formed environments (see definition below). We can easily
      check that every derivation only contains well-formed environments, provided that the initial judgment also use a well-formed environment.
      It is a consequence of the fact that rule \Rule{Case} requires $e$ to be typeable and that it only refines subexpressions of $e$.

        \subsubsection{Environments}

        \begin{definition}[Well-formed environment]
          We say that an environment $\Gamma$ is \emph{well-formed} if and only if
          $\forall e\in\dom\Gamma \text{ such that $e$ is not a variable}\ \exists t.\ \Gamma\setminus\{e\} \vdash e:t$.

          In other words, an environment can refine the type of an expression, but only if this expression is already typeable
          without this entry in the environment (possibly with a strictly weaker type than the one recorded in $\Gamma$).
        \end{definition}

        \begin{definition}[Bottom environment]
          Let $\Gamma$ be an environment.\\
          $\Gamma$ is bottom (noted $\Gamma = \bot$) if and only if $\exists e\in\dom\Gamma.\ \Gamma(e)\simeq\Empty$.
        \end{definition}

          \begin{definition}[(Pre)order on environments]
            Let $\Gamma$ and $\Gamma'$ be two environments. We write $\Gamma' \leq \Gamma$ if and only if:
            \begin{align*}
              &\Gamma'=\bot \text{ or } (\Gamma\neq\bot \text{ and } \forall e \in \dom \Gamma.\ \Gamma' \vdash e : \Gamma(e))
            \end{align*}
            This relation is a preorder (proof below).
          \end{definition}

          \begin{definition}[Application of a substitution to an environment]
            Let $\Gamma$ be an environment and $\rho$ a substitution from expressions to expressions.
            The environment $\Gamma\rho$ is defined by:
            \begin{align*}
              &\dom {\Gamma\rho} = \dom \Gamma \rho\\
              &\forall e \in \dom {\Gamma\rho}, (\Gamma\rho)(e) = \bigwedge_{\{e' \in \dom \Gamma \alt e'\rho\equiv e\}}\Gamma(e')
            \end{align*}
          \end{definition}

          \begin{definition}[Ordinary environments]
            We say that an environment $\Gamma$ is ordinary if and only if its domain only contains variables.
          \end{definition}

          \subsubsection{Subject Reduction}\label{app:subject-reduction}

            \begin{property}[$\valsemantic \_$ properties]
              \begin{align*}
                &\forall s.\ \forall t.\ \valsemantic s \subseteq \valsemantic t \Leftrightarrow s \leq t\\
                &\valsemantic \Empty = \varnothing\\
                &\forall t.\ \valsemantic {\neg t} = \values \setminus \valsemantic t\\
                &\forall s.\ \forall t.\ \valsemantic {s\vee t} = \valsemantic s \cup \valsemantic t
              \end{align*}
            \end{property}
            \begin{proof}
            See theorem 5.5, lemmas 6.19, 6.22, 6.23 of~\cite{Frisch2008}.
            \end{proof}

            \begin{lemma}[Alpha-renaming]
              Both the type system and the semantics are invariant by alpha-renaming.
            \end{lemma}
            \begin{proof}
            Straightforward.
            For the type system, it is a consequence of the fact that environments are up to alpha-renaming.
            For the semantics, it is a consequence of the fact that parallel substitutions in \Rule{$\tau\kappa$}
            are up to alpha-renaming.
            \end{proof}

            \begin{lemma}[Soundness and completeness of value typing]
              Let $v$ be a value, $t$ a type, and $\Gamma$ an environment.
              \begin{itemize}[nosep]
              \item If $\Gamma \vdash v:t$ and $\Gamma\neq\bot$, then $v \in \valsemantic{t}$.

              \item If $v \in \valsemantic{t}$ and $v$ is well-typed in $\Gamma$, then $\Gamma\vdash v:t$.
              \end{itemize}
            \end{lemma}
            \begin{proof}
            Immediate by definition of $\valsemantic{.}$.
            \end{proof}

            \begin{lemma}[Monotonicity]
              Let $\Gamma$ and $\Gamma'$ be two environments such that $\Gamma' \leq \Gamma$.
              Then, we have:
              \begin{align*}
                \forall e,t.\ &\Gamma \vdash e:t \Rightarrow \Gamma' \vdash e:t\\
                \forall e,t,\Gamma_1.\ &\Gamma \evdash e t \Gamma_1 \Rightarrow \exists {\Gamma_1}'\leq \Gamma_1.\ \Gamma' \evdash e t {\Gamma_1}'\\
                \forall e,t,\varpi,t'.\ &\pvdash \Gamma e t \varpi:t' \Rightarrow\ \pvdash {\Gamma'} e t \varpi:t'
              \end{align*}
            \end{lemma}
            \begin{proof}
            Immediate, by replacing every occurrence of rule \Rule{Env} in the
            derivation with $\Gamma$
            by the corresponding derivation with $\Gamma'$, followed by an
            application of rule \Rule{Subs} if needed.
            \end{proof}

            \begin{corollary}[Preorder relation]
              The relation $\leq$ on environments is a preorder.
            \end{corollary}

            \begin{lemma}[Value refinement 1]
              If we have $\pvdash \Gamma e t \varpi.x:t'$ with $x\in\{0,1,l,r,f,s\}$ (and $e$ well-typed in $\Gamma$) such that $\forall y.\ \occ e {\varpi.y}$ is a value
              and $v = \occ e {\varpi.x} \not\in \valsemantic{t'}$, we can derive $\pvdash \Gamma e t \varpi:\Empty$.
            \end{lemma}

            \begin{proof}
            We proceed by induction on the derivation of $\pvdash \Gamma e t \varpi.x:t'$.

            We perform a case analysis on the last rule:
            \begin{description}
              \item[\Rule{PTypeof}] In this case we have $\Gamma \vdash \occ e {\varpi.x} : t'$ with $v \not\in \valsemantic{t'}$.
              Thus we can derive $\Gamma \vdash \occ e {\varpi.x} : \Empty$ by using the rule \Rule{Inter} and the rules \Rule{Abs+}, \Rule{Abs-} or \Rule{Const}.

              Let us show that we also have $\Gamma \vdash \occ e \varpi : \Empty$.
              \begin{itemize}
                \item If $x=0$, we know that $\occ e \varpi$ is an application, and we can conclude easily given that $\Empty \leq \arrow \Any \Empty$.
                \item If $x=1$, we know that $\occ e \varpi$ is an application, and we can conclude easily given that $\arrow \Empty \Empty \simeq \arrow \Empty \Any$.
                \item If $x=f$ or $x=s$, we know that $\occ e \varpi$ is a projection, and we can conclude easily given that $\Empty \simeq \pair \Empty \Empty$.
                \item If $x=l$ or $x=r$, we know that $\occ e \varpi$ is a pair, and we can conclude easily given that $\pair \Empty \Any \simeq \pair \Any \Empty \simeq \Empty$.
              \end{itemize}
              Hence we can derive $\Gamma \vdash \occ e \varpi : \Empty$.

              \item[\Rule{PInter}] We must have $v \not\in \valsemantic{t_1 \land t_2}$. It implies $v \not\in \valsemantic{t_1} \cap \valsemantic{t_2}$
              and thus $v \not\in \valsemantic{t_1}$ or $v \not\in \valsemantic{t_2}$. Hence, we can conclude just by applying the induction hypothesis.
              \item[\Rule{PSubs}] Trivial (we use the induction hypothesis).
              \item[\Rule{PEps}] This case is impossible.
              \item[\Rule{PAppL}] We have $v \not\in \valsemantic{\neg (\arrow {t_1} {\neg t_2})}$. Thus, we have $v \in \valsemantic{\arrow {t_1} {\neg t_2}}$
              and in consequence we can derive $\Gamma \vdash v:\arrow {t_1} {\neg t_2}$ (because $e$ is well-typed in $\Gamma$).

              Recall that $\occ e {\varpi.1}$ is necessarily a value (by hypothesis).
              By using the induction hypothesis on $\pvdash \Gamma e t \varpi.1:t_1$, we can suppose $\occ e {\varpi.1} \in \valsemantic{t_1}$ (otherwise, we can conclude directly).
              Thus, we can derive $\Gamma \vdash \occ e {\varpi.1}: t_1$.

              From $\Gamma \vdash v:\arrow {t_1} {\neg t_2}$ and $\Gamma \vdash \occ e {\varpi.1}: t_1$, we can derive $\Gamma \vdash \occ e {\varpi}: \neg t_2$
              using the rule \Rule{App}.

              Now, by starting from the premise $\pvdash \Gamma e t \varpi:t_2$ and using the rules \Rule{PInter} and \Rule{PTypeof}, we can derive
              $\pvdash \Gamma e t \varpi:\Empty$.

              \item[\Rule{PAppR}] We have $v \not\in \valsemantic{\neg t_1}$. Thus, we have $v \in \valsemantic{t_1}$ and in consequence
              we can derive $\Gamma \vdash v:t_1$.

              Recall that $\occ e {\varpi.0}$ is necessarily a value (by hypothesis).
              By using the induction hypothesis on $\pvdash \Gamma e t \varpi.0:\arrow {t_1} {t_2}$, we can suppose $\occ e {\varpi.0} \in \valsemantic{\arrow {t_1} {t_2}}$ (otherwise, we can conclude directly).
              Thus, we can derive $\Gamma \vdash \occ e {\varpi.0}: \arrow {t_1} {t_2}$ (because $e$ is well-typed in $\Gamma$).

              From $\Gamma \vdash v:t_1$ and $\Gamma \vdash \occ e {\varpi.0}: \arrow {t_1} {t_2}$, we can derive $\Gamma \vdash \occ e {\varpi}: t_2$
              using the rule \Rule{App}.

              Now, by starting from the premise $\pvdash \Gamma e t \varpi:t_2'$ and using the rules \Rule{PInter} and \Rule{PTypeof}, we can derive
              $\pvdash \Gamma e t \varpi:\Empty$.

              \item[\Rule{PPairL}] We have $v \not\in \valsemantic{t_1}$. Thus, we have $v \in \valsemantic{\neg t_1}$
              and in consequence we can derive $\Gamma \vdash v:\neg t_1$.

              Hence, we can derive $\Gamma \vdash \occ e \varpi:\pair {\neg t_1} \Any$ ($e$ is well-typed in $\Gamma$).

              Now, by starting from the premise $\pvdash \Gamma e t \varpi:\pair {t_1} {t_2}$ and using the rules \Rule{PInter} and \Rule{PTypeof},
              we can derive $\pvdash \Gamma e t \varpi:\Empty$.

              \item[\Rule{PPairR}] Similar to the previous case.

              \item[\Rule{PFst}] We have $v \not\in \valsemantic{\pair {t'} \Any}$.
              As we also have $v \in \valsemantic{\pair \Any \Any}$ (because $e$ is well-typed in $\Gamma$),
              we can deduce $v \in \valsemantic{\pair {(\neg t')} \Any}$.

              Hence, we can derive $\Gamma \vdash v:\pair {(\neg t')} \Any$ and then $\Gamma \vdash \occ e \varpi:\neg t'$.

              Now, by starting from the premise $\pvdash \Gamma e t \varpi:t'$ and using the rules \Rule{PInter} and \Rule{PTypeof},
              we can derive $\pvdash \Gamma e t \varpi:\Empty$.
              \item[\Rule{PSnd}] Similar to the previous case.
            \end{description}
            \end{proof}

            \begin{corollary}[Value refinement 2]
              For any derivable judgement of the form $\Gamma \evdash e t \Gamma'$ (with $e$ well-typed in $\Gamma$), we can construct a derivation of $\Gamma \evdash e t \Gamma''$ with $\Gamma''\leq\Gamma'$
              that never uses the rule \Rule{Path} on a path $\varpi.x$ such that $\forall y.\ \occ e {\varpi.y}$ refers to a value.
            \end{corollary}

            \begin{proof}
            We can easily remove every such rule from the derivation. If $\occ e {\varpi.x} \in \valsemantic{t'}$, the \Rule{Path} rule is useless
            and we can freely remove it. Otherwise, if $\occ e {\varpi.x} \not\in \valsemantic{t'}$, we can use the previous lemma to
            replace it with a \Rule{Path} rule on $\varpi$.
            \end{proof}

            \begin{lemma}[Value testing]
              For any derivable judgement of the form $\Gamma \evdash v t \Gamma'$ (with $v$ a value),
              we have $v \in \valsemantic{t} \Rightarrow \Gamma\leq\Gamma'$.
            \end{lemma}

            \begin{proof}
            As $v$ is a value, the applications of \Rule{Path} have a path $\varpi$ only composed of $l$ and $r$
            and such that $\occ e \varpi$ is a value.

            Thus, any derivation $\pvdash \Gamma v t \varpi:t'$ can only contains the rules
            \Rule{PTypeof}, \Rule{PInter}, \Rule{PSubs}, \Rule{PEps}, \Rule{PPairL} and \Rule{PPairR}.

            Moreover, as $v \in \valsemantic{t}$, the rules \Rule{PEps} can be replaced by a \Rule{PTypeof}.
            Thus we can easily derive $\Gamma \vdash v:t'$ (we replace \Rule{PTypeof} by \Rule{Typeof},
            \Rule{PInter} by \Rule{Inter}, etc.).
            \end{proof}

            \begin{lemma}[Substitution]
              Let $\Gamma$ be an environment. Let $e_a$ and $e_b$ be two expressions.

              Let us suppose that $e_b$ is closed and that $e_a$ has one of the following form:
              \begin{itemize}
                \item $x$ (variable)
                \item $\ite e t {e_1} {e_2}$ (if expression)
                \item $v$ (value)
                \item $v v$ (application of two values)
                \item $(v,v)$ (product of two values)
              \end{itemize}
              Let us also suppose that $\forall t.\ \Gamma \vdash e_a : t \Rightarrow \Gamma\subst {e_a} {e_b} \vdash e_b:t$.

              Then, by noting $\rho = \subst {e_a} {e_b}$ we have:
              \begin{align*}
                &\forall e,t.\ \Gamma \vdash e:t \Rightarrow \Gamma\rho \vdash e\rho:t
              \end{align*}
            \end{lemma}

            \begin{proof}
            Let $\Gamma$, $e_a$, $e_b$ be as in the statement.

            We note $\rho$ the substitution $\subst {e_a} {e_b}$.

            We consider a derivation of $\Gamma \vdash e:t$.

            By using the value refinement lemma, we can assume without loss of generality that our derivation does not contain
            any rule \Rule{Path} on a path $\varpi.x$ such that $\forall y.\ \occ e {\varpi.y}$ refers to a value.

            We can also assume w.l.o.g. that every application of the \Rule{Path} rule is such that $\Gamma',(\occ e \varpi:t') \leq \Gamma'$. If it is not the case,
            we can easily transform the derivation by intersecting $t'$ with $\Gamma'(\occ e \varpi)$
            using the rules \Rule{PInter}, \Rule{PTypeof} and \Rule{Env}.
            The rest of the derivation can easily be adapted by adding some \Rule{Subs} rules when needed.

            Finally, we can assume that, in any environment appearing in the derivation, if the environnement is not bottom,
            then a value $v$ can only be mapped to a type $t$ such that $v\in\valsemantic{t}$. If it is not the case, then we just have to change the
            \Rule{Path} rule that introduce $(v:t)$ into a path rule that introduce $(v:\Empty)$,
            by using the rules \Rule{PInter} and \Rule{PTypeof} (if $v\not\in\valsemantic{t}$, then $v\in\valsemantic{\neg t}$
            and thus $\Gamma \vdash v:\neg t$ is derivable).

            Now, let's prove by induction on the derivation the following properties:

            \begin{align*}
              &\forall e,t.\ \Gamma \vdash e:t \Rightarrow \Gamma\rho \vdash e\rho:t\\
              &\forall e,t,\Gamma'.\ \Gamma \evdash {e} {t} \Gamma' \Rightarrow \Gamma\rho \evdash {e\rho} {t} {\Gamma'}\rho
              \text{ and we still have } \forall t.\ \Gamma' \vdash e_a : t \Rightarrow \Gamma'\rho \vdash e_b:t\\
              &\forall e,t,\varpi,t' \text{ s.t. $\occ {e\rho} \varpi$ is defined}.\ \pvdash \Gamma e t \varpi:t' \Rightarrow \pvdash {\Gamma\rho} {e\rho} t \varpi:t'
            \end{align*}

            We proceed by case analysis on the last rule of the derivation at the left of the $\Rightarrow$ in order to construct the derivation at the right.

            If the last judgement is of the form $\Gamma \vdash e_a: t$, then we can directly conclude with the hypotheses of the lemma.
            Thus, we can suppose it is not the case.

            There are many cases depending on the last rule:

            \begin{description}
              \item[\Rule{Env}] If $e\in\dom\Gamma$, then we have $e\rho\in\dom{\Gamma\rho}$ and $(\Gamma\rho)(e\rho)\leq\Gamma(e)$.
              Thus we can easily derive $\Gamma\rho\vdash e\rho:t$ with the rule \Rule{Env} and \Rule{Subs}.
              \item[\Rule{Efq}] If there exists $e\in\dom\Gamma$ such that $\Gamma(e)=\Empty$, then $(\Gamma\rho)(e\rho)=\Empty$
              so we can easily derive $\Gamma\rho\vdash e\rho:t$ with the rule \Rule{Efq}.
              \item[\Rule{Inter}] Trivial (by using the induction hypothesis).
              \item[\Rule{Subs}] Trivial (by using the induction hypothesis).
              \item[\Rule{Const}] In this case, $c\rho = c$ (because $c \neq e_a$). Thus it is trivial.
              \item[\Rule{App}] We have $(e_1 e_2)\rho = (e_1\rho) (e_2\rho)$ (because $e_1 e_2 \neq e_a$).
              Thus we can directly conclude by using the induction hypothesis.
              \item[\Rule{Abs+}] We have $(\lambda^{t'}x.e)\rho = \lambda^{t'}x.(e\rho)$ (because $\lambda^{t'}x.e \neq e_a$).

              By alpha-renaming, we can suppose that the variable $x$ is a new fresh variable that does not appear
              in $e_a$ nor $e_b$ ($e_b$ is closed).

              We can thus use the induction hypothesis on all the judgements $\Gamma, x:s_i \vdash e:t_i$.
              \item[\Rule{Abs-}] Trivial (by using the induction hypothesis).
              \item[\Rule{Proj}] We have $(\pi_i e)\rho = \pi_i (e\rho)$ (because $\pi_i e \neq e_a$).
              Thus we can directly conclude by using the induction hypothesis.
              \item[\Rule{Pair}] We have $(e_1,e_2)\rho = (e_1\rho,e_2\rho)$ (because $(e_1,e_2) \neq e_a$).
              Thus we can directly conclude by using the induction hypothesis.
              \item[\Rule{Case}]
              We have $(\ite e {t_{if}} {e_1} {e_2})\rho = \ite {e\rho} {t_{if}} {e_1\rho} {e_2\rho}$ (because $\ite e {t_{if}} {e_1} {e_2} \neq e_a$).

              We apply the induction hypothesis on the judgements $\Gamma\vdash e:t_0$ and $\Gamma\evdash e {t_{if}} \Gamma_1$.
              We get $\Gamma\rho\vdash e\rho:t_0$, $\Gamma\rho\evdash {e\rho} {t_{if}} \Gamma_1\rho$ and
              $\forall t'.\ \Gamma_1 \vdash e_a : t' \Rightarrow \Gamma_1\rho \vdash e_b:t'$.
              Now, we can apply the induction hypothesis on $\Gamma_1\vdash e_1:t$ and we have $\Gamma_1\rho\vdash e_1\rho:t$.

              We proceed similarly on the judgments $\Gamma\evdash e {\neg t_{if}} \Gamma_2$ and $\Gamma_2\vdash e_2:t$, and so we have all the premises
              to apply the \Rule{Case} rule in order to get $\Gamma\rho \vdash \ite {e\rho} {t_{if}} {e_1\rho} {e_2\rho}:t'$.

              \item[\Rule{Base}] Trivial.
              \item[\Rule{Path}] We have by using the induction hypothesis $\Gamma\rho \evdash {e\rho} t \Gamma'\rho$
              and $\forall t''.\ \Gamma' \vdash e_a : t'' \Rightarrow \Gamma'\rho \vdash e_b:t''$.

              First, let's show that we can derive $\Gamma\rho \evdash {e\rho} {t} {\Gamma''}\rho$ with $\Gamma''=\Gamma',(\occ e \varpi:t')$.

              There are two cases:
              \begin{itemize}
                \item $\occ e \varpi$ is a strict sub-expression of $e_a$.

                In this case, it means that among its three possible forms,
                $e_a$ is of the form $v v$ or $(v,v)$.
                According to the assumptions we made on the derivation at the beginning of the proof,
                it implies that $\varpi=\epsilon$.
                Hence, $e$ does not contain any occurrence of $e_a$, so it is easy to conclude.

                \item $\occ e \varpi$ is not a strict sub-expression of $e_a$.

                In this case, we know that $\occ {e\rho} \varpi$ is defined.

                Thus we can apply the induction hypothesis on $\pvdash {\Gamma'} {e} t \varpi:t'$.
                It gives $\pvdash {\Gamma'\rho} {e\rho} t \varpi:t'$.
                If $\occ {e\rho}\varpi \in\dom{\Gamma'\rho}$, and $(\Gamma'\rho)(\occ {e\rho}\varpi)=t''\not\geq t'$,
                then we can derive $\pvdash {\Gamma'\rho} {e\rho} t \varpi:t'\land t''$ just by using the rules
                \Rule{PInter}, \Rule{PTypeof} and \Rule{Env}.

                Using this last judgement together with $\Gamma \evdash e t \Gamma'$, we can derive with the rule \Rule{Path}
                the wanted $\Gamma\rho \evdash {e\rho} {t} {\Gamma''}\rho$.
              \end{itemize}

              Now, let's show that $\forall t'.\ \Gamma'' \vdash e_a : t' \Rightarrow \Gamma''\rho \vdash e_b:t'$.

              Let $t'$ be such that $\Gamma''\vdash e_a:t'$.

              Recall that we have $\Gamma' \vdash e_a : t' \Rightarrow \Gamma'\rho \vdash e_b:t'$.

              If $\Gamma'' = \bot$, then $\Gamma''\rho = \bot$ so we are done. So lets us suppose $\Gamma'' \neq \bot$.

              Let us separate the proof in two cases:
              \begin{itemize}
                \item If $\occ e \varpi \not\equiv e_a$. In this case, let's show that we have  $\Gamma'\vdash e_a:t'$.
                Indeed, in the typing derivation of $\Gamma''\vdash e_a:t'$, the \Rule{Env} rules can only be applied on
                subexpressions of $e_a$.

                If $\occ e \varpi$ is not a strict subexpression of $e_a$
                (and thus not a subexpression as $\occ e \varpi \not\equiv e_a$), there is no \Rule{Env} rule applied to $\occ e \varpi$
                in the derivation of $\Gamma''\vdash e_a:t'$ and thus we can easily derive $\Gamma'\vdash e_a:t'$.

                If $\occ e \varpi$, is a strict sub-expression of $e_a$, it must be a value (given the possible forms of $e_a$).
                Moreover, as $\Gamma'' \neq \bot$, we have $\forall v\in\dom{\Gamma''}.\ v\in\valsemantic{\Gamma''(v)}$ (recall the assumptions at the beginning of the proof)
                and thus $\forall v\in\dom{\Gamma''}.\ \Gamma'\vdash v : \Gamma''(v)$.
                Thus we can derive $\Gamma'\vdash e_a:t'$ just by replacing every \Rule{Env} rule applied to $\occ e \varpi$ in the derivation of $\Gamma''\vdash e_a:t'$
                by the relevant derivation.

                From $\Gamma'\vdash e_a:t'$ we deduce $\Gamma'\rho\vdash e_b:t'$.
                As $\Gamma''\leq\Gamma'$ (according to the assumptions we made on the derivation at the beginning of the proof)
                and $\dom{\Gamma'}\subseteq\dom{\Gamma''}$, we have $\Gamma''\rho\leq\Gamma'\rho$ and thus, by monotonicity,
                $\Gamma''\rho\vdash e_b:t'$.

                \item
                If $\occ e \varpi \equiv e_a$. Let us note $t_a=\Gamma''(e_a)$. This time,
                we can't derive $\Gamma'\vdash e_a:t'$ from $\Gamma''\vdash e_a:t'$ because the rule \Rule{Env}
                could be used on $\occ e \varpi=e_a$ (which may not be a value).

                However, the rule \Rule{Env} can only be used on $e_a$ at the end of the derivation of $\Gamma''\vdash e_a:t'$:
                there can't be any \Rule{App}, \Rule{Abs+}, \Rule{Proj}, \Rule{Pair} or \Rule{Case} after because the premises of these rules only contain strict sub-expressions of their
                consequence.
                Thus, we can easily transform the derivation so that every \Rule{Env} applied on $e_a$ is directly followed by an \Rule{Inter}:
                if there is any \Rule{Abs-} or \Rule{Subs} between, we can move it after.

                Then, we can (temporarily) remove from the derivation all \Rule{Env} applied on $e_a$:
                for each, we just replace the following \Rule{Inter} rule by its other premise.

                It yields a derivation for $\Gamma'' \vdash e_a : t''$ such that $t''\land t_a \leq t'$ and without any \Rule{Env} applied to $e_a$.
                Thus, we can transform it into a derivation of $\Gamma' \vdash e_a : t''$ as in the previous point, and we get $\Gamma'\rho \vdash e_b : t''$.
                Still as before, we get a derivation for $\Gamma''\rho\vdash e_b : t''$ by monotonicity.

                Now, we can append at the end of this derivation a rule \Rule{Inter} with a rule \Rule{Env} applied to $e_b$.
                As $(\Gamma''\rho)(e_b) \leq \Gamma''(e_a) = t_a$, we obtain a derivation for $\Gamma''\rho\vdash e_b : t'$ (we can add a final \Rule{Subs} rule if needed).
              \end{itemize}

              \item[\Rule{PTypeof}] Trivial (by using the induction hypothesis).
              \item[\Rule{P$\cdots$}] All the remaining rules are trivial.
            \end{description}
            \end{proof}

            \begin{theorem}[Subject reduction]
              Let $\Gamma$ be an ordinary environment, $e$ and $e'$ two expressions and $t$ a type.
              If $\Gamma\vdash e:t$ and $e\uleadsto e'$, then $\Gamma\vdash e':t$.
            \end{theorem}

            \begin{proof}
            Let $\Gamma$, $e$, $e'$ and $t$ be as in the statement.

            We construct a derivation for $\Gamma \vdash e':t$ by induction on the derivation of $\Gamma \vdash e:t$.

            If $\Gamma = \bot$ this theorem is trivial, so we can suppose $\Gamma \neq \bot$.

            We proceed by case analysis on the last rule of the derivation:

            \begin{description}
              \item[\Rule{Env}] As $\Gamma$ is ordinary, it means that $e$ is a variable.
              It contradicts the fact that $e$ reduces to $e'$ so this case is impossible.
              \item[\Rule{Efq}] This case is impossible as $\Gamma \neq \bot$.
              \item[\Rule{Inter}] Trivial (by using the induction hypothesis).
              \item[\Rule{Subs}] Trivial (by using the induction hypothesis).
              \item[\Rule{Const}] Impossible case (no reduction possible).
              \item[\Rule{App}] In this case, $e\equiv e_1 e_2$. There are three possible cases:
              \begin{itemize}
                \item $e_2$ is not a value. In this case, we must have $e_2\uleadsto e_2'$
                and $e'\equiv e_1 e_2'$. We can easily conclude using the induction hypothesis.
                \item $e_2$ is a value and $e_1$ is not. In this case, we must have $e_1\uleadsto e_1'$
                and $e'\equiv e_1' e_2$. We can easily conclude using the induction hypothesis.
                \item Both $e_1$ and $e_2$ are values. This is the difficult case.
                We have $e_1\equiv \lambda^{\bigwedge_{i\in I}\arrow{s_i}{t_i}}x.e_x$
                with $\bigwedge_{i\in I}\arrow{s_i}{t_i} \leq \arrow{s}{t}$ and $\Gamma \vdash e_2:s$.
                We can suppose that $x$ is a new fresh variable that does not appear in our environment
                (if it is not the case, we can alpha-rename $e_1$).

                This means that $s\leq \bigvee_{i\in I} s_i$ and that for any non-empty $I'$ such that
                $s\not\leq \bigvee_{i\in I\setminus I'} s_i$, we have $\bigwedge_{i\in I'} t_i \leq t$
                (see lemma 6.8 of~\cite{Frisch2008}). Let us take $I'=\{i\in I\alt e_2\in\valsemantic{s_i}\}$.
                We have $I'$ not empty: $e_2\in\valsemantic{s}$ and $s\leq \bigvee_{i\in I} s_i$, so according to
                $\valsemantic \_$ properties we have at least one $i$ such that $e_2\in\valsemantic{s_i}$.
                We also have $s\not\leq \bigvee_{i\in I\setminus I'} s_i$, otherwise there would be a $i\not\in I$
                such that $e_2\in\valsemantic{s_i}$ (contradiction with the definition of $I'$).
                As a consequence, we get $\bigwedge_{i\in I'} t_i \leq t$.

                Now, let's prove that $\Gamma \vdash e':\bigwedge_{i\in I'} t_i$ (which, by subsumption,
                yields $\Gamma \vdash e': t$). For that, we show that for any $i\in I'$, $\Gamma \vdash e':t_i$
                (it is then easy to conclude by using the \Rule{Inter} rule).

                Let $i\in I'$. We have $e_2\in\valsemantic{s_i}$, and so $\Gamma \vdash e_2:s_i$ ($e_2$ is well-typed in $\Gamma$).
                As $e_1$ is well-typed in $\Gamma$, there must be in its derivation an application of the rule \Rule{Abs+}
                which guarantees $\Gamma,(x:s_i) \vdash e_x:t_i$ (recall that $\Gamma \neq \bot$ and $\Gamma$ is ordinary so there is no abstraction in $\dom\Gamma$).
                Let us note $\Gamma'=\Gamma,(x:s_i)$.
                We can deduce, using the substitution lemma, that $\Gamma'\subst x {e_2} \vdash e_x\subst x {e_2}: t_i$.

                Moreover, $\Gamma'\subst x {e_2} = \Gamma,(e_2:s_i)$ and $\Gamma \leq \Gamma,(e_2:s_i)$.
                Thus, by monotonicity, we deduce $\Gamma \vdash e_x\subst x {e_2}: t_i$,
                that is $\Gamma \vdash e': t_i$.
              \end{itemize}
              \item[\Rule{Abs+}] Impossible case (no reduction possible).
              \item[\Rule{Abs-}] Impossible case (no reduction possible).
              \item[\Rule{Proj}] In this case, $e\equiv \pi_i e_0$. There are two possible cases:
              \begin{itemize}
                \item $e_0$ is not a value. In this case, we must have $e_0\uleadsto e_0'$
                and $e'\equiv \pi_i e_0'$. We can easily conclude using the induction hypothesis.
                \item $e_0$ is a value.
                Given that $e_0 \leq \pair \Any \Any$, we have $e_0 = (v_1,v_2)$ with $v_1$ and $v_2$ two values.
                We also have $e \idleadsto v_i$.
                The derivation of $\Gamma \vdash (v_1,v_2): \pair {t_1} {t_2}$ must contain a rule \Rule{Pair}
                which guarantees $\Gamma \vdash v_i: t_i$ (recall that $\Gamma \neq \bot$ and $\Gamma$ is ordinary so there is no pair in $\dom\Gamma$).
                It concludes this case.
              \end{itemize}
              \item[\Rule{Pair}] In this case, $e\equiv (e_1,e_2)$. There are two possible cases:
              \begin{itemize}
                \item $e_2$ is not a value. In this case, we must have $e_2\uleadsto e_2'$
                and $e'\equiv (e_1, e_2')$. We can easily conclude using the induction hypothesis.
                \item $e_2$ is a value and $e_1$ is not. In this case, we must have $e_1\uleadsto e_1'$
                and $e'\equiv (e_1', e_2)$. We can easily conclude using the induction hypothesis.
              \end{itemize}
              \item[\Rule{Case}] In this case, $e\equiv \ite {e_0} {t_{if}} {e_1} {e_2}$. There are three possible cases:
              \begin{itemize}
                \item $e_0$ is a value and $e_0 \in \valsemantic{t_{if}}$. In this case we have $e' \equiv e_1$.
                We have derivations for $\Gamma \vdash e_0: t_0$, $\Gamma \evdash {e_0} {t_{if}} \Gamma'$ and $\Gamma'\vdash e_1:t$.

                As $e_0$ is a value and $e_0 \in \valsemantic{t_{if}}$, we have $\Gamma\leq\Gamma'$ by using the value testing lemma.
                Thus, by monotonicity, we have $\Gamma\vdash e_1:t$.
                \item $e_0$ is a value and $e_0 \not\in \valsemantic{t}$. This case is similar to the previous one (we replace $t_{if}$ by $\neg t_{if}$ and $e_1$ by $e_2$).
                \item $e_0$ is not a value.
                In this case, we have $e_0\xleadsto{e_a\mapsto e_b} e_0'$ and $e'\equiv \ite {e_0\rho} {t_{if}} {e_1\rho} {e_2\rho}\equiv e\rho$
                with $\rho = \subst{e_a}{e_b}$.

                First, let's notice that we have $e_b$ closed (only closed expressions are reducible),
                and $e_a$ has one of the following forms:
                \begin{itemize}
                  \item $\ite e t {e_1} {e_2}$ (if expression)
                  \item $v v$ (application of two values)
                  \item $(v,v)$ (product of two values)
                \end{itemize}
                It can be easily proved by induction on the derivation of the reduction step.
                Secondly, as $e_a\uleadsto e_b$ and as the derivation of this reduction is a strict subderivation of that of $e\uleadsto e'$,
                we can use the induction hypothesis on $e_a\uleadsto e_b$ and we obtain $\forall t'.\ \Gamma \vdash e_a : t' \Rightarrow \Gamma\rho \vdash e_b:t'$.
                Thus, we can conclude directly by using the substitution lemma on $e$ and $\rho$.
              \end{itemize}
            \end{description}
            \end{proof}

            \subsubsection{Progress}

            \begin{lemma}[Inversion]
              \begin{align*}
                &\valsemantic{\pair {t_1} {t_2}} = \{(v_1,v_2) \alt \vvdash v_1:t_1, \vvdash v_2:t_2 \}\\
                &\valsemantic{b} = \{c \alt \basic{c} \leq b \}\\
                &\valsemantic{\arrow t s} = \{\lambda^{\bigwedge_{i\in I}\arrow {t_i} {s_i}}x.e \alt \bigwedge_{i\in I}\arrow {t_i} {s_i} \leq \arrow t s\}
              \end{align*}
            \end{lemma}

            \begin{proof}
            See lemma 6.21 of~\cite{Frisch2008}
            \end{proof}

            \begin{theorem}[Progress]
            If $\varnothing \vdash e:t$, then either $e$ is a value or there exists $e'$ such that $e \uleadsto e'$.
            \end{theorem}

            \begin{proof}
            We proceed by induction on the derivation $\varnothing \vdash e:t$.
            We consider the last rule of this derivation:

            \begin{description}
              \item[\Rule{Env}] This case is impossible (the environment is empty).
              \item[\Rule{Efq}] This case is impossible (the environment is empty).
              \item[\Rule{Inter}] Straightforward application of the induction hypothesis.
              \item[\Rule{Subs}] Straightforward application of the induction hypothesis.
              \item[\Rule{Const}] In this case, $e$ must be a constant so $e$ is a value.
              \item[\Rule{App}] We have $e=e_1\ e_2$, with $\varnothing \vdash e_1: \arrow s t$
              and $\varnothing \vdash e_2 : s$. If one of the $e_i$ can be reduced, then
              $e$ can also be reduced using the reduction rule \Rule{$\kappa$}.

              Otherwise, by using the induction hypothesis we get that both $e_1$ and $e_2$ are values.
              Moreover, by using the inversion lemma, we know that $e_1$ has the form
              $\lambda^{\bigwedge_{i\in I} \arrow {t_i} {s_i}}x.e_0$. In consequence, $e$ is reducible
              (the reduction rule \Rule{$\beta$} can be applied).
              \item[\Rule{Abs+}] In this case, $e$ must be a lambda abstraction, so $e$ is a value.
              \item[\Rule{Abs-}] Straightforward application of the induction hypothesis.
              \item[\Rule{Case}] We have $e = \ite {e_0} {t'} {e_1} {e_2}$. If $e_0$ can be reduced,
              then $e$ can also be reduced using the reduction rule \Rule{$\tau\kappa$}.

              Otherwise, by using the induction hypothesis we get that $e_0$ is a value.
              In consequence, $e$ is reducible (the reduction rule \Rule{$\tau_i$} can be applied).
              \item[\Rule{Proj}] We have $e=\pi_i e_0$, $t=t_i$, $\varnothing \vdash e_0: \pair {t_1} {t_2}$.
              If $e_0$ can be reduced, then $e$ can also be reduced using the rule \Rule{$\kappa$}.

              Otherwise, by using the induction hypothesis we get that $e_0$ is a value.
              Moreover, by using the inversion lemma, we know that $e_0$ has the form $(v_1, v_2)$.
              In consequence, $e$ is reducible (the reduction rule \Rule{$\pi$} can be applied).
              \item[\Rule{Pair}] We have $e=(e_1,e_2)$. If one of the $e_i$ can be reduced, then
              $e$ can also be reduced using the reduction rule \Rule{$\kappa$}.

              Otherwise, by using the induction hypothesis we get that both $e_1$ and $e_2$ are values.
              In consequence, $e$ is also a value.
            \end{description}
            \end{proof}

\newpage

\section{Typing Algorithm: Operators, Type Schemes, Proofs of Soundness and Completeness}
\label{sec:proofs-algo}
We  give in this section a typing algorithm that uses types schemes and is more general than the one presented in the main
body of the paper (that is, one whose completeness is not limited to
positive expressions). We start by defining how to compute the ``$\worra{}{}$'' operator and then we define type schemes and the algorithm. We prove that this algorithm (as well as the
one in the main body of the paper) is sound w.r.t.\ to the declarative type system
and that, under certain restrictions it is also complete.

            \subsection{Operator $\worra {} {}$}\label{app:worra}

            In this section, we will use the algorithmic definition of $\worra {} {}$ and show that it is equivalent to its
            descriptive definition.

            \begin{eqnarray*}
                &t \simeq \bigvee_{i\in I}\left(\bigwedge_{p\in P_i}(s_p\to t_p)\bigwedge_{n\in N_i}\neg(s_n'\to t_n')\right)\\
                &\worra t s = \dom t \wedge\bigvee_{i\in I}\left(\bigwedge_{\{P \subseteq P_i\alt s \leq \bigvee_{p \in P} \neg t_p\}} \left(\bigvee_{p \in P} \neg s_p \right)\right)
            \end{eqnarray*}

            \begin{lemma}[Correctness of $\worra {} {}$]
                $\forall t, s.\ \apply t {(\dom t \setminus (\worra t s))} \leq \neg s$
            \end{lemma}

            \begin{proof}
            Let $t$ an arrow type. $t \simeq \bigvee_{i\in I}\left(\bigwedge_{p\in P_i}(s_p\to t_p)\bigwedge_{n\in N_i}\neg(s_n'\to t_n')\right)$\\
            Let $s$ be any type.

            Let's prove that $\apply t {(\dom t \setminus (\worra t s))} \leq \neg s$ (with the algorithmic definition for $\worra {} {}$).\\
            Equivalently, we want $(\apply t {(\dom t \setminus (\worra t s))}) \land s \simeq \Empty$.\\

            Let $u$ be a type such that $u \leq \dom t$ and $(\apply t u) \land s \not\simeq \Empty$ (if such a type does not exist, we are done).\\
            Let's show that $u \land (\worra t s) \not\simeq \Empty$ (we can easily deduce the wanted property from that, by the absurd).\\
            For that, we should prove the following:\\

            \[
                \exists i \in I.\
                u \land \bigwedge_{\{P \subseteq P_i\alt s \leq \bigvee_{p \in P} \neg t_p\}} \left(\bigvee_{p \in P} \neg s_p \right) \not\simeq \Empty
            \]

          From $(\apply t u) \land s \not\simeq \Empty$, we can take (using the algorithmic definition of $\circ$) $i \in I$ and $Q \subsetneq P_i$ such that:\\
          \[ u \not\leq\bigvee_{q\in Q}s_q \text{\ \ \ and\ \ \ } (\bigwedge_{p\in P_i\setminus Q}t_p) \land s \not\simeq \Empty \]

            For any $P \subseteq P_i$ such that $s \leq \bigvee_{p \in P} \neg t_p$ (equivalently, $s \land \bigwedge_{p \in P} t_p \simeq \Empty$),\\
            we have $P \cap Q \neq \varnothing$ (by the absurd, because $(\bigwedge_{p\in P_i\setminus Q}t_p) \land s \not\simeq \Empty$).\\
            Consequently, we have:
            \[ \forall P \subseteq P_i.\ s \leq \bigvee_{p \in P} \neg t_p \Rightarrow \bigwedge_{p \in P} s_p \leq \bigvee_{q \in Q}s_q \]
            We can deduce that:
            \[ \bigvee_{\{P \subseteq P_i\alt s \leq \bigvee_{p \in P} \neg t_p\}} \left(\bigwedge_{p \in P} s_p\right) \leq \bigvee_{q \in Q}s_q \]
            Moreover, as $u \not\leq \bigvee_{q\in Q}s_q$, we have $u \not\leq \bigvee_{\{P \subseteq P_i\alt s \leq \bigvee_{p \in P} \neg t_p\}} \left(\bigwedge_{p \in P} s_p\right)$.\\
            This is equivalent to the wanted result.
            \end{proof}

            \begin{lemma}[$\worra {} {}$ alternative definition]
                The following algorithmic definition for $\worra {} {}$ is equivalent to the previous one:

                \[\forall t, s.\ \worra t s \simeq \bigvee_{i \in I} \left( \bigvee_{\{P \subseteq P_i\alt s \not\leq \bigvee_{p \in P} \neg t_p \} }\left(
                    \dom t \land \bigwedge_{p\in P_i} s_p \land \bigwedge_{n \in P_i \setminus P} \neg s_n \right)\right)\]
            \end{lemma}

            \begin{proof}
            \begingroup
            \allowdisplaybreaks
            \begin{align*}
                \worra t s & = \dom t \wedge\bigvee_{i\in I}\left(
                    \bigwedge_{\{P \subseteq P_i\alt s \leq \bigvee_{p \in P} \neg t_p\}} \left(\bigvee_{p \in P} \neg s_p \right)\right)\\
                & \simeq \bigvee_{i\in I} \left(\dom t \wedge
                    \bigwedge_{\{P \subseteq P_i\alt s \leq \bigvee_{p \in P} \neg t_p\}} \left(\bigvee_{p \in P} \neg s_p \right)\right)\\
                & \simeq \bigvee_{i\in I} \left(\left( \dom t \wedge \bigvee_{p \in P_i} s_p \right) \wedge
                    \bigwedge_{\{P \subseteq P_i\alt s \leq \bigvee_{p \in P} \neg t_p\}} \left(\bigvee_{p \in P} \neg s_p \right)\right)\\
                & \simeq \bigvee_{i\in I} \left(\left( \dom t \wedge \bigvee_{p \in P_i} \left(s_p \land \bigvee_{P \subseteq P_i \setminus \{p\}}
                    \left(\bigwedge_{p \in P} s_p \land \bigwedge_{n \in (P_i \setminus \{p\}) \setminus P} \neg s_n \right) \right) \right) \wedge
                    \bigwedge_{\{P \subseteq P_i\alt s \leq \bigvee_{p \in P} \neg t_p\}} \left(\bigvee_{p \in P} \neg s_p \right)\right)\\
                & \simeq \bigvee_{i\in I} \left(\left( \dom t \wedge \bigvee_{p \in P_i} \left(\bigvee_{P \subseteq P_i \setminus \{p\}} \left( s_p \land
                    \bigwedge_{p \in P} s_p \land \bigwedge_{n \in (P_i \setminus \{p\}) \setminus P} \neg s_n \right) \right) \right) \wedge
                    \bigwedge_{\{P \subseteq P_i\alt s \leq \bigvee_{p \in P} \neg t_p\}} \left(\bigvee_{p \in P} \neg s_p \right)\right)\\
                & \simeq \bigvee_{i\in I} \left(\left( \dom t \wedge \bigvee_{\substack{P \subseteq P_i\\P \neq \varnothing}}
                    \left(\bigwedge_{p \in P} s_p \land \bigwedge_{n \in P_i \setminus P} \neg s_n \right) \right) \wedge
                    \bigwedge_{\{P \subseteq P_i\alt s \leq \bigvee_{p \in P} \neg t_p\}} \left(\bigvee_{p \in P} \neg s_p \right)\right)\\
                & \simeq \bigvee_{i\in I} \left( \dom t \wedge \bigvee_{\substack{P \subseteq P_i\\P \neq \varnothing}}
                    \left(\bigwedge_{p \in P} s_p \land \bigwedge_{n \in P_i \setminus P} \neg s_n \right) \setminus
                    \bigvee_{\{P \subseteq P_i\alt s \leq \bigvee_{p \in P} \neg t_p\}} \left(\bigwedge_{p \in P} s_p \right)\right)\\
                & \simeq \bigvee_{i\in I} \left( \dom t \wedge \bigvee_{\substack{P \subseteq P_i\\P \neq \varnothing}}
                    \left(\bigwedge_{p \in P} s_p \land \bigwedge_{n \in P_i \setminus P} \neg s_n \right) \setminus
                    \bigvee_{\{P \subseteq P_i\alt s \leq \bigvee_{p \in P} \neg t_p\}}
                    \left(\bigwedge_{p \in P} s_p \land \bigwedge_{n \in P_i \setminus P} \neg s_n \right)\right)\\
                & \simeq \bigvee_{i \in I} \left( \dom t \land \bigvee_{\{P \subseteq P_i\alt s \not\leq \bigvee_{p \in P} \neg t_p \} }\left(
                    \bigwedge_{p\in P_i} s_p \land \bigwedge_{n \in P_i \setminus P} \neg s_n \right)\right)\\
                & \simeq \bigvee_{i \in I} \left( \bigvee_{\{P \subseteq P_i\alt s \not\leq \bigvee_{p \in P} \neg t_p \} }\left(
                    \dom t \land \bigwedge_{p\in P_i} s_p \land \bigwedge_{n \in P_i \setminus P} \neg s_n \right)\right)
            \end{align*}
            \endgroup
            \end{proof}

            \begin{lemma}[Optimality of $\worra {} {}$]
                Let $t$, $s$, two types.
                For any $u$ such that $\apply t {(\dom t \setminus u)} \leq \neg s$, we have $\worra t s \leq u$.
            \end{lemma}

            \begin{proof}
            Let $t$ an arrow type. $t \simeq \bigvee_{i\in I}\left(\bigwedge_{p\in P_i}(s_p\to t_p)\bigwedge_{n\in N_i}\neg(s_n'\to t_n')\right)$\\
            Let $s$ be any type.

            Let $u$ be such that $\apply t {(\dom t \setminus u)} \leq \neg s$. We want to prove that $\worra t s \leq u$.

            We have:
            \[\worra t s = \bigvee_{i \in I} \left(\bigvee_{\{P \subseteq P_i\alt s \not\leq \bigvee_{p \in P} \neg t_p \} } a_{i,P}\right) \]
            With:
            \[a_{i,P}=\dom t \land \bigwedge_{p\in P_i} s_p \land \bigwedge_{n \in P_i \setminus P} \neg s_n\]

            Let $i \in I$ and $P \subseteq P_i$ such that $s \not\leq \bigvee_{p \in P} \neg t_p$ (equivalently, $s \land \bigwedge_{p \in P} t_p \not\simeq \Empty$) and such that $a_{i,P} \not\simeq \Empty$.\\
            For convenience, let $a = a_{i,P}$. We just have to show that $a \leq u$.

            By the absurd, let's suppose that $a \setminus u \not\simeq \Empty$ and show that $(\apply t {(\dom t \setminus u)}) \land s \not\simeq \Empty$.\\

            Let's recall the algorithmic definition of $\circ$:
            \[\apply t {(\dom t \setminus u)} = \bigvee_{i\in I}\left(\bigvee_{\{Q\subsetneq P_i\alt \dom t \setminus u \not\leq\bigvee_{q\in Q}s_q\}}\left(\bigwedge_{p\in P_i\setminus Q}t_p\right)\right)\]

            Let's take $Q = P_i \setminus P$. We just have to prove that:
            \[ \dom t \setminus u \not\leq\bigvee_{q\in Q}s_q \text{\ \ \ and\ \ \ } s \land \bigwedge_{p\in P_i\setminus Q}t_p \not\simeq \Empty \]

            As $P_i \setminus Q = P$, we immediatly have $s \land \bigwedge_{p\in P_i\setminus Q}t_p \not\simeq \Empty$.

            Moreover, we know that $a \leq \bigwedge_{q \in Q} \neg s_q$ (definition of $a_{i,P}$), so we have:
            \[a \land \bigwedge_{q\in Q} \neg s_q \simeq a\]
            Thus: \[(a \setminus u) \land \bigwedge_{q\in Q} \neg s_q \simeq (a \land \bigwedge_{q\in Q} \neg s_q) \setminus u \simeq a \setminus u \not\simeq \Empty\]
            And so: \[ a \setminus u \not\leq \bigvee_{q\in Q}s_q \]

            As $ \dom t \setminus u \geq a \setminus u$, we can immediatly obtain the remaining inequality.
            \end{proof}

            \begin{theorem}[Characterization of $\worra {} {}$]
                $\forall t, s.\ \worra t s = \min \{ u \alt \apply t {(\dom t \setminus u)} \leq \neg s \}$.
            \end{theorem}

            \begin{proof}
            Immediate consequence of the previous results.
            \end{proof}

\subsection{Type Schemes}

We introduce for the proofs the notion of \emph{type schemes}
and we define a more powerful algorithmic type system that uses them.
It allows us to have a stronger (but still partial) completeness theorem.

The proofs for the algorithmic type system presented in \ref{sec:algorules} can be derived
from the proofs of this section (see Section \ref{sec:proofs_algorithmic_without_ts}).

\subsubsection{Type schemes}\label{app:typeschemes}

We introduce the new syntactic category of \emph{type schemes} which are the terms~$\ts$ inductively produced by the following grammar.
\[
\begin{array}{lrcl}
  \textbf{Type schemes} & \ts & ::= & t \alt \tsfun {\arrow t t ; \cdots ; \arrow t t} \alt \ts \tstimes \ts \alt \ts \tsor \ts \alt \tsempty
\end{array}
\]
Type schemes denote sets of types, as formally stated by the following definition:
\begin{definition}[Interpretation of type schemes]
  We define the function $\tsint {\_}$ that maps type schemes into sets of types.\svvspace{-2.5mm}
  \begin{align*}
    \begin{array}{lcl}
    \tsint t &=& \{s\alt t \leq s\}\\
    \tsint {\tsfunone {t_i} {s_i}_{i=1..n}} &=& \{s\alt
    \exists s_0 = \bigwedge_{i=1..n} \arrow {t_i} {s_i}
    \land \bigwedge_{j=1..m} \neg (\arrow {t_j'} {s_j'}).\
    \Empty \not\simeq s_0 \leq s \}\\
    \tsint {\ts_1 \tstimes \ts_2} &=& \{s\alt \exists t_1 \in \tsint {\ts_1}\
    \exists t_2 \in \tsint {\ts_2}.\ \pair {t_1} {t_2} \leq s\}\\
    \tsint {\ts_1 \tsor \ts_2} &=& \{s\alt \exists t_1 \in \tsint {\ts_1}\
    \exists t_2 \in \tsint {\ts_2}.\ {t_1} \vee {t_2} \leq s\}\\
    \tsint \tsempty &=& \varnothing
    \end{array}
  \end{align*}
\end{definition}
Note that $\tsint \ts$ is closed under subsumption and intersection
 and that $\tsempty$, which denotes the
empty set of types is different from $\Empty$ whose interpretation is
the set of all types.

\begin{lemma}[\cite{Frisch2008}]
  Let $\ts$ be a type scheme and $t$ a type. It is possible to decide the assertion $t \in \tsint \ts$,
  which we also write $\ts \leq t$.
\end{lemma}

We can now formally define the relation $v\in t$ used in
Section~\ref{sec:opsem} to define the dynamic semantics of
the language. First, we associate each (possibly, not well-typed)
value to a type scheme representing the best type information about
the value. By induction on the definition of values: $\tyof c {} =
\basic c$, $\tyof {\lambda^{\wedge_{i\in I}s_i\to t_i} x.e}{}=
\tsfun{s_i\to t_i}_{i\in I}$, $\tyof {(v_1,v_2)}{} = \tyof
      {v_1}{}\tstimes\tyof {v_1}{}$. Then we have $v\in t\iffdef \tyof
      v{}\leq t$.

We also need to perform intersections of type schemes so as to intersect the static type of an expression (i.e., the one deduced by conventional rules) with the one deduced by occurrence typing (i.e., the one derived by $\vdashp$). For our algorithmic system (see \Rule{Env$_{\scriptscriptstyle\mathcal{A}}$} in Section~\ref{sec:algorules}) all we need to define is the intersection of a type scheme with a type:
\begin{lemma}[\cite{Frisch2008}]
  Let $\ts$ be a type scheme and $t$ a type. We can compute a type scheme, written $t \tsand \ts$, such that
  \(\tsint {t \tsand \ts} = \{s \alt \exists t' \in \tsint \ts.\ t \land t' \leq s \}\)
\end{lemma}
Finally, given a type scheme $\ts$ it is straightforward to choose in its interpretation a type $\tsrep\ts$ which serves as the canonical representative of the set (i.e., $\tsrep \ts \in \tsint \ts$):
\begin{definition}[Representative]
  We define a function $\tsrep {\_}$ that maps every non-empty type scheme into a type, \textit{representative} of the set of types denoted by the scheme.\svvspace{-2mm}
  \begin{align*}
    \begin{array}{lcllcl}
    \tsrep t &=& t  &  \tsrep {\ts_1 \tstimes \ts_2} &=& \pair {\tsrep {\ts_1}} {\tsrep {\ts_2}}\\
    \tsrep {\tsfunone {t_i} {s_i}_{i\in I}} &=& \bigwedge_{i\in I} \arrow {t_i} {s_i} \qquad&    \tsrep {\ts_1 \tsor \ts_2} &=& \tsrep {\ts_1} \vee \tsrep {\ts_2}\\
    \tsrep \tsempty && \textit{undefined}
    \end{array}\svvspace{-4mm}
  \end{align*}
\end{definition}

Type schemes are already present in the theory of semantic subtyping presented in~\cite[Section
6.11]{Frisch2008}. In particular, it explains how the operators such as $\circ$, $\bpl t$ and $\bpr t$
can be extended to type schemes (see also~\citep[\S4.4]{Cas15} for a detailed description).

\subsection{Algorithmic type system with type schemes}

We present here a refinement of the algorithmic type system presented in \ref{sec:algorules}
that associates to an expression a type scheme instead of a regular type.
This allows to type expressions more precisely and thus to have a more powerful
(but still partial) completeness theorem in regards to the declarative type system.

The results about this new type system will be used in \ref{sec:proofs_algorithmic_without_ts} in order to obtain a soundness and completeness
theorem for the algorithmic type system presented in \ref{sec:algorules}.

            \begin{mathpar}
              \Infer[Efq\Aats]
          { }
          { \Gamma, (e:\Empty) \vdashAts e': \Empty }
          { \begin{array}{c}\text{\tiny with priority over}\\[-1.8mm]\text{\tiny all the other rules}\end{array}}
          \qquad
          \Infer[Var\Aats]
              { }
              { \Gamma \vdashAts x: \Gamma(x) }
              { x\in\dom\Gamma}
          \\
          \Infer[Env\Aats]
              { \Gamma\setminus\{e\} \vdashAts e : \ts }
              { \Gamma \vdashAts e: \Gamma(e) \tsand \ts }
              { e\in\dom\Gamma \text{ and } e \text{ not a variable}}
          \qquad
          \Infer[Const\Aats]
              { }
              {\Gamma\vdashAts c:\basic{c}}
              {c\not\in\dom\Gamma}
           \\
              \Infer[Abs\Aats]
                  {\Gamma,x:s_i\vdashAts e:\ts_i'\\ \ts_i'\leq t_i}
                  {
                  \Gamma\vdashAts\lambda^{\wedge_{i\in I}\arrow {s_i} {t_i}}x.e:\textstyle\tsfun {\arrow {s_i} {t_i}}_{i\in I}
                  }
                  {\lambda^{\wedge_{i\in I}\arrow {s_i} {t_i}}x.e\not\in\dom\Gamma}
                  \\
              \Infer[App\Aats]
                  {
                    \Gamma \vdashAts e_1: \ts_1\\
                    \Gamma \vdashAts e_2: \ts_2\\
                    \ts_1 \leq \arrow \Empty \Any\\
                    \ts_2 \leq \dom {\ts_1}
                  }
                  { \Gamma \vdashAts {e_1}{e_2}: \ts_1 \circ \ts_2 }
                  { {e_1}{e_2}\not\in\dom\Gamma}
                  \\
              \Infer[Case\Aats]
                    {\Gamma\vdashAts e:\ts_0\\
                    \Refine {e,t} \Gamma \vdashAts e_1 : \ts_1\\
                    \Refine {e,\neg t} \Gamma \vdashAts e_2 : \ts_2}
                    {\Gamma\vdashAts \tcase {e} t {e_1}{e_2}: \ts_1\tsor \ts_2}
                    { \tcase {e} {t\!} {\!e_1\!}{\!e_2}\not\in\dom\Gamma}
              \\
              \Infer[Proj\Aats]
              {\Gamma \vdashAts e:\ts\and \ts\leq\pair\Any\Any}
              {\Gamma \vdashAts \pi_i e:\bpi_{\mathbf{i}}(\ts)}
              {\pi_i e\not\in\dom\Gamma}

              \Infer[Pair\Aats]
              {\Gamma \vdashAts e_1:\ts_1 \and \Gamma \vdashAts e_2:\ts_2}
              {\Gamma \vdashAts (e_1,e_2):{\ts_1}\tstimes{\ts_2}}
              {(e_1,e_2)\not\in\dom\Gamma}

            \end{mathpar}

            \begin{align*}
              \tyof e \Gamma =
              \left\{\begin{array}{ll}
                \ts & \text{if } \Gamma \vdashAts e:\ts \\
                \tsempty & \text{otherwise}
              \end{array}\right.
            \end{align*}

            \begin{eqnarray}
              \constr\epsilon{\Gamma,e,t} & = & t\\
              \constr{\varpi.0}{\Gamma,e,t} & = & \neg(\arrow{\env {\Gamma,e,t}{(\varpi.1)}}{\neg \env {\Gamma,e,t} (\varpi)})\\
              \constr{\varpi.1}{\Gamma,e,t} & = & \worra{\tsrep {\tyof{\occ e{\varpi.0}}\Gamma}}{\env {\Gamma,e,t} (\varpi)}\\
              \constr{\varpi.l}{\Gamma,e,t} & = & \bpl{\env {\Gamma,e,t} (\varpi)}\\
              \constr{\varpi.r}{\Gamma,e,t} & = & \bpr{\env {\Gamma,e,t} (\varpi)}\\
              \constr{\varpi.f}{\Gamma,e,t} & = & \pair{\env {\Gamma,e,t} (\varpi)}\Any\\
              \constr{\varpi.s}{\Gamma,e,t} & = & \pair\Any{\env {\Gamma,e,t} (\varpi)}\\
              \env {\Gamma,e,t} (\varpi) & = & \tsrep {\constr \varpi {\Gamma,e,t} \tsand \tyof {\occ e \varpi} \Gamma}
            \end{eqnarray}

            \begin{align*}
              &\RefineStep {e,t} (\Gamma) = \Gamma' \text{ with:}\\
              &\dom{\Gamma'}=\dom{\Gamma} \cup \{e' \alt \exists \varpi.\ \occ e \varpi \equiv e'\}\\
              &\Gamma'(e') =
                \left\{\begin{array}{ll}
                  \bigwedge_{\{\varpi \alt \occ e \varpi \equiv e'\}}
                  \env {\Gamma,e,t} (\varpi) & \text{if } \exists \varpi.\ \occ e \varpi \equiv e' \\
                  \Gamma(e') & \text{otherwise}
                \end{array}\right.\\&\\
              &\Refine {e,t} \Gamma={\RefineStep {e,t}}^{n_o}(\Gamma)\qquad\text{with $n$ a global parameter}
            \end{align*}

            \subsection{Proofs for the algorithmic type system with type schemes}

            This section is about the algorithmic type system with type schemes (soundness and some completeness properties).

            Note that, now that we have type schemes, use a different but more convenient definition for $\tyof e \Gamma$ that the one
            in Section~\ref{sec:typenv}:
            \begin{align*}
              \tyof e \Gamma =
              \left\{\begin{array}{ll}
                \ts & \text{if } \Gamma \vdashAts e:\ts \\
                \tsempty & \text{otherwise}
              \end{array}\right.
            \end{align*}

            In this way, $\tyof e \Gamma$ is always defined but is equal to $\tsempty$ when $e$ is not
            well-typed in $\Gamma$.

            We will reuse the definitions and notations introduced in the previous proofs.
            In particular, we only consider well-formed environments, as in the proofs of the declarative type system.

            \subsubsection{Soundness}

            \begin{theorem}[Soundness of the algorithm]\label{soundnessAts}
              For every $\Gamma$, $e$, $t$, $n_o$, if $\tyof e \Gamma \leq t$, then we can derive $\Gamma \vdash e:t$.

              More precisely:
              \begin{align*}
                &\forall \Gamma, e, t.\ \tyof e \Gamma \leq t \Rightarrow \Gamma \vdash e:t\\
                &\forall \Gamma, e, t, \varpi.\ \tyof e \Gamma \neq \tsempty \Rightarrow \pvdash \Gamma e t \varpi:\env {\Gamma,e,t} (\varpi)\\
                &\forall \Gamma, e, t.\ \tyof e \Gamma \neq \tsempty \Rightarrow \Gamma \evdash e t \Refine {e,t} \Gamma
              \end{align*}
            \end{theorem}

            \begin{proof}
            We proceed by induction over the structure of $e$
            and, for two identical $e$, on the domain of $\Gamma$ (with the inclusion order).

            Let's prove the first property.
            Let $t$ such that $\tsint{\tyof e \Gamma} \leq t$.

            If $\Gamma = \bot$, we trivially have $\Gamma \vdash e:t$ with the rule \Rule{Efq}.
            Let's assume $\Gamma \neq \bot$.

            If $e=x$ is a variable, then the last rule used is \Rule{Var\Aats}.
            We can derive $\Gamma \vdash x:t$ by using the rule \Rule{Env} and \Rule{Subs}.
            So let's assume that $e$ is not a variable.

            If $e\in\dom\Gamma$, then the last rule used is \Rule{Env\Aats}.
            Let $t'\in\tsint{\ts}$ such that $t'\land\Gamma(e)\leq t$.
            The induction hypothesis gives $\Gamma\setminus\{e\} \vdash e:t'$
            (the premise uses the same $e$ but the domain of $\Gamma$ is strictly smaller).
            Thus, we can build a derivation $\Gamma \vdash e:t$ by using the rules \Rule{Subs}, \Rule{Inter},
            \Rule{Env} and the derivation $\Gamma\setminus\{e\} \vdash e:t'$.

            Now, let's suppose that $e\not\in\dom\Gamma$.

            \begin{description}
              \item[$e=c$] The last rule is \Rule{Const\Aats}. We derive easily $\Gamma \vdash c:t$ with \Rule{Const} and \Rule{Subs}.
              \item[$e=x$] Already treated.
              \item[$e=\lambda^{\bigwedge_{i\in I} \arrow {t_i}{s_i}}x.e'$]
              The last rule is \Rule{Abs\Aats}.
              We have $\bigwedge_{i\in I} \arrow {t_i}{s_i} \leq t$.
              Using the definition of type schemes, let $t'=\bigwedge_{i\in I} \arrow {t_i}{s_i} \land \bigwedge_{j\in J} \neg \arrow {t'_j}{s'_j}$ such that $\Empty \neq t' \leq t$.
              The induction hypothesis gives, for all $i\in I$, $\Gamma,x:s_i\vdash e':t_i$.

              Thus, we can derive $\Gamma\vdash e:\bigwedge_{i\in I} \arrow {t_i}{s_i}$ using the rule \Rule{Abs+}, and with \Rule{Inter} and
              \Rule{Abs-} we can derive $\Gamma\vdash e:t'$. We can conclude by applying \Rule{Subs}.
              \item[$e=e_1 e_2$] The last rule is \Rule{App\Aats}.
              We have $\apply {\ts_1} {\ts_2} \leq t$. Thus, let $t_1$ and $t_2$ such that $\ts_1 \leq t_1$, $\ts_2 \leq t_2$ and $\apply {t_1} {t_2} \leq t$.
              We know, according to the descriptive definition of $\apply {} {}$, that there exists $s\leq t$ such that $t_1 \leq \arrow {t_2} s$.

              By using the induction hypothesis, we have $\Gamma\vdash e_1:t_1$ and $\Gamma\vdash e_2:t_2$. We can thus derive
              $\Gamma\vdash e_1:\arrow {t_2} s$ using \Rule{Subs}, and together with $\Gamma\vdash e_2:t_2$ it gives
              $\Gamma\vdash e_1\ e_2:s$ with \Rule{App}. We conclude with \Rule{Subs}.

              \item[$e=\pi_i e'$] The last rule is \Rule{Proj\Aats}. We have $\bpi_i \ts \leq t$. Thus, let $t'$ such that $\ts \leq t'$ and $\bpi_i t' \leq t$.
              We know, according to the descriptive definition of $\bpi_i$, that there exists $t_i\leq t$ such that $t' \leq \pair \Any {t_i}$ (for $i=2$) or $t' \leq \pair {t_i} \Any$ (for $i=1$).

              By using the induction hypothesis, we have $\Gamma\vdash e':t'$, and thus we easily conclude using \Rule{Subs} and \Rule{Proj}
              (for instance for the case $i=1$, we can derive $\Gamma\vdash e':\pair {t_i} \Any$ with \Rule{Subs} and then use \Rule{Proj}).

              \item[$e=(e_1,e_2)$] The last rule is \Rule{Pair\Aats}. We conclude easily with the induction hypothesis and the rules \Rule{Subs} and \Rule{Pair}.

              \item[$e=\ite {e_0} t {e_1} {e_2}$] The last rule is \Rule{Case\Aats}. We conclude easily with the induction hypothesis and the rules
              \Rule{Subs} and \Rule{Case} (for the application of \Rule{Case}, $t'$ must be taken equal to $t_1 \vee t_2$ with $t_1$ and $t_2$ such that $\ts_1\leq t_1$, $\ts_2\leq t_2$ and $t_1 \vee t_2 \leq t$).
            \end{description}\ \\

            Now, let's prove the second property.
            We perform a (nested) induction on $\varpi$.

            Recall that $\env {\Gamma,e,t} (\varpi) = \tsrep {\constr \varpi {\Gamma,e,t} \tsand \tyof {\occ e \varpi} \Gamma}$.

            For any $t'$ such that $\tyof {\occ e \varpi} \Gamma \leq t'$, we can easily derive $\pvdash \Gamma e t \varpi : t'$ by using the outer induction hypothesis
            (the first property that we have proved above) and the rule \Rule{PTypeof}.

            Now we have to derive $\pvdash \Gamma e t \varpi : \constr \varpi {\Gamma,e,t}$ (then it will be easy to conclude using the rule \Rule{PInter}).
            \begin{description}
              \item[$\varpi=\epsilon$] We use the rule \Rule{PEps}.
              \item[$\varpi=\varpi'.1$]
              Let's note $f=\tsrep {\tyof {\occ e {\varpi'.0}} \Gamma}$, $s=\env {\Gamma,e,t} (\varpi')$ and $t_{\text{res}} = \worra f s$.

              By using the outer and inner induction hypotheses, we can derive $\pvdash \Gamma e t \varpi'.0 : f$ and $\pvdash \Gamma e t \varpi' : s$.

              By using the descriptive definition of $\worra {} {}$, we have $t' = \apply f {(\dom f \setminus t_{\text{res}})} \leq \neg s$.

              Moreover, by using the descriptive definition of $\apply {} {}$ on $t'$, we have
              $f \leq \arrow {(\dom f \setminus t_{\text{res}})} {t'}$.

              As $t'\leq \neg s$, it gives $f \leq \arrow {(\dom f \setminus t_{\text{res}})} {\neg s}$.

              Let's note $t_1 = \dom f \setminus t_{\text{res}}$ and $t_2 = \neg s$. The above inequality can be rewritten $f \leq \arrow {t_1} {t_2}$.

              Thus, by using \Rule{PSubs} on the derivation $\pvdash \Gamma e t \varpi'.0 : f$, we can derive $\pvdash \Gamma e t \varpi'.0 : \arrow {t_1} {t_2}$.
              We have:
              \begin{itemize}
                \item $t_2 \land s \simeq \Empty$ (as $t_2 = \neg s$)
                \item $\neg t_1 = t_{\text{res}} \vee \neg \dom f = t_{\text{res}}$
              \end{itemize}

              In consequence, we can conclude by applying the rule \Rule{PAppR}
              with the premises $\pvdash \Gamma e t \varpi'.0 : \arrow {t_1} {t_2}$ and $\pvdash \Gamma e t \varpi' : s$.

              \item[$\varpi=\varpi'.0$] By using the inner induction hypothesis and the previous case we've just proved, we can derive
              $\pvdash \Gamma e t \varpi' : \env {\Gamma,e,t} (\varpi')$ and $\pvdash \Gamma e t \varpi'.1 : \env {\Gamma,e,t} (\varpi'.1)$.
              Hence we can apply \Rule{PAppL}.
              \item[$\varpi=\varpi'.l$] Let's note $t_1=\bpi_1 \env {\Gamma,e,t} (\varpi')$.
              According to the descriptive definition of $\bpi_1$, we have $\env {\Gamma,e,t} (\varpi') \leq \pair {t_1} {\Any}$.

              The inner induction hypothesis gives $\pvdash \Gamma e t \varpi' : \env {\Gamma,e,t} (\varpi')$, and thus using the rule \Rule{PSubs}
              we can derive $\pvdash \Gamma e t \varpi' : \pair {t_1} \Any$. We can conclude just by applying the rule \Rule{PPairL} to this premise.
              \item[$\varpi=\varpi'.r$] This case is similar to the previous.
              \item[$\varpi=\varpi'.f$] The inner induction hypothesis gives $\pvdash \Gamma e t \varpi' : \env {\Gamma,e,t} (\varpi')$,
              so we can conclude by applying \Rule{PFst}.
              \item[$\varpi=\varpi'.s$] The inner induction hypothesis gives $\pvdash \Gamma e t \varpi' : \env {\Gamma,e,t} (\varpi')$,
              so we can conclude by applying \Rule{PSnd}.
            \end{description}\ \\

            Finally, let's prove the third property.
            Let $\Gamma' = \Refine {e,t} \Gamma = {\RefineStep{e,t}}^{n_0}(\Gamma)$.
            We want to show that $\Gamma \evdash e t \Gamma'$ is derivable.

            First, let's note that $\evdash e t$ is transitive:
            if $\Gamma \evdash e t \Gamma'$ and $\Gamma' \evdash e t \Gamma''$, then $\Gamma \evdash e t \Gamma''$.
            The proof is quite easy: we can just start from the derivation of $\Gamma \evdash e t \Gamma'$, and we add
            at the end a slightly modified version of the derivation of $\Gamma' \evdash e t \Gamma''$ where:
            \begin{itemize}
              \item the initial \Rule{Base} rule has been removed in order to be able to do the junction,
              \item all the $\Gamma'$ at the left of $\evdash e t$ are replaced by $\Gamma$
              (the proof is still valid as this $\Gamma'$ at the left is never used in any rule)
            \end{itemize}

            Thanks to this property, we can suppose that $n_0 = 1$ (and so $\Gamma' = \RefineStep{e,t}(\Gamma)$).
            If it is not the case, we just have to proceed by induction on $n_0$ and use the transitivity property.

            Let's build a derivation for $\Gamma \evdash e t \Gamma'$.

            By using the proof of the second property on $e$ that we've done just before, we get:
            $\forall \varpi.\ \pvdash \Gamma e t \varpi: \env {\Gamma,e,t} (\varpi)$.

            Let's recall a monotonicity property: for any $\Gamma_1$ and $\Gamma_2$ such that $\Gamma_2 \leq \Gamma_1$, we have
            $\forall t'.\ \pvdash {\Gamma_1} e t \varpi:t' \Rightarrow \pvdash {\Gamma_2} e t \varpi:t'$.\\
            Moreover, when we also have $\occ e \varpi \in\dom{\Gamma_2}$, we can derive $\pvdash {\Gamma_2} e t \varpi:t'\land\Gamma_2(\occ e \varpi)$
            (just by adding a \Rule{PInter} rule with a \Rule{PTypeof} and a \Rule{Env}).

            Hence, we can apply successively a \Rule{Path} rule for all valid $\varpi$ in $e$,
            with the following premises ($\Gamma_\varpi$ being the previous environment, that trivially verifies $\Gamma_\varpi\leq\Gamma$):\\

            \begin{tabular}{lll}
              If $\occ e \varpi \in\dom{\Gamma_\varpi}$&$\pvdash {\Gamma_\varpi} e t \varpi: \env {\Gamma,e,t} (\varpi)\land\Gamma_\varpi(\occ e \varpi)$&$\Gamma\evdash e t {\Gamma_\varpi}$\\
              Otherwise&$\pvdash {\Gamma_\varpi} e t \varpi: \env {\Gamma,e,t} (\varpi)$&$\Gamma\evdash e t {\Gamma_\varpi}$
            \end{tabular}

            At the end, it gives the judgement $\Gamma \evdash e t \Gamma'$, so it concludes the proof.
            \end{proof}

            \subsubsection{Completeness}

            \begin{definition}[Bottom environment]
              Let $\Gamma$ an environment.\\
              $\Gamma$ is bottom (noted $\Gamma = \bot$) iff $\exists e\in\dom\Gamma.\ \Gamma(e)\simeq\Empty$.
            \end{definition}

            \begin{definition}[Algorithmic (pre)order on environments]
            Let $\Gamma$ and $\Gamma'$ two environments. We write $\Gamma' \leqA \Gamma$ iff:
            \begin{align*}
                &\Gamma'=\bot \text{ or } (\Gamma\neq\bot \text{ and } \forall e \in \dom \Gamma.\ \tyof e \Gamma \leq \Gamma(e))
            \end{align*}

            For an expression $e$, we write $\Gamma' \leqA^e \Gamma$ iff:
            \begin{align*}
              &\Gamma'=\bot \text{ or } (\Gamma\neq\bot \text{ and } \forall e' \in \dom \Gamma \text{ such that $e'$ is a subexpression of $e$}.\ \tyof {e'} \Gamma \leq \Gamma({e'}))
            \end{align*}

            Note that if $\Gamma' \leqA \Gamma$, then $\Gamma' \leqA^e \Gamma$ for any $e$.
            \end{definition}

            \begin{definition}[Order relation for type schemes]
              Let $\ts_1$ and $\ts_2$ two type schemes. We write $\ts_2 \leq \ts_1$ iff $\tsint {\ts_1} \subseteq \tsint{\ts_2}$.
            \end{definition}

          \begin{lemma} When well-defined, the following inequalities hold:
            \begin{align*}
              &\forall t,\ts.\ \tsrep{t\tsand\ts} \leq t \land \tsrep{\ts}\\
              &\forall t_1,t_2,\ts_1,\ts_2.\ t_1 \leq t_2 \text{ and } \ts_1\leq\ts_2 \text{ and } \tsrep{\ts_1}\leq\tsrep{\ts_2} \Rightarrow \tsrep{t_1\tsand\ts_1} \leq \tsrep{t_2\tsand\ts_2}\\
              &\forall \ts_1,\ts_2.\ \tsrep{\apply {\ts_1}{\ts_2}} \leq \apply {\tsrep {\ts_1}}{\tsrep {\ts_2}}
            \end{align*}
          \end{lemma}

          \begin{proof}
            Straightfoward, by induction on the structure of $\ts$.
          \end{proof}

          \begin{lemma}[Monotonicity of the algorithm] Let $\Gamma$, $\Gamma'$ and $e$ such that $\Gamma'\leqA^e \Gamma$ and $\tyof e \Gamma \neq \tsempty$. We have:
            \begin{align*}
              &\tyof e {\Gamma'} \leq \tyof e {\Gamma} \text{ and } \tsrep{\tyof e {\Gamma'}} \leq \tsrep{\tyof e {\Gamma}}\\
              &\forall t,\varpi.\ \env {\Gamma',e,t} (\varpi) \leq \env {\Gamma,e,t} (\varpi)\\
              &\forall t.\ \Refine {e,t} {\Gamma'} \leqA^e \Refine {e,t} \Gamma\\
            \end{align*}
          \end{lemma}

          \begin{proof}
          We proceed by induction over the structure of $e$
          and, for two identical $e$, on the domains of $\Gamma$ and $\Gamma'$ (with the lexicographical inclusion order).

          Let's prove the first property: $\tyof e {\Gamma'} \leq \tyof e {\Gamma} \text{ and } \tsrep{\tyof e {\Gamma'}} \leq \tsrep{\tyof e {\Gamma}}$.
          We will focus on showing $\tyof e {\Gamma'} \leq \tyof e {\Gamma}$.

          The property $\tsrep{\tyof e {\Gamma'}} \leq \tsrep{\tyof e {\Gamma}}$
          can be proved in a very similar way, by using the fact that operators on type schemes like $\tsand$ or $\apply {} {}$ are also monotone.
          (Note that the only rule that introduces the type scheme constructor $\tsfun {\_}$ is \Rule{Abs\Aats}.)

          If $\Gamma' = \bot$ we can conclude directly with the rule \Rule{Efq\Aats}.
          So let's assume $\Gamma' \neq \bot$ and $\Gamma \neq \bot$
          (as $\Gamma = \bot \Rightarrow \Gamma' = \bot$ by definition of $\leqA^e$).

          If $e=x$ is a variable, then the last rule used in $\tyof e \Gamma$ and $\tyof e {\Gamma'}$ is \Rule{Var\Aats}.
          As $\Gamma' \leqA^e \Gamma$, we have $\Gamma'(e) \leq \Gamma(e)$ and thus
          we can conclude with the rule \Rule{Var\Aats}.
          So let's assume that $e$ is not a variable.

          If $e\in\dom\Gamma$, then the last rule used in $\tyof e \Gamma$ is \Rule{Env\Aats}.
          As $\Gamma' \leqA^e \Gamma$, we have $\tyof e {\Gamma'} \leq \Gamma(e)$.
          Moreover, by applying the induction hypothesis, we get $\tyof e {\Gamma'\setminus \{e\}} \leq \tyof e {\Gamma\setminus \{e\}}$
          (we can easily verify that $\Gamma'\setminus\{e\} \leqA^e \Gamma\setminus\{e\}$).
          \begin{itemize}
            \item If we have $e\in\dom{\Gamma'}$, we have according to the rule \Rule{Env\Aats}
            $\tyof e {\Gamma'} \leq \tyof e {\Gamma'\setminus \{e\}} \leq \tyof e {\Gamma\setminus \{e\}}$.

            Together with $\tyof e {\Gamma'} \leq \Gamma(e)$,
            we deduce $\tyof e {\Gamma'} \leq \Gamma(e) \tsand \tyof e {\Gamma\setminus \{e\}} = \tyof e {\Gamma}$.
            \item  Otherwise, we have $e\not\in\dom{\Gamma'}$. Thus
            $\tyof e {\Gamma'} = \tyof e {\Gamma'\setminus \{e\}} \leq \Gamma(e) \tsand \tyof e {\Gamma\setminus \{e\}}=\tyof e {\Gamma}$.
          \end{itemize}

          If $e\not\in\dom\Gamma$ and $e\in\dom{\Gamma'}$, the last rule is \Rule{Env\Aats} for $\tyof e {\Gamma'}$.
          As $\Gamma'\setminus\{e\} \leqA^e \Gamma\setminus\{e\} = \Gamma$,
          we have $\tyof e {\Gamma'} \leq \tyof e {\Gamma'\setminus \{e\}} \leq \tyof e {\Gamma}$ by induction hypothesis.

          Thus, let's suppose that $e\not\in\dom\Gamma$ and $e\not\in\dom{\Gamma'}$.
          From now we know that the last rule in the derivation of $\tyof e {\Gamma}$ and $\tyof e {\Gamma'}$ (if any) is the same.

          \begin{description}
            \item[$e=c$] The last rule is \Rule{Const\Aats}. It does not depend on $\Gamma$ so this case is trivial.
            \item[$e=x$] Already treated.
            \item[$e=\lambda^{\bigwedge_{i\in I} \arrow {t_i}{s_i}}x.e'$]
            The last rule is \Rule{Abs\Aats}.
            We have $\forall i\in I.\ \Gamma', (x:s_i) \leqA^{e'} \Gamma, (x:s_i)$ (quite straightforward)
            so by applying the induction hypothesis we have $\forall i\in I.\ \tyof {e'} {\Gamma', (x:s_i)} \leq \tyof {e'} {\Gamma, (x:s_i)}$.

            \item[$e=e_1 e_2$] The last rule is \Rule{App\Aats}.
            We can conclude immediately by using the induction hypothesis and noticing that $\apply {} {}$ is monotonic for both of its arguments.

            \item[$e=\pi_i e'$] The last rule is \Rule{Proj\Aats}.
            We can conclude immediately by using the induction hypothesis and noticing that $\bpi_i$ is monotonic.

            \item[$e=(e_1,e_2)$] The last rule is \Rule{Pair\Aats}.
            We can conclude immediately by using the induction hypothesis.

            \item[$e=\ite {e_0} t {e_1} {e_2}$] The last rule is \Rule{Case\Aats}.
            By using the induction hypothesis we get $\Refine {e_0,t} {\Gamma'} \leqA^{e_0} \Refine {e_0,t} \Gamma$.
            We also have $\Gamma' \leqA^{e_1} \Gamma$ (as $e_1$ is a subexpression of $e$).

            From those two properties, let's show that we can deduce $\Refine {e_0,t} {\Gamma'} \leqA^{e_1} \Refine {e_0,t} \Gamma$:

            Let $e'\in\dom{\Refine {e_0,t} \Gamma}$ a subexpression of $e_1$.
            \begin{itemize}
              \item If $e'$ is also a subexpression of $e_0$, we can directly deduce\\
              $\tyof {e'} {\Refine {e_0,t} \Gamma'} \leq (\Refine {e_0,t} \Gamma) (e')$
              by using $\Refine {e_0,t} {\Gamma'} \leqA^{e_0} \Refine {e_0,t} \Gamma$.

              \item Otherwise, as $\Refine {e_0,t} {\_}$ is reductive,
              we have $\Refine {e_0,t} {\Gamma'} \leqA \Gamma'$ and thus by using the induction hypothesis
              $\tyof {e'} {\Refine {e_0,t} {\Gamma'}} \leq \tyof {e'} {\Gamma'}$.
              We also have $\tyof {e'} {\Gamma'} \leq \Gamma(e')$ by using $\Gamma' \leqA^{e_1} \Gamma$.
              We deduce $\tyof {e'} {\Refine {e_0,t} {\Gamma'}} \leq \Gamma(e') = (\Refine {e_0,t} \Gamma) (e')$.
            \end{itemize}

            So we have $\Refine {e_0,t} {\Gamma'} \leqA^{e_1} \Refine {e_0,t} \Gamma$.
            Consequently, we can apply the induction hypothesis again to get
            $\tyof {e_1} {\Refine {e_0,t} {\Gamma'}} \leq \tyof {e_1} {\Refine {e_0,t} {\Gamma}}$.

            We proceed the same way for the last premise.
          \end{description}\ \\

          Now, let's prove the second property.
          We perform a (nested) induction on $\varpi$.

          Recall that we have $\forall t_1,t_2,\ts_1,\ts_2.\ t_1 \leq t_2 \text{ and } \ts_1\leq\ts_2 \text{ and } \tsrep{\ts_1}\leq\tsrep{\ts_2} \Rightarrow \tsrep{t_1\tsand\ts_1} \leq \tsrep{t_2\tsand\ts_2}$ (lemma above).

          Thus, in order to prove
          $\tsrep {\constr {\varpi} {\Gamma',e,t} \tsand \tyof {\occ e \varpi} {\Gamma'}} \leq \tsrep {\constr {\varpi} {\Gamma,e,t} \tsand \tyof {\occ e \varpi} {\Gamma}}$,
          we can prove the following:
          \begin{align*}
            &\constr {\varpi} {\Gamma',e,t} \leq \constr {\varpi} {\Gamma,e,t}\\
            &\tyof {\occ e \varpi} {\Gamma'} \leq \tyof {\occ e \varpi} {\Gamma}\\
            &\tsrep{\tyof {\occ e \varpi} {\Gamma'}} \leq \tsrep{\tyof {\occ e \varpi} {\Gamma}}
          \end{align*}

          The two last inequalities can be proved
          with the outer induction hypothesis (for $\varpi=\epsilon$ we use the proof of the first property above).

          Thus we just have to prove that $\constr {\varpi} {\Gamma',e,t} \leq \constr {\varpi} {\Gamma,e,t}$.
          The only case that is interesting is the case $\varpi=\varpi'.1$.

          First, we can notice that the $\worra {} {}$ operator is monotonic for its second argument
          (consequence of its declarative definition).

          Secondly, let's show that for any function types $t_1 \leq t_2$, and for any type $t'$,
          we have $(\worra {t_1} {t'}) \land \dom {t_2} \leq \worra {t_2} {t'}$. By the absurd, let's suppose it is not true.
          Let's note $t'' = (\worra {t_1} {t'}) \land \dom {t_2}$.
          Then we have $t'' \leq \dom {t_2} \leq \dom {t_1}$ and $t_2 \leq \arrow {t''} {t'}$ and
          $t_1 \not\leq \arrow {t''} {t'}$, which contradicts $t_1 \leq t_2$.

          Let's note $t_1 = \tsrep {\tyof{\occ e{\varpi'.0}}{\Gamma'}}$ and $t_2 = \tsrep {\tyof{\occ e{\varpi'.0}}\Gamma}$
          and $t'=\env {\Gamma,e,t} (\varpi')$.
          As $e$ is well-typed, and using the inner induction hypothesis, we have $\tsrep {\tyof {\occ e {\varpi'.1}} {\Gamma'}} \leq \tsrep {\tyof {\occ e {\varpi'.1}} {\Gamma}} \leq \dom {t_2}$.\\
          Thus, using this property, we get:\\
          \begin{align*}
          &(\worra {t_1} {t'}) \land \tsrep {\tyof {\occ e {\varpi'.1}} {\Gamma'}}\\
          \leq &(\worra {t_2} {t'}) \land \tsrep {\tyof {\occ e {\varpi'.1}} {\Gamma}}
          \end{align*}

          Then, using the monotonicity of the second argument of $\worra {}{}$ and the outer induction hypothesis:
          \begin{align*}
            &(\worra {t_1} {\env {\Gamma',e,t} (\varpi')}) \land \tsrep {\tyof {\occ e {\varpi'.1}} {\Gamma'}}\\
            \leq &(\worra {t_2} {\env {\Gamma,e,t} (\varpi')}) \land \tsrep {\tyof {\occ e {\varpi'.1}} {\Gamma}}
          \end{align*}
          \\

          \

          Finally, we must prove the third property.\\
          It is straightforward by using the previous result and the induction hypothesis:\\
          $\forall e'$ s.t. $\exists \varpi.\ \occ e \varpi \equiv e'$, we get
          $\bigwedge_{\{\varpi\alt \occ e \varpi \equiv e'\}} \env {\Gamma',e,t} (\varpi) \leq \bigwedge_{\{\varpi\alt \occ e \varpi \equiv e'\}} \env {\Gamma,e,t} (\varpi)$.

          The rest follows.
          \end{proof}

          \begin{definition}[Positive derivation]
            A derivation of the declarative type system is said positive iff it does not contain any rule \Rule{Abs-}.
          \end{definition}

          \begin{theorem}[Completeness for positive derivations]\label{completenessAtsPositive}
            For every $\Gamma$, $e$, $t$ such that we have a positive derivation of $\Gamma \vdash e:t$,
            there exists a global parameter $n_o$ with which $\tsrep{\tyof e \Gamma} \leq t$.

            More precisely:
            \begin{align*}
              &\forall \Gamma, e, t.\ \Gamma \vdash e:t \text{ has a positive derivation } \Rightarrow \tsrep{\tyof e \Gamma} \leq t\\
              &\forall \Gamma, \Gamma', e, t.\ \Gamma \evdash e t \Gamma' \text{ has a positive derivation } \Rightarrow \Refine {e,t} \Gamma \leqA \Gamma' \text{ (for $n_o$ large enough)}
            \end{align*}
          \end{theorem}

          \begin{proof}
          We proceed by induction on the derivation.

          Let's prove the first property. We have a positive derivation of $\Gamma \vdash e:t$.

          If $\Gamma = \bot$, we can conclude directly using \Rule{Efq\Aats}. Thus, let's suppose $\Gamma \neq \bot$.

          If $e=x$ is a variable, then the derivation only uses \Rule{Env}, \Rule{Inter} and \Rule{Subs}.
          We can easily conclude just be using \Rule{Var\Aats}. Thus, let's suppose $e$ is not a variable.

          If $e\in\dom\Gamma$, we can have the rule \Rule{Env} applied to $e$ in our derivation, but in this case
          there can only be \Rule{Inter} and \Rule{Subs} after it (not \Rule{Abs-} as we have a positive derivation).
          Thus, our derivation contains a derivation of $\Gamma \vdash e:t'$ that does not use the rule \Rule{Env} on $e$
          and such that $t'\land \Gamma(e) \leq t$ (actually, it is possible for our derivation to typecheck $e$ only using the rule \Rule{Env}:
          in this case we can take $t'=\Any$ and use the fact that $\Gamma$ is well-formed).
          Hence, we can build a positive derivation for $\Gamma\setminus\{e\} \vdash e:t'$.
          By using the induction hypothesis we deduce that $\tsrep{\tyof e {\Gamma\setminus\{e\}}} \leq t'$.
          Thus, by looking at the rule \Rule{Env\Aats},
          we deduce $\tsrep{\tyof e \Gamma} \leq \Gamma(e) \land \tsrep{\tyof e {\Gamma\setminus\{e\}}} \leq t$.
          It concludes this case, so let's assume $e\not\in\dom\Gamma$.

          Now we analyze the last rule of the derivation:

          \begin{description}
            \item[\Rule{Env}] Impossible case ($e\not\in\dom\Gamma$).
            \item[\Rule{Inter}] By using the induction hypothesis we get $\tsrep{\tyof e \Gamma} \leq t_1$ and $\tsrep{\tyof e \Gamma} \leq t_2$.
            Thus, we have $\tsrep{\tyof e \Gamma} \leq t_1 \land t_2$.
            \item[\Rule{Subs}] Trivial using the induction hypothesis.
            \item[\Rule{Const}] We know that the derivation of $\tyof e \Gamma$ (if any) ends with the rule \Rule{Const\Aats}.
            Thus this case is trivial.
            \item[\Rule{App}] We know that the derivation of $\tyof e \Gamma$ (if any) ends with the rule \Rule{App\Aats}.
            Let $\ts_1 = \tyof {e_1} \Gamma$ and $\ts_2 = \tyof {e_2} \Gamma$.
            With the induction hypothesis we have $\tsrep {\ts_1} \leq \arrow {t_1} {t_2}$ and $\tsrep {\ts_2} \leq t_1$, with $t_2=t$.
            According to the descriptive definition of $\apply{}{}$, we have
            $\apply{\tsrep {\ts_1}}{\tsrep {\ts_2}} \leq \apply{\arrow {t_1}{t_2}}{t_1} \leq t_2$.
            As we also have $\tsrep{\apply {\ts_1} {\ts_2}} \leq \apply{\tsrep {\ts_1}}{\tsrep {\ts_2}}$,
            we can conclude that $\tyof e \Gamma \leq t_2=t$.

            \item[\Rule{Abs+}] We know that the derivation of $\tyof e \Gamma$ (if any) ends with the rule \Rule{Abs\Aats}.
            This case is straightforward using the induction hypothesis.
            \item[\Rule{Abs-}] This case is impossible (the derivation is positive).
            \item[\Rule{Case}] We know that the derivation of $\tyof e \Gamma$ (if any) ends with the rule \Rule{Case\Aats}.
            By using the induction hypothesis and the monotonicity lemma, we get $\tsrep{\ts_1}\leq t$ and $\tsrep{\ts_2}\leq t$.
            So we have $\tsrep{\ts_1\tsor\ts_2}=\tsrep{\ts1}\vee\tsrep{\ts2}\leq t$.
            \item[\Rule{Proj}] Quite similar to the case \Rule{App}.
            \item[\Rule{Pair}] We know that the derivation of $\tyof e \Gamma$ (if any) ends with the rule \Rule{Pair\Aats}.
            We just use the induction hypothesis and the fact that $\tsrep{\ts_1\tstimes\ts_2}=\pair {\tsrep{\ts1}} {\tsrep{\ts2}}$.
          \end{description}

          \

          Now, let's prove the second property. We have a positive derivation of $\Gamma \evdash e t \Gamma'$.

          \begin{description}
            \item[\Rule{Base}] Any value of $n_o$ will give $\Refine {e,t} \Gamma \leqA \Gamma$, even $n_o = 0$.
            \item[\Rule{Path}] We have $\Gamma' = \Gamma_1,(\occ e \varpi:t')$.
            By applying the induction hypothesis on the premise $\Gamma \evdash e t \Gamma_1$, we have
            $\RefineStep {e,t}^n (\Gamma) = \Gamma_2$ with $\Gamma_2 \leqA \Gamma_1$ for a certain $n$.

            We now proceed by induction on the derivation $\pvdash {\Gamma_1} e t \varpi:t'$
            to show that we can obtain $\env {\Gamma'',e,t} (\varpi) \leq t'$ with $\Gamma'' = \RefineStep {e,t}^{n'} (\Gamma_2)$
            for a certain $n'$. It is then easy to conclude by taking $n_o = n+n'$.

            \begin{description}
              \item[\Rule{PSubs}] Trivial using the induction hypothesis.
              \item[\Rule{PInter}] By using the induction hypothesis we get:
              \begin{align*}
                &\env {\Gamma_1'',e,t} (\varpi) \leq t_1\\
                &\env {\Gamma_2'',e,t} (\varpi) \leq t_2\\
                &\RefineStep {e,t}^{n_1} (\Gamma_1) \leqA \Gamma_1''\\
                &\RefineStep {e,t}^{n_2} (\Gamma_2) \leqA \Gamma_2''
              \end{align*}

              By taking $n'=\max (n_1,n_2)$,
              we can have $\Gamma'' = \RefineStep {e,t}^{n'} (\Gamma_2)$ with $\Gamma'' \leqA \Gamma_1''$ and $\Gamma'' \leqA \Gamma_2''$.
              Thus, by using the monotonicity lemma, we can obtain $\env {\Gamma'',e,t} (\varpi) \leq t_1 \land t_2 = t'$.
              \item[\Rule{PTypeof}] By using the outer induction hypothesis we get
              $\tsrep{\tyof {\occ e \varpi} {\Gamma_2}} \leq t'$.
              Moreover we have $\env {\Gamma_2,e,t} (\varpi) \leq \tsrep{\tyof {\occ e \varpi} {\Gamma_2}}$
              (by definition of $\env {}$), thus we can conclude directly.
              \item[\Rule{PEps}] Trivial.
              \item[\Rule{PAppR}] By using the induction hypothesis we get:
              \begin{align*}
                &\env {\Gamma_1'',e,t} (\varpi.0) \leq \arrow {t_1} {t_2}\\
                &\env {\Gamma_2'',e,t} (\varpi) \leq t_2'\\
                & t_2\land t_2' \simeq \Empty\\
                &\RefineStep {e,t}^{n_1} (\Gamma_1) \leqA \Gamma_1''\\
                &\RefineStep {e,t}^{n_2} (\Gamma_2) \leqA \Gamma_2''
              \end{align*}

              By taking $n'=\max (n_1,n_2) + 1$,
              we can have $\Gamma'' = \RefineStep {e,t}^{n'} (\Gamma_2)$ with $\Gamma'' \leqA \RefineStep {e,t} (\Gamma_1'')$
              and $\Gamma'' \leqA \RefineStep {e,t} (\Gamma_2'')$.

              In consequence, we have $\tsrep{\tyof {\occ e {\varpi.0}} {\Gamma''}} \leq \env {\Gamma_1'',e,t} (\varpi.0) \leq \arrow {t_1} {t_2}$
              (by definition of $\RefineStep {e,t}$).
              We also have, by monotonicity, $\env {\Gamma'',e,t} (\varpi) \leq t_2'$.

              As $t_2\land t_2' \simeq \Empty$, we have:
              \begin{align*}
                &\apply {(\arrow {t_1} {t_2})} {(\dom{\arrow {t_1} {t_2}}\setminus(\neg t_1))}\\
                &\simeq \apply {(\arrow {t_1} {t_2})} {t_1} \simeq t_2 \leq \neg t_2'
              \end{align*}

              Thus, by using the declarative definition of $\worra {} {}$, we know that
              $\worra {(\arrow {t_1} {t_2})} {t_2'} \leq \neg t_1$.

              According to the properties on $\worra {} {}$ that we have proved in the proof of the monotonicity lemma,
              we can deduce:
              \begin{align*}
              &t_1 \land \worra {\tsrep{\tyof {\occ e {\varpi.0}} {\Gamma''}}} {\env {\Gamma'',e,t} (\varpi)}\\
              &\leq t_1 \land \worra {(\arrow {t_1} {t_2})} {t_2'} \leq t_1 \land \neg t_1 \simeq \Empty
              \end{align*}
              And thus $\worra {\tsrep{\tyof {\occ e {\varpi.0}} {\Gamma''}}} {\env {\Gamma'',e,t} (\varpi)} \leq \neg t_1$.

              It concludes this case.

              \item[\Rule{PAppL}] By using the induction hypothesis we get:
              \begin{align*}
                &\env {\Gamma_1'',e,t} (\varpi.1) \leq t_1\\
                &\env {\Gamma_2'',e,t} (\varpi) \leq t_2\\
                &\RefineStep {e,t}^{n_1} (\Gamma_1) \leqA \Gamma_1''\\
                &\RefineStep {e,t}^{n_2} (\Gamma_2) \leqA \Gamma_2''
              \end{align*}

              By taking $n'=\max (n_1,n_2)$,
              we can have $\Gamma'' = \RefineStep {e,t}^{n'} (\Gamma_2)$ with $\Gamma'' \leqA \Gamma_1''$ and $\Gamma'' \leqA \Gamma_2''$.
              Thus, by using the monotonicity lemma, we can obtain $\env {\Gamma'',e,t} (\varpi.0) \leq \neg (\arrow {t_1} {\neg t_2}) = t'$.

              \item[\Rule{PPairL}] Quite straightforward using the induction hypothesis and the descriptive definition of $\bpi_1$.
              \item[\Rule{PPairR}] Quite straightforward using the induction hypothesis and the descriptive definition of $\bpi_2$.
              \item[\Rule{PFst}] Trivial using the induction hypothesis.
              \item[\Rule{PSnd}] Trivial using the induction hypothesis.
            \end{description}
          \end{description}
          \end{proof}

          From this result, we will now prove a stronger but more complex completeness theorem.
          We were not able to prove full completeness, just a partial form of it. Indeed,
          the use of nested \Rule{PAppL} yields a precision that the algorithm loses by applying \tsrep{}
          in the definition of \constrf{}. Completeness is recovered by forbidding nested negated arrows on the
          left-hand side of negated arrows.

          \begin{definition}[Rank-0 negated derivation]
            A derivation of the declarative type system is said rank-0 negated iff any application of \Rule{PAppL}
            has a positive derivation as first premise ($\pvdash \Gamma e t \varpi.1:t_1$).
          \end{definition}

          \noindent The use of this terminology is borrowed from the ranking of higher-order
          types, since, intuitively, it corresponds to typing a language in
          which  in the types used in dynamic tests, a negated arrow never occurs  on the
          left-hand side of another negated arrow.

          \begin{lemma}
            If $e$ is an application, then $\tyof e \Gamma$ does not contain any constructor $\tsfun \cdots$.
            Consequently, we have $\tsrep{\tyof e \Gamma} \simeq \tyof e \Gamma$.
          \end{lemma}

          \begin{proof}
            By case analysis: neither \Rule{Efq\Aats}, \Rule{Env\Aats} nor \Rule{App\Aats} can produce a type
            containing a constructor $\tsfun \cdots$.
          \end{proof}

          \begin{theorem}[Completeness for rank-0 negated derivations]
            For every $\Gamma$, $e$, $t$ such that we have a rank-0 negated derivation of $\Gamma \vdash e:t$, there exists a global parameter $n_o$
            with which $\tyof e \Gamma \leq t$.

            More precisely:
            \begin{align*}
              &\forall \Gamma, e, t.\ \Gamma \vdash e:t \text{ has a rank-0 negated derivation } \Rightarrow \tyof e \Gamma \leq t\\
              &\forall \Gamma, \Gamma', e, t.\ \Gamma \evdash e t \Gamma' \text{ has a rank-0 negated derivation } \Rightarrow \Refine {e,t} \Gamma \leqA \Gamma' \text{ (for $n_o$ large enough)}
            \end{align*}
          \end{theorem}

          \begin{proof}
            This proof is done by induction. It is quite similar to that of the completeness for positive derivations.
            In consequence, we will only detail cases that are quite different from those of the previous proof.

            Let's begin with the first property. We have a rank-0 negated derivation of $\Gamma \vdash e:t$.
            We want to show $\tyof e \Gamma \leq t$ (note that this is weaker than showing $\tsrep {\tyof e \Gamma} \leq t$).

            As in the previous proof, we can suppose that $\Gamma \neq \bot$ and that $e$ is not a variable.

            The case $e\in\dom\Gamma$ is also very similar, but there is an additional case to consider:
            the rule \Rule{Abs-} could possibly be used after a rule \Rule{Env} applied on $e$.
            However, this case can easily be eliminated by changing the premise of this \Rule{Abs-} with another one
            that does not use the rule \Rule{Env} on $e$ (the type of the premise does not matter for the rule \Rule{Abs-},
            even $\Any$ suffices). Thus let's assume $e\not\in\dom\Gamma$.

            Now we analyze the last rule of the derivation (only the cases that are not similar are shown):
            \begin{description}
              \item[\Rule{Abs-}] We know that the derivation of $\tyof e \Gamma$ (if any) ends with the rule \Rule{Abs\Aats}.
              Moreover, by using the induction hypothesis on the premise, we know that $\tyof e \Gamma \neq \tsempty$.
              Thus we have $\tyof e \Gamma \leq \neg (\arrow {t_1} {t_2}) = t$ (because every type $\neg (\arrow {s'} {t'})$
              such that $\neg (\arrow {s'} {t'}) \land \bigwedge_{i\in I} \arrow {s_i} {t_i} \neq \Empty$ is in $\tsint{\tsfun{\arrow {s_i} {t_i}}}$).
            \end{description}

            \

            Now let's prove the second property. We have a rank-0 negated derivation of $\Gamma \evdash e t \Gamma'$.

            \begin{description}
              \item[\Rule{Base}] Any value of $n_o$ will give $\Refine {e,t} \Gamma \leqA \Gamma$, even $n_o = 0$.
              \item[\Rule{Path}] We have $\Gamma' = \Gamma_1,(\occ e \varpi:t')$.

              As in the previous proof of completeness,
              by applying the induction hypothesis on the premise $\Gamma \evdash e t \Gamma_1$, we have
              $\RefineStep {e,t}^n (\Gamma) = \Gamma_2$ with $\Gamma_2 \leqA \Gamma_1$ for a certain $n$.

              However, this time, we can't prove $\env {\Gamma'',e,t} (\varpi) \leq t'$ with $\Gamma'' = \RefineStep {e,t}^{n'} (\Gamma_2)$
              for a certain $n'$: the induction hypothesis is weaker than in the previous proof
              (we don't have $\tsrep {\tyof e \Gamma} \leq t$ but only $\tyof e \Gamma \leq t$).

              Instead, we will prove by induction on the derivation $\pvdash {\Gamma_1} e t \varpi:t'$ that
              $\env {\Gamma'',e,t} (\varpi) \tsand \tyof {\occ e \varpi} {\Gamma''} \leq t'$.
              It suffices to conclude in the same way as in the previous proof:
              by taking $n_o = n+n'$,
              it ensures that our final environment $\Gamma_{n_o}$ verifies $\tyof {\occ e \varpi} \Gamma_{n_o} \leq t'$
              and thus we have $\Gamma_{n_o} \leq \Gamma'$
              (given that $\tsrep{\Empty} = \Empty$, we also easily verify that if $\Gamma' = \bot \Rightarrow \Gamma_{n_o}=\bot$).

              \begin{description}
                \item[\Rule{PSubs}] Trivial using the induction hypothesis.
                \item[\Rule{PInter}] Quite similar to the previous proof (the induction hypothesis is weaker, but it works the same way).
                \item[\Rule{PTypeof}] By using the outer induction hypothesis we get $\tyof {\occ e \varpi} {\Gamma_2} \leq t'$ so it is trivial.
                \item[\Rule{PEps}] Trivial.
                \item[\Rule{PAppR}] By using the induction hypothesis, we get:
                \begin{align*}
                  &\env {\Gamma_1'',e,t} (\varpi.0) \tsand \tyof {\occ e {\varpi.0}} {\Gamma_1''} \leq \arrow {t_1} {t_2}\\
                  &\env {\Gamma_2'',e,t} (\varpi) \tsand \tyof {\occ e {\varpi}} {\Gamma_2''} \leq t_2'\\
                  &t_2 \land t_2' \simeq \Empty\\
                  &\RefineStep {e,t}^{n_1} (\Gamma_1) \leqA \Gamma_1''\\
                  &\RefineStep {e,t}^{n_2} (\Gamma_2) \leqA \Gamma_2''
                \end{align*}

                Moreover, as $\occ e {\varpi}$ is an application, we can use the lemma above to deduce
                $\env {\Gamma_2'',e,t} (\varpi) \tsand \tyof {\occ e {\varpi}} {\Gamma_2''} = \env {\Gamma_2'',e,t} (\varpi)$
                (see definition of $\env {}$).

                Thus we have $\env {\Gamma_2'',e,t} (\varpi) \leq t_2'$.\\
                We also have $\env {\Gamma_1'',e,t} (\varpi.0) \leq \tsrep{\env {\Gamma_1'',e,t} (\varpi.0) \tsand \tyof {\occ e {\varpi.0}} {\Gamma_1''}} \leq \arrow {t_1} {t_2}$.

                Now we can conclude exactly as in the previous proof (by taking $n'=\max (n_1,n_2)$).

                \item[\Rule{PAppL}] We know that the left premise is a positive derivation.
                Thus, using the previous completeness theorem, we get:
                \begin{align*}
                  &\env {\Gamma_1'',e,t} (\varpi.1) \leq t_1\\
                  &\RefineStep {e,t}^{n_1} (\Gamma_1) \leqA \Gamma_1''
                \end{align*}

                By using the induction hypothesis, we also get:
                \begin{align*}
                  &\env {\Gamma_2'',e,t} (\varpi) \tsand \tyof {\occ e {\varpi}} {\Gamma_2''} \leq t_2\\
                  &\RefineStep {e,t}^{n_2} (\Gamma_2) \leqA \Gamma_2''
                \end{align*}

                Moreover, as $\occ e {\varpi}$ is an application, we can use the lemma above to deduce
                $\env {\Gamma_2'',e,t} (\varpi) \tsand \tyof {\occ e {\varpi}} {\Gamma_2''} = \env {\Gamma_2'',e,t} (\varpi)$
                (see definition of $\env {}$).

                Thus we have $\env {\Gamma_2'',e,t} (\varpi) \leq t_2$.

                Now we can conclude exactly as in the previous proof (by taking $n'=\max (n_1,n_2)$).

                \item[\Rule{PPairL}] Quite straightforward using the induction hypothesis and the descriptive definition of $\bpi_1$.
                \item[\Rule{PPairR}] Quite straightforward using the induction hypothesis and the descriptive definition of $\bpi_2$.
                \item[\Rule{PFst}] Quite straightforward using the induction hypothesis.
                \item[\Rule{PSnd}] Quite straightforward using the induction hypothesis.
              \end{description}
            \end{description}
          \end{proof}

    \subsection{Proofs for the algorithmic type system without type
      schemes}
    \label{sec:proofs_algorithmic_without_ts}

    In this section, we consider the algorithmic type system without type schemes, as defined in \ref{sec:algorules}.

    \subsubsection{Soundness}

    \begin{lemma}
      \label{soundness_simple_ts}
    For every $\Gamma$, $e$, $t$, $n_o$, if $\Gamma\vdashA e: t$, then there exists $\ts \leq t$ such that $\Gamma \vdashAts e: \ts$.
    \end{lemma}

    \begin{proof}
      Straightforward induction over the structure of $e$.
    \end{proof}

    \begin{theorem}[Soundness of the algorithmic type system without type schemes]\label{soundnessA}
    For every $\Gamma$, $e$, $t$, $n_o$, if $\Gamma\vdashA e: t$, then $\Gamma \vdash e:t$.
    \end{theorem}

    \begin{proof}
      Trivial by using the theorem \ref{soundnessAts} and the previous lemma.
    \end{proof}

    \subsubsection{Completeness}

    \[
      \begin{array}{lrcl}
        \textbf{Simple type} & t_{s} & ::= & b \alt \pair {t_{s}} {t_{s}} \alt t_{s} \vee t_{s} \alt \neg t_{s} \alt \Empty \alt \arrow \Empty \Any\\
        \textbf{Positive type} & t_+ & ::= & t_{s} \alt t_+ \vee t_+ \alt t_+ \land t_+ \alt \arrow {t_+} {t_+} \alt \arrow {t_+} {\neg t_+}\\
        \textbf{Positive abstraction type} & t^\lambda_+ & ::= & \arrow {t_+} {t_+} \alt \arrow {t_+} {\neg t_+} \alt t^\lambda_+ \land t^\lambda_+\\
        \textbf{Positive expression} & e_+ & ::= & c\alt x\alt e_+ e_+\alt\lambda^{t^\lambda_+} x.e_+\alt \pi_j e_+\alt(e_+,e_+)\alt\tcase{e_+}{t_{s}}{e_+}{e_+}
      \end{array}
    \]

    \begin{lemma}
      If we restrict the language to positive expressions $e_+$,
      then we have the following property:

      $\forall \Gamma, e_+, \ts.\ \Gamma \vdashAts e_+:\ts \Rightarrow \Gamma \vdashA e_+:\tsrep{\ts}$
    \end{lemma}

    \begin{proof}
      We can prove it by induction over the structure of $e_+$.

      The main idea of this proof is that, as $e_+$ is a positive expression, the rule \Rule{Abs-} is not needed anymore
      because the negative part of functional types (i.e. the $N_i$ part of their DNF) becomes useless:

      \begin{itemize}
        \item When typing an application $e_1 e_2$, the negative part of the type of $e_1$
        is ignored by the operator $\apply {} {}$.
        \item Moreover, as there is no negated arrows in the domain of lambda-abstractions,
        the negative arrows of the type of $e_2$ can also be ignored.
        \item Similarly, negative arrows can be ignored when refining an application ($\worra {} {}$ also ignore the negative part
        of the type of $e_1$).
        \item Finally, as the only functional type that we can test is $\arrow \Empty \Any$, a functional type
        cannot be refined to $\Empty$ due to its negative part, and thus we can ignore its negative part
        (it makes no difference relatively to the rule \Rule{Efq\Aats}).
      \end{itemize}
    \end{proof}

    \begin{theorem}[Completeness of the algorithmic type system for positive expressions]\label{completenessA}
      For every type environment $\Gamma$ and positive expression $e_+$, if
      $\Gamma\vdash e_+: t$, then there exist $n_o$ and  $t'$ such that $\Gamma\vdashA
      e_+: t'$.
    \end{theorem}

    \begin{proof}
      Trivial by using the theorem \ref{completenessAtsPositive} and the previous lemma.
    \end{proof}



\iflongversion
\else
\section{Record types operators}\label{app:recop}
\input{record_operations}
\newpage

\section{A more precise rule for inference}\label{app:optimize}
In our prototype we have implemented for the inference of arrow type the following rule:
\[\begin{array}{l}
   {\small\Rule{AbsInf++}}\\
  \frac
  {  \begin{array}{c}
    T = \{ (s\setminus\bigvee_{s'\in\psi(x)}s',t) \} \cup \{ (s',t') ~|~
      s'\in\psi(x) \land \Gamma,x:s'\vdash e:t' \}\\
\textstyle\Gamma,x:s\vdash e\triangleright\psi
    \qquad
    \Gamma,x:s\setminus\bigvee_{s'\in\psi(x)}s'\vdash e:t
    \end{array}
    }
    {
    \textstyle\Gamma\vdash\lambda x{:}s.e:{\bigwedge_{(s',t') \in T}s'\to t'}
    }
\end{array}\]
instead of the simpler \Rule{AbsInf+} given in
Section \ref{sec:refining}. The difference of this new rule with
respect to \Rule{AbsInf+} is that the typing of the body
is made under the hypothesis $x:s\setminus\bigvee_{s'\in\psi(x)}s'$,
that is, the domain of the function minus all the input types
determined by the $\psi$-analysis. This yields an even better refinement
of the function type that makes a difference for instance with the
inference for the function \code{xor\_} (see Code 3
in Table~\ref{tab:implem}): the old rule would have returned a less precise type. The rule above is defined for functions annotated by a single arrow type:
the extension to annotations with intersections of multiple arrows is similar to the one we did in the
simpler setting of Section~\ref{sec:refining}.

\newpage
\section{A Roadmap to Polymorphic Types}
\label{app:roadmap}
\input{roadmappolymorphism}

\fi
\end{document}